\documentclass[letterpaper, 11pt]{article}
\usepackage{graphicx}
\usepackage{amsmath}
\usepackage{amsfonts}
\usepackage{hyperref}
\usepackage{amsthm}
\usepackage{amscd}
\usepackage{mathrsfs}
\usepackage{caption}
\usepackage{subcaption}
\usepackage{float}
\usepackage[top=1in,left=1in,right=1in,bottom=1in]{geometry}

\newtheorem{theorem}{Theorem}[section]
      
      \newtheorem{proposition}[theorem]{Proposition}
      \newtheorem{corollary}[theorem]{Corollary}
      
      \newtheorem{example}[theorem]{Example}
      \newtheorem{remark}[theorem]{Remark}
    \counterwithin{figure}{section}

\hypersetup{
	colorlinks = true,
	linkcolor  = red,
	citecolor  = green,
	filecolor  = cyan,
	urlcolor   = magenta,}

\numberwithin{equation}{section}

\title{Factorization for  the full-line matrix Schr\"odinger equation
and a unitary transformation to the half-line scattering\thanks{Research partially supported by project PAPIIT-DGAPA UNAM 
IN100321} }
\author{Tuncay Aktosun\\
Department of Mathematics\\
University of Texas at Arlington\\
Arlington, TX 76019-0408, USA\\
\\
Ricardo Weder\thanks{Emeritus Fellow, 
Sistema Nacional de Investigadores}\\ 
Departamento de F\'\i sica Matem\'atica\\
Instituto de Investigaciones en
Matem\'aticas Aplicadas y en Sistemas\\
Universidad Nacional Aut\'onoma de M\'exico\\
Apartado Postal 20-126, IIMAS-UNAM\\
 Ciudad de M\'exico  01000, M\'exico}

\date{}

\begin{document}

\maketitle

\centerline{\it Dedicated to Vladimir Aleksandrovich Marchenko  to commemorate his 100th birthday}

\begin{abstract}
The scattering matrix for the full-line matrix Schr\"odinger equation is analyzed when the corresponding matrix-valued
potential is selfadjoint, integrable, and has a finite first moment.
The matrix-valued potential is decomposed into a finite number of fragments, and a factorization formula is presented 
expressing the matrix-valued scattering coefficients in terms of the matrix-valued
scattering coefficients for the fragments. Using the factorization formula,
some explicit examples are provided illustrating that in general the left and 
right matrix-valued transmission coefficients 
are unequal.  A unitary transformation is established between the full-line matrix Schr\"odinger operator and
the half-line matrix Schr\"odinger operator with a particular selfadjoint
boundary condition and by relating the full-line and half-line potentials appropriately. 
Using that unitary transformation, the relations are established between the full-line and the half-line quantities
such as the Jost solutions, the physical solutions, and the scattering matrices. Exploiting the connection between the 
corresponding full-line and half-line scattering matrices, Levinson's theorem on the full line is proved
and is related to
Levinson's theorem on the half line.

\end{abstract}



\newpage

\section{Introduction}
\label{section1}

In this paper we consider certain aspects
of the matrix-valued Schr\"odinger equation on the full line
\begin{equation}
\label{1.1}
-\psi''+V(x)\,\psi=k^2\psi,\qquad x\in\mathbb R,
\end{equation}
where $x$ represents the spacial coordinate, $\mathbb R:=(-\infty,+\infty),$ 
the prime denotes the $x$-derivative, the wavefunction $\psi$
may be an $n\times n$ matrix or a column vector with
$n$ components. Here, $n$ can be chosen as any fixed positive integer, including the special value $n=1$ which corresponds to the scalar case.
The potential $V$ is assumed to be an $n\times n$ matrix-valued function
of $x$ satisfying the selfadjointness
\begin{equation}
\label{1.2}
V(x)^\dagger=V(x),\qquad x\in\mathbb R,
\end{equation}
with the dagger denoting the matrix adjoint
(matrix transpose and complex conjugation),
and also belonging to the Faddeev class, i.e. satisfying the condition
\begin{equation}
\label{1.3}
\int_{-\infty}^\infty dx\,(1+|x|)\,|V(x)|<+\infty,
\end{equation}
with $|V(x)|$ denoting the
operator norm of the matrix $V(x).$
Since all matrix norms are equivalent for $n\times n$ matrices,
any other matrix norm can be used in \eqref{1.3}.
We use the conventions and notations
from \cite{AKV2001} and refer the reader to
that reference for further details.

Let us decompose the potential $V$ into two pieces 
$V_1$ and $V_2$ as
\begin{equation}
\label{1.4}
V(x)=V_1(x)+V_2(x),\qquad x\in\mathbb R,
\end{equation}
where we have defined
\begin{equation}
\label{1.5}
V_1(x):=
\begin{cases}V(x),\qquad x < 0,\\
\noalign{\medskip}
0,\qquad x>0,
\end{cases}
\quad
V_2(x):=
\begin{cases}0,\qquad x < 0,\\
\noalign{\medskip}
V(x),\qquad x>0.
\end{cases}
\end{equation}
We refer to $V_1$ and $V_2$ as the
left and right fragments of $V,$ respectively. We are interested in relating the $n\times n$ matrix-valued scattering 
coefficients  corresponding to $V$ to the $n\times n$ matrix-valued scattering coefficients corresponding to $V_1$ and $V_2,$ respectively.
This is done in Theorem~\ref{theorem3.3} by presenting 
a factorization formula
in terms of the transition matrix $\Lambda(k)$ defined in \eqref{2.20} 
and an equivalent factorization formula 
 in terms of $\Sigma(k)$ defined in \eqref{2.21}. 
 In fact, in Theorem~\ref{theorem3.6} the scattering coefficients for
 $V$ themselves are expressed in terms of the scattering coefficients for the fragments $V_1$ and $V_2.$

 The factorization
result of Theorem~\ref{theorem3.3} corresponds to the case where the potential $V$ is fragmented into two pieces at 
the fragmentation point $x=0$ as in \eqref{1.5}. 
In Theorem~\ref{theorem3.4} the factorization result of Theorem~\ref{theorem3.3} is generalized
by showing that the fragmentation point can be chosen arbitrarily. In Corollary~\ref{corollary3.5} the factorization formula
is further generalized to the case where the matrix potential is arbitrarily decomposed into
any finite number of fragments and by expressing the transition matrices $\Lambda(k)$ and
$\Sigma(k)$ in terms of the respective transition matrices corresponding to the fragments. 

Since the potential fragments are either supported on a half line or
compactly supported,   
the corresponding reflection coefficients have meromorphic extensions in $k$ from the real axis to 
the upper-half or lower-half complex plane or to the whole complex plane, respectively. 
Thus, it is more efficient to deal with the scattering coefficients of the fragments
than the scattering coefficients of the whole potential.
Furthermore, it is easier to determine 
the scattering coefficients when the corresponding potential is compactly supported or supported on a half line. 
We refer the reader to \cite{A1992} for a proof of the factorization formula for the full-line scalar Schr\"odinger equation
and to \cite{AKV1996}
for a generalization of
that factorization formula.
A composition rule has been presented in \cite{KS1999}
to express the factorization of the 
scattering matrix of a quantum graph in terms of the scattering matrices of its subgraphs.

The factorization formulas are useful in the analysis of direct and inverse scattering problems
because they help us to understand the scattering from the whole potential in terms of the scattering from 
the fragments of that potential.
We recall that the direct scattering problem on the half line consists of the
determination of the scattering matrix and the bound-state information
when the potential and the boundary condition are known.
The goal in the inverse scattering problem on the half line
is to recover the potential and the boundary condition when the
scattering matrix and the bound-state information are available.
The direct and inverse scattering problems on the full line are similar
to those on the half line except for the absence of a boundary condition.
For the direct and inverse scattering theory for the half-line matrix Schr\"odinger equation,
we refer the reader to the seminal monograph \cite{AM1963} of
Agranovich and Marchenko
when the
Dirichlet boundary condition is used 
and to our recent monograph \cite{AW2021} when the
general selfadjoint boundary condition is used.

The factorization formulas yield an efficient method to determine the scattering coefficients for the whole potential by first
determining the scattering coefficients for the potential fragments. For example, in Section~\ref{section4} we provide
some explicit examples to illustrate that the matrix-valued transmission coefficients from the left and from the right
are not necessarily equal even though the equality holds in the scalar case. In our examples, we determine
the left and right transmission coefficients explicitly with the help of the factorization result of Theorem~\ref{theorem3.6}.
Since the resulting explicit expressions for those transmission coefficients are extremely lengthy, we use the symbolic software Mathematica
on the first author's personal computer in order
to obtain those lengthy expressions.
Even though those transmission coefficients could be determined without using the factorization result, 
it becomes difficult or impossible to determine them directly
and demonstrate their unequivalence by using 
Mathematica on the same personal computer. 

In Section~2.4 of \cite{AW2021} we have presented
a unitary transformation between the half-line $2n\times 2n$ matrix Schr\"odinger operator with
a specific selfadjoint boundary condition and the full-line $n \times n$ matrix Schr\"odinger operator with a point interaction at $x=0.$ 
Using that unitary transformation and by introducing a full-line physical solution with a 
point interaction, in \cite{W2022} the relation between the half-line and full-line physical solutions and the relation between
the half-line and full-line scattering matrices have been found. In our current paper, we elaborate on such relations
in the absence of a point interaction on the full line, and we show how the half-line physical solution and the standard full-line physical solution
are related to each other and also show how the half-line and full-line scattering matrices are related to each other.
We also show how some other relevant half-line and full-line quantities are related to each other.
For example, in Theorem~\ref{theorem5.2} we establish the relationship between 
the determinant of the half-line Jost matrix and the determinant of the full-line transmission coefficient, and
in Theorem~\ref{theorem5.3} we provide the relationship between 
the determinant of the half-line scattering matrix and the determinant of the full-line transmission coefficient.
Those results help us establish in Section~\ref{section6}
Levinson's theorem for the full-line $n\times n$ matrix
Schr\"odinger operator and compare it with Levinson's theorem for the
corresponding half-line $2n\times 2n$ matrix
Schr\"odinger operator. 

We have the following remark on the notation we use. There are many equations in our paper
of the form
\begin{equation}\label{1.6}
a(k)=b(k),\qquad k\in\mathbb R\setminus\{0\},
\end{equation}
where $a(k)$ and $b(k)$ are continuous in $k\in\mathbb R\setminus\{0\},$ 
the quantity $a(k)$ is also continuous at $k=0,$ but $b(k)$ is not necessarily well defined at $k=0.$
We write \eqref{1.6} as
\begin{equation}
\label{1.7}
a(k)=b(k),\qquad k\in\mathbb R,
\end{equation}
with the understanding that we interpret $a(0)=b(0)$ in the sense that
by the continuity of $a(k)$ at $k=0,$ the limit of $b(k)$ at $k=0$ exists and we have
\begin{equation}
\label{1.8}
a(0)=\displaystyle\lim_{k\to 0}b(k).
\end{equation}

Our paper is organized as follows. In Section~\ref{section2} we provide 
the relevant results related to the scattering problem for \eqref{1.1}, and this is done
by presenting the Jost solutions, the physical solutions, the scattering coefficients, the scattering matrix
for \eqref{1.1}, and the relevant properties of those quantities.
In Section~\ref{section3} we establish our factorization formula by relating the scattering coefficients for
the full-line potential $V$ to the scattering coefficients for the fragments of the potential.
We also provide an alternate version of the factorization formula.
In Section~\ref{section4} we elaborate on the relation between the matrix-valued left and right transmission coefficients,
and we provide some explicit examples to illustrate that they are in general not equal to each other. In Section~\ref{section5},
we elaborate on the unitary transformation connecting the half-line and full-line matrix
Schr\"odinger operators, and we establish the connections between
the half-line and full-line scattering matrices, the half-line and full-line Jost solutions,
the half-line and full-line physical solutions, the half-line Jost matrix and the full-line transmission coefficients,
and the half-line and full-line zero-energy solutions that are bounded.
Finally, in Section~\ref{section6} we present
Levinson's theorem for the full-line matrix Schr\"odinger operator and compare it
with Levinson's theorem for the half-line matrix Schr\"odinger operator 
with a selfadjoint boundary condition.

\section{The full-line matrix Schr\"odinger equation}
\label{section2}

In this section we provide a summary of the results relevant to
the scattering problem for the full-line matrix Schr\"odinger equation
\eqref{1.1}. In particular, for \eqref{1.1} we introduce the pair of Jost solutions 
$f_{\rm{l}}(k,x)$ and $f_{\rm{r}}(k,x);$ the pair of physical solutions
$\Psi_{\rm{l}}(k,x)$ and $\Psi_{\rm{r}}(k,x);$ the four $n\times n$ matrix-valued
scattering coefficients
$T_{\rm{l}}(k),$ $T_{\rm{r}}(k),$ $L(k),$ and $R(k);$
the $2n\times 2n$ scattering matrix $S(k);$
the three relevant $2n\times 2n$ matrices
$F_{\rm{l}}(k,x),$ $F_{\rm{r}}(k,x),$ and $G(k,x);$ and
the pair of $2n\times 2n$ transition matrices $\Lambda(k)$
and $\Sigma(k).$ For the preliminaries needed in this section, we refer the reader to 
\cite{AKV2001}. For some earlier results on the full-line matrix Schr\"odinger equation, the reader can
consult \cite{LW2000,MO1982,O1985}.

When the potential $V$ in \eqref{1.1} satisfies \eqref{1.2} and \eqref{1.3}, 
there are two particular $n\times n$ matrix-valued solutions
to \eqref{1.1}, known as the left and right Jost solutions and denoted by
$f_{\rm{l}}(k,x)$ and $f_{\rm{r}}(k,x),$ respectively, satisfying the respective spacial asymptotics
\begin{equation}
\label{2.1}
f_{\rm{l}}(k,x)=e^{ikx}\left[I+o(1)\right],\quad f'_{\rm{l}}(k,x)=ik\,e^{ikx}\left[I+o(1)\right],\qquad x\to+\infty,
\end{equation}
\begin{equation}
\label{2.2}
f_{\rm{r}}(k,x)=e^{-ikx}\left[I+o(1)\right],\quad f'_{\rm{r}}(k,x)=-ik\,e^{ikx}\left[I+o(1)\right],\qquad x\to-\infty.
\end{equation}
For each $x\in
\mathbb R,$ the Jost solutions have analytic extensions  in $k$ from the real axis
$\mathbb R$ of the complex plane $\mathbb C$ to the
upper-half complex plane $\mathbb C^+$
and  they are continuous in $k\in\overline{\mathbb C^+},$ where we have defined
$\overline{\mathbb C^+}:=\mathbb C^+\cup\mathbb R.$ 
As listed in (2.1)--(2.3) of \cite{AKV2001}, 
we have  the integral representations for $f_{\rm{l}}(k,x)$ and $f_{\rm{r}}(k,x),$ which are respectively
given by
\begin{equation}
\label{4.6}
e^{-ikx}f_{\rm{l}}(k,x)=I+\displaystyle\frac{1}{2ik}\displaystyle\int_x^\infty dy\left[e^{2ik(y-x)}-1\right] V(y)\,e^{-iky}f_{\rm{l}}(k,y),
\end{equation}
\begin{equation}
\label{4.7}
e^{ikx}f_{\rm{r}}(k,x)=I+\displaystyle\frac{1}{2ik}\displaystyle\int_{-\infty}^x dy\left[e^{2ik(x-y)}-1\right] V(y)\,e^{iky}f_{\rm{r}}(k,y).
\end{equation}

For each fixed $k\in\mathbb R\setminus\{0\},$ the combined $2n$ columns of $f_{\rm{l}}(k,x)$ and $f_{\rm{r}}(k,x)$ form a fundamental set
for \eqref{1.1}, and
any solution to \eqref{1.1} can be expressed as a linear combination of those column-vector
solutions. 
The $n\times n$ matrix-valued scattering coefficients are defined \cite{AKV2001} in terms of the spacial asymptotics of
the Jost solutions via
\begin{equation}
\label{2.3}
f_{\rm{l}}(k,x)=e^{ikx}\,T_{\rm{l}}(k)^{-1}+e^{-ikx} L(k)\,T_{\rm{l}}(k)^{-1}+o(1),\qquad x\to-\infty,
\end{equation}
\begin{equation}
\label{2.4}
f_{\rm{r}}(k,x)=e^{-ikx}\,T_{\rm{r}}(k)^{-1}+e^{ikx} R(k)\,T_{\rm{r}}(k)^{-1}+o(1),\qquad x\to +\infty,
\end{equation}
where $T_{\rm{l}}(k)$ is the left transmission coefficient,
$T_{\rm{r}}(k)$ is the right transmission coefficient,
$L(k)$ is the left reflection coefficient, and $R(k)$ is the right
reflection coefficient. With the help of \eqref{4.6}--\eqref{2.4}, it can be shown that
\begin{equation}
\label{2.3a}
f'_{\rm{l}}(k,x)=ik\,e^{ikx}\,T_{\rm{l}}(k)^{-1}-ik\,e^{-ikx} L(k)\,T_{\rm{l}}(k)^{-1}+o(1),\qquad x\to-\infty,
\end{equation}
\begin{equation}
\label{2.4a}
f'_{\rm{r}}(k,x)=-ik\,e^{-ikx}\,T_{\rm{r}}(k)^{-1}+ik\,e^{ikx} R(k)\,T_{\rm{r}}(k)^{-1}+o(1),\qquad x\to +\infty.
\end{equation}
As a result, the leading asymptotics in \eqref{2.3a} and \eqref{2.4a} are obtained by taking the $x$-derivatives
of the leading asymptotics in \eqref{2.3} and \eqref{2.4}, respectively.
From (2.16) and (2.17) of \cite{AKV2001} it follows that the matrices $T_{\rm{l}}(k)$ and $T_{\rm{r}}(k)$ are invertible for $k \in \mathbb R \setminus\{0\}.$
We remark that \eqref{2.3} and \eqref{2.4} hold in the limit $k\to 0$ as their left-hand sides are continuous at $k=0$
even though each of the four matrices $T_{\rm{l}}(k)^{-1},$ $T_{\rm{r}}(k)^{-1},$ $L(k)\,T_{\rm{l}}(k)^{-1},$ and $R(k)\,T_{\rm{r}}(k)^{-1}$ generically
behaves as $O(1/k)$ when $k\to 0.$
The $2n\times 2n$ scattering matrix for \eqref{1.1} is defined as
\begin{equation}
\label{2.5}
S(k):=\begin{bmatrix}T_{\rm{l}}(k)&R(k) \\ 
\noalign{\medskip}
 L(k)& T_{\rm{r}}(k)\end{bmatrix},\qquad k\in\mathbb R.
\end{equation}
As seen from Theorem~3.1 and Theorem~4.6 of \cite{AKV2001}, the scattering coefficients
can be defined first via \eqref{2.3} and \eqref{2.4} for $ k \in \mathbb R \setminus\{0\}$ and then their domain
can be extended in a continuous way to include $ k=0.$ When the potential $V$ in \eqref{1.1} satisfies \eqref{1.2} and \eqref{1.3},
from \cite{MO1982,O1985} and the comments below (2.21) in \cite{AKV2001} it follows that the $n\times n$ matrix-valued
transmission coefficients
$T_{\rm{l}}(k)$ and $T_{\rm{r}}(k)$ have meromorphic
extensions in $k$ from $\mathbb R$ to $\mathbb C^+$ where any possible poles are simple and can only occur
on the positive imaginary axis. On the other hand,
the domains of the $n\times n$ matrix-valued
reflection coefficients
$L(k)$ and $R(k)$ cannot be extended from $k\in\mathbb R$ unless the
potential $V$ in \eqref{1.1} satisfies further restrictions besides \eqref{1.2} and \eqref{1.3}.

The left and right physical solutions to \eqref{1.1}, denoted by $\Psi_{\rm{l}}(k,x)$ and $\Psi_{\rm{r}}(k,x),$ are the
two particular
$n\times n$ matrix-valued solutions that are related to the Jost solutions
$f_{\rm{l}}(k,x)$ and $f_{\rm{r}}(k,x),$ respectively, as
\begin{equation}
\label{2.6}
\Psi_{\rm{l}}(k,x):=f_{\rm{l}}(k,x)\,T_{\rm{l}}(k),\quad \Psi_{\rm{r}}(k,x):=f_{\rm{r}}(k,x)\,T_{\rm{r}}(k),
\end{equation}
and, as seen from \eqref{2.1}, \eqref{2.2}, \eqref{2.3}, \eqref{2.4}, and \eqref{2.6} they satisfy the spacial asymptotics
\begin{equation}
\label{2.7}
\Psi_{\rm{l}}(k,x)=e^{ikx}T_{\rm{l}}(k)+o(1),\quad \Psi_{\rm{r}}(k,x)=e^{-ikx}I +e^{ikx} R(k)+o(1),\qquad x\to+\infty,
\end{equation}
\begin{equation}
\label{2.8}
\Psi_{\rm{l}}(k,x)=e^{ikx}I+e^{-ikx} L(k)+o(1),\quad \Psi_{\rm{r}}(k,x)=e^{-ikx} T_{\rm{r}}(k)+o(1),\qquad x\to-\infty.
\end{equation}
Using \eqref{2.7} and \eqref{2.8}, we can interpret $\Psi_{\rm{l}}(k,x)$ in terms of the matrix-valued plane wave $e^{ikx}I$ of unit amplitude
sent from $x=-\infty,$ the matrix-valued reflected plane wave $e^{-ikx} L(k)$ of amplitude
$L(k)$ at $x=-\infty,$ and the matrix-valued transmitted plane wave $e^{ikx}\,T_{\rm{l}}(k)$ of amplitude $T_{\rm{l}}(k)$ at $x=+\infty.$
Similarly, the physical solution $\Psi_{\rm{r}}(k,x)$ can be interpreted 
 in terms of the matrix-valued plane wave $e^{-ikx}I$ of unit amplitude
sent from $x=+\infty,$ the matrix-valued reflected plane wave $e^{ikx}R(k)$ of amplitude
$R(k)$ at $x=+\infty,$ and the matrix-valued transmitted plane wave $e^{-ikx}\,T_{\rm{r}}(k)$ of amplitude $T_{\rm{r}}(k)$ at $x=-\infty.$

The following proposition is a useful consequence of Theorem~7.3 on p.~28 of \cite{CL1955}, where
we use $\det$ and $\text{\rm{tr}}$ to denote the matrix determinant and the matrix trace, respectively.

\begin{proposition}\label{proposition2.1} Assume that the $n\times n$ matrix-valued potential $V$ in \eqref{1.1} satisfies \eqref{1.2} and \eqref{1.3}. Then,
for any pair of $n\times n$ matrix-valued solutions $\varphi(k,x)$ and $\psi(k,x)$ to \eqref{1.1}, the $2n\times 2n$ 
matrix determinant given by
\begin{equation}
\label{2.9}
\det\begin{bmatrix}\varphi (k,x) &\psi(k,x)\\
\noalign{\medskip}
\varphi'(k,x)&\psi'(k,x)\end{bmatrix},
\end{equation}
is independent of $x$
and can only depend on $k.$
\end{proposition}

\begin{proof}
The second-order matrix-valued systems \eqref{1.1} for $\varphi(k,x)$ and $\psi(k,x)$ can be expressed as a first-order $2n \times 2n$ matrix-valued
system as
\begin{equation}\label{2.10}
\displaystyle\frac{d}{dx}\begin{bmatrix}\varphi (k,x) &\psi (k,x)\\
\noalign{\medskip}
\varphi'(k,x)&\psi'(k,x)\end{bmatrix} =\begin{bmatrix} 0 & I \\
\noalign{\medskip}
V(x)- k^2& 0\end{bmatrix} \begin{bmatrix}\varphi (k,x) &\psi (k,x)\\
\noalign{\medskip}
\varphi'(k,x)&\psi' (k,x)\end{bmatrix}.
\end{equation}
From Theorem~7.3 on p.~28 of \cite{CL1955} we know that \eqref{2.10} implies
\begin{equation}\label{2.11}
\begin{split}
\displaystyle\frac{d}{dx}&\left(\det\begin{bmatrix}\varphi (k,x) &\psi(k,x)\\
\noalign{\medskip}
\varphi'(k,x)&\psi'(k,x)\end{bmatrix}\right)\\
&=\left( \text{\rm{tr}}\begin{bmatrix} 0 & I \\
\noalign{\medskip}
V(x)- k^2& 0 \end{bmatrix} \right)\left(\det\begin{bmatrix}\varphi (k,x) &\psi(k,x)\\
\noalign{\medskip}
\varphi'(k,x)&\psi'(k,x)\end{bmatrix}\right).
\end{split}
\end{equation}
Since the coefficient matrix in \eqref{2.10} has zero trace, the right-hand side of \eqref{2.11} is zero,
and hence the vanishment of the left-hand side of \eqref{2.11} shows that
the determinant in \eqref{2.9} cannot depend on $x.$
\end{proof}

Since $k$ appears as $k^2$ in \eqref{1.1} and we already know that
$f_{\rm{l}}(k,x)$ and $f_{\rm{r}}(k,x)$ are solutions to \eqref{1.1}, it follows that
$f_{\rm{l}}(-k,x)$ and $f_{\rm{r}}(-k,x)$ are also solutions
to \eqref{1.1}. From the known properties of the Jost solutions $f_{\rm{l}}(k,x)$ and $f_{\rm{r}}(k,x)$
we conclude that, for each $x\in
\mathbb R,$ the solutions $f_{\rm{l}}(-k,x)$ and $f_{\rm{r}}(-k,x)$ have analytic extensions  in $k$ from the real axis
$\mathbb R$ to the
lower-half complex plane $\mathbb C^-$
and they are continuous in $k\in\overline{\mathbb C^-},$ where we have defined $\overline{\mathbb C^-}:=\mathbb C^-\cup\mathbb R.$
In terms of the four solutions $f_{\rm{l}}(k,x),$ $f_{\rm{r}}(k,x),$
$f_{\rm{l}}(-k,x),$ $f_{\rm{r}}(-k,x)$
to \eqref{1.1}, we 
introduce three useful $2n\times 2n$ matrices as
\begin{equation}
\label{2.12}
F_{\rm{l}}(k,x):=\begin{bmatrix}f_{\rm{l}}(k,x)&f_{\rm{l}}(-k,x) \\
\noalign{\medskip}
f'_{\rm{l}}(k,x)&f'_{\rm{l}}(-k,x)\end{bmatrix},  \qquad x \in \mathbb R,
\end{equation}
\begin{equation}
\label{2.13}
F_{\rm{r}}(k,x):=\begin{bmatrix}f_{\rm{r}}(-k,x)&f_{\rm{r}}(k,x)\\
\noalign{\medskip}
f'_{\rm{r}}(-k,x)&f'_{\rm{r}}(k,x)\end{bmatrix}, \qquad x \in \mathbb R, 
\end{equation}
\begin{equation}
\label{2.14}
G(k,x):=\begin{bmatrix}f_{\rm{l}}(k,x)&f_{\rm{r}}(k,x)\\
\noalign{\medskip}
f'_{\rm{l}}(k,x)&f'_{\rm{r}}(k,x)\end{bmatrix},  \qquad x \in \mathbb R.
\end{equation}
Since the $k$-domains of those four solutions are known, from \eqref{2.12} and \eqref{2.13} we see that
$F_{\rm{l}}(k,x)$ and $F_{\rm{r}}(k,x)$ are defined when $k\in\mathbb R$ and 
from \eqref{2.14} we see that $G(k,x)$ is defined when $k\in\overline{\mathbb C^+}.$

In the next proposition, 
with the help of Proposition~\ref{proposition2.1} 
we show that the determinant of each of the three
matrices defined in \eqref{2.12}, \eqref{2.13}, and \eqref{2.14}, respectively, is independent of $x,$ and in fact
we determine the values of those determinants explicitly. 
We also establish the equivalence of the determinants of the left and right
transmission coefficients.

\begin{proposition}\label{proposition2.2}
Assume that the $n\times n$ matrix-valued potential $V$ in \eqref{1.1} satisfies \eqref{1.2} and \eqref{1.3}.
We then have the following:

\begin{enumerate}
\item[\text{\rm(a)}] The determinants of the $2n\times 2n$ matrices $F_{\rm{l}}(k,x),$ $F_{\rm{r}}(k,x),$ and $G(k,x)$ defined
in \eqref{2.12}, \eqref{2.13}, and \eqref{2.14}, respectively, 
are independent of $x,$ and we have
\begin{equation}\label{2.15}
\det\left[F_{\rm{l}}(k,x)\right]= (-2ik )^n, \qquad k \in \mathbb R,
\end{equation}
\begin{equation}\label{2.16}
\det\left[F_{\rm{r}}(k,x)\right]= (-2ik )^n, \qquad k \in \mathbb R,
\end{equation}
\begin{equation}\label{2.17}
\det\left[G(k,x)\right]= \frac{(-2ik )^n}{\det\left[T_{\rm{r}}(k)\right]}, \qquad k \in \overline{\mathbb C^+}.
\end{equation}

\item[\text{\rm(b)}] The $n\times n$ matrix-valued left transmission coefficient $T_{\rm{l}}(k)$ and right transmission coefficient $T_{\rm{r}}(k)$ 
have the same determinant, i.e. we have
\begin{equation} \label{2.18}
\det[T_{\rm{l}}(k)] = \det [T_{\rm{r}}(k)], \qquad k\in \overline{\mathbb C^+} . 
\end{equation}

\item[\text{\rm(c)}] We have
\begin{equation} \label{2.19}
\det[\Lambda(k)] =1,\quad \det[\Sigma(k)] =1,\qquad k\in\mathbb R,
\end{equation}
where the $2n\times 2n$ matrices $\Lambda(k)$ and $\Sigma(k)$ are defined in terms of the scattering coefficients in \eqref{2.5} as
\begin{equation}
\label{2.20}
\Lambda(k):=\begin{bmatrix}T_{\rm{l}}(k)^{-1}&L(-k)\,T_{\rm{l}}(-k)^{-1} \\
\noalign{\medskip}
L(k)\,T_{\rm{l}}(k)^{-1}&T_{\rm{l}}(-k)^{-1}\end{bmatrix},\qquad k\in\mathbb R \setminus\{0\},
\end{equation}
\begin{equation}
\label{2.21}
\Sigma(k):=\begin{bmatrix}T_{\rm{r}}(-k)^{-1}&R(k)\,T_{\rm{r}}(k)^{-1} \\
\noalign{\medskip}
R(-k)\,T_{\rm{r}}(-k)^{-1}&T_{\rm{r}}(k)^{-1}\end{bmatrix},\qquad k\in\mathbb R \setminus\{0\} .
\end{equation}
\end{enumerate}
\end{proposition}

\begin{proof} As already mentioned,
the four quantities $f_{\rm{l}}(k,x),$ $f_{\rm{r}}(k,x),$ $f_{\rm{l}}(-k,x),$ $f_{\rm{r}}(-k,x)$
are all solutions to \eqref{1.1}, and hence from Proposition~\ref{proposition2.1}  it follows that
the determinants of $F_{\rm{l}}(k,x),$ $F_{\rm{r}}(k,x),$ and $G(k,x)$ are independent of $x.$ In their respective $k$-domains, we can evaluate each of
those determinants as $x\to+\infty$ and $x\to-\infty,$ and we know that we have the equivalent values.
Using \eqref{2.1} in \eqref{2.12} we get
\begin{equation}\label{2.22}
\det\left[F_{\rm{l}}(k,x)\right]=\det\begin{bmatrix} e^{ikx}I&  e^{-ikx}I\\
\noalign{\medskip}
ik\, e^{ikx}I&-ik\, e^{-ikx}I\end{bmatrix}+ o(1), \qquad x \to +\infty.
\end{equation}
Using an elementary row block operation on the block matrix appearing on the right-hand side of \eqref{2.22}, i.e.
by multiplying the first row block by $ikI$ and subtracting the resulting row block from the second row block, we obtain
\eqref{2.15}.
In a similar way, using \eqref{2.2} in \eqref{2.13}, we have
\begin{equation}\label{2.23}
\det\left[F_{\rm{r}}(k,x)\right]=\det\begin{bmatrix} e^{ikx}I&  e^{-ikx}I\\
\noalign{\medskip}
ik\, e^{ikx}I &- ik\, e^{-ikx}I\end{bmatrix}+ o(1), \qquad x \to  -\infty,
\end{equation}
and again using the aforementioned elementary row block operation on the matrix appearing on the right-hand side of \eqref{2.23},
we get \eqref{2.16}. For $k\in\mathbb R\setminus\{0\}$ using \eqref{2.1} and \eqref{2.4}--\eqref{2.4a} in \eqref{2.14}, we obtain
\begin{equation}\label{2.24}
G(k,x)= K_{\rm{r}}(k,x)+o(1), \qquad x \to +\infty,
\end{equation}
where we have defined
\begin{equation*}
K_{\rm{r}}(k,x):=\begin{bmatrix} e^{ikx}I&  e^{-ikx}\,T_{\rm{r}}(k)^{-1}+ e^{ikx} R(k)\,T_{\rm{r}}(k)^{-1}  \\
\noalign{\medskip}
\displaystyle ik\, e^{ikx}I& \displaystyle  -ik\,e^{-ikx}\,T_{\rm{r}}(k)^{-1}+\displaystyle  ik\,e^{ikx} R(k)\,T_{\rm{r}}(k)^{-1}\end{bmatrix}.
\end{equation*}
Using the aforementioned elementary row block operation on the matrix $K_{\rm{r}}(k,x),$ we get
\begin{equation}\label{2.26}
\det\left[K_{\rm{r}}(k,x)\right]=\det \begin{bmatrix} e^{ikx}I&  e^{-ikx}\,T_{\rm{r}}(k)^{-1}+ e^{ikx} R(k)\,T_{\rm{r}}(k)^{-1}  \\
\noalign{\medskip}
\displaystyle 0& -2ik\,e^{-ikx}\,T_{\rm{r}}(k)^{-1}\end{bmatrix},
\end{equation}
where $0$ denotes the $n\times n$ zero matrix.
From \eqref{2.24} and \eqref{2.26} we get \eqref{2.17}. Thus, the proof of (a) is complete.
Let us now turn to the proof of (b).
Using \eqref{2.2}, \eqref{2.3}, and \eqref{2.3a} in \eqref{2.14}, we obtain
\begin{equation}\label{2.27}
G(k,x)= K_{\rm{l}}(k,x)+o(1), \qquad x \to -\infty,
\end{equation}
where we have defined
\begin{equation*}
K_{\rm{l}}(k,x):=\begin{bmatrix} e^{ikx}\,T_{\rm{l}}(k)^{-1}+ e^{-ikx} L(k)\,T_{\rm{l}}(k)^{-1}&  e^{-ikx}I \\
\noalign{\medskip}
ik\,e^{ikx}\,T_{\rm{l}}(k)^{-1}-ik e^{-ikx} L(k)\,T_{\rm{l}}(k)^{-1}&  -ik\,e^{-ikx}I\end{bmatrix}.
\end{equation*}
Using an elementary row block operation on the matrix $K_{\rm{l}}(k,x),$ i.e.
by multiplying the first row block by $ikI$ and adding the resulting row block to the second row block, we get
\begin{equation}\label{2.29}
\det\left[K_{\rm{l}}(k,x)\right]=\det \begin{bmatrix} e^{ikx}\,T_{\rm{l}}(k)^{-1}+ e^{-ikx} L(k)\,T_{\rm{l}}(k)^{-1}&  e^{-ikx}I \\
\noalign{\medskip}
\displaystyle 2ik\,e^{ikx}\,T_{\rm{l}}(k)^{-1}& 0\end{bmatrix}.
\end{equation}
Interchanging the first and second row blocks of the matrix appearing on the right-hand side of \eqref{2.29}, we have
\begin{equation}\label{2.30}
\det\left[K_{\rm{l}}(k,x)\right]=(-1)^n\det \begin{bmatrix} 2ik\,e^{ikx}\,T_{\rm{l}}(k)^{-1}&0 \\
\noalign{\medskip}
\displaystyle e^{ikx}\,T_{\rm{l}}(k)^{-1}+ e^{-ikx} L(k)\,T_{\rm{l}}(k)^{-1}&  e^{-ikx}I\end{bmatrix}.
\end{equation}
Then, from \eqref{2.27} and \eqref{2.30} we conclude that
\begin{equation}\label{2.31}
\det\left[G(k,x)\right]= \frac{(-2ik )^n}{\det\left[T_{\rm{l}}(k)\right]}, \qquad k \in \mathbb R\setminus\{0\},
\end{equation}
and by comparing \eqref{2.17} and \eqref{2.31} we obtain \eqref{2.17} for $k\in\mathbb R\setminus\{0\}.$
However, for each fixed $x\in\mathbb R$ the quantity $G(k,x)$ is analytic in
$k\in\mathbb C^+$ and continuous in $k\in\overline{\mathbb C^+}.$
Thus, \eqref{2.17} holds for $k\in\overline{\mathbb C^+}$
and the proof of (b) is also complete.
Using \eqref{2.3} and 
\eqref{2.3a} in \eqref{2.12}, and by exploiting the fact that
$\det[F_{\rm{l}}(k,x)]$ is independent of $x,$ we evaluate
$\det[F_{\rm{l}}(k,x)]$ as $x\to-\infty,$ and we get the equality
\begin{equation}\label{2.32}
\det\left[F_{\rm{l}}(k,x)\right]
=\det \begin{bmatrix} q_1& q_2\\
q_3& q_4\end{bmatrix},
\end{equation}
where we have defined
\begin{equation*}
q_1:=e^{ikx}T_{\rm{l}}(k)^{-1}+ e^{-ikx} L(k)\,T_{\rm{l}}(k)^{-1},
\end{equation*}
\begin{equation*}
q_2:=e^{-ikx}T_{\rm{l}}(-k)^{-1}+ e^{ikx} L(-k)\,T_{\rm{l}}(-k)^{-1},
\end{equation*}
\begin{equation*}
q_3:=ik\,e^{ikx}T_{\rm{l}}(k)^{-1}-ik\, e^{-ikx} L(k)\,T_{\rm{l}}(k)^{-1},
\end{equation*}
\begin{equation*}
q_4:= -ik\,e^{-ikx}T_{\rm{l}}(-k)^{-1}+ik\, e^{ikx} L(-k)\,T_{\rm{l}}(-k)^{-1}.
\end{equation*}
Using two consecutive elementary row block operations on the matrix on the right-hand side of \eqref{2.32}, i.e. by
multiplying the first row block by $ikI$ and adding the resulting row block to the second row block and then
by dividing the resulting second row block by $2ik$ and subtracting the resulting
row block from the first row block,
we can write \eqref{2.32}
as
\begin{equation}\label{2.33}
\det\left[F_{\rm{l}}(k,x)\right]
=\det \begin{bmatrix} e^{-ikx} L(k)\,T_{\rm{l}}(k)^{-1}&  e^{-ikx}T_{\rm{l}}(-k)^{-1}\\
\noalign{\medskip}
2ik\,e^{ikx}T_{\rm{l}}(k)^{-1}& 2ik\, e^{ikx} L(-k)\,T_{\rm{l}}(-k)^{-1}\end{bmatrix}.
\end{equation}
By interchanging the first and second row blocks of the matrix appearing on the right-hand side of \eqref{2.33} and
simplifying the determinant of the resulting matrix, from \eqref{2.33} we obtain
\begin{equation}\label{2.34}
\det\left[F_{\rm{l}}(k,x)\right]
=(-2ik)^n \det \begin{bmatrix} T_{\rm{l}}(k)^{-1}&   L(-k)\,T_{\rm{l}}(-k)^{-1}\\
\noalign{\medskip}
L(k)\,T_{\rm{l}}(k)^{-1}& T_{\rm{l}}(-k)^{-1}\end{bmatrix}.
\end{equation}
Comparing \eqref{2.15}, \eqref{2.20}, and \eqref{2.34}, we see that the first equality in \eqref{2.19} holds.
We remark that the matrix $\Lambda(k)$ defined in \eqref{2.20} behaves as $O(1/k)$ as $k\to 0.$
On the other hand, we observe that
the first equality of \eqref{2.19} holds at $k=0$ by the continuity
argument based on \eqref{1.6}--\eqref{1.8}.
In a similar way, using \eqref{2.4} and 
\eqref{2.4a} in \eqref{2.13}, and by exploiting the fact that
$\det[F_{\rm{r}}(k,x)]$ is independent of $x,$ 
we evaluate
$\det[F_{\rm{r}}(k,x)]$  as $x\to+\infty$ and we get the equality
\begin{equation}\label{2.35}
\det\left[F_{\rm{r}}(k,x)\right]
=\det \begin{bmatrix} q_5&  q_6\\
q_7& q_8\end{bmatrix},
\end{equation}
where we have defined
\begin{equation*}
q_5:= e^{ikx}\,T_{\rm{r}}(-k)^{-1}+ e^{-ikx} R(-k)\,T_{\rm{r}}(-k)^{-1},
\end{equation*}
\begin{equation*}
q_6:=e^{-ikx}\,T_{\rm{r}}(k)^{-1}+ e^{ikx} R(k)\,T_{\rm{r}}(k)^{-1},
\end{equation*}
\begin{equation*}
q_7:=ik\,e^{ikx}\,T_{\rm{r}}(-k)^{-1}-ik\, e^{-ikx} R(-k)\,T_{\rm{r}}(-k)^{-1},
\end{equation*}
\begin{equation*}
q_8:=-ik\,e^{-ikx}\,T_{\rm{r}}(k)^{-1}+ik\, e^{ikx} R(k)\,T_{\rm{r}}(k)^{-1}.
\end{equation*}
Using two elementary row block operations on the matrix on the right-hand side of \eqref{2.35},  i.e. by
multiplying the first row block by $ikI$ and adding the resulting row block to the second row block and then
by dividing the resulting second row block by $2ik$ and subtracting the resulting
row block from the first row block,
we can write \eqref{2.35}
as
\begin{equation}\label{2.36}
\det\left[F_{\rm{r}}(k,x)\right]
=\det \begin{bmatrix} e^{-ikx} R(-k)\,T_{\rm{r}}(-k)^{-1}&  e^{-ikx}\,T_{\rm{r}}(k)^{-1}\\
\noalign{\medskip}
2ik\,e^{ikx}\,T_{\rm{r}}(-k)^{-1}& 2ik\, e^{ikx} R(k)\,T_{\rm{r}}(k)^{-1}\end{bmatrix}.
\end{equation}
By interchanging the first and second row blocks of the matrix appearing on the right-hand side of \eqref{2.36} and
simplifying the determinant of the resulting matrix, from \eqref{2.36} we obtain
\begin{equation}\label{2.37}
\det\left[F_{\rm{r}}(k,x)\right]
=(-2ik)^n \det \begin{bmatrix} T_{\rm{r}}(-k)^{-1}&   R(k)\,T_{\rm{r}}(k)^{-1}\\
\noalign{\medskip}
R(-k)\,T_{\rm{r}}(-k)^{-1}& T_{\rm{r}}(k)^{-1}\end{bmatrix}.
\end{equation}
Comparing \eqref{2.16}, \eqref{2.21}, and \eqref{2.37}, 
we see that the second equality in \eqref{2.19} holds.
We remark that the matrix $\Sigma(k)$ defined in \eqref{2.21} behaves as $O(1/k)$ as $k\to 0,$
but the second equality of \eqref{2.19} holds at $k=0$ by the continuity argument
expressed in \eqref{1.6}--\eqref{1.8}.
\end{proof}

We note that, in the proof of 
Proposition~\ref{proposition2.2}, instead of using elementary row block operations, we could alternatively make use of the 
matrix factorization formula involving a Schur complement. Such a factorization formula is given by
\begin{equation}
\label{2.38}
\begin{bmatrix}M_1&M_2 \\
\noalign{\medskip}
M_3&M_4\end{bmatrix}=
\begin{bmatrix}I&0\\
\noalign{\medskip}
M_3\, M_1^{-1}&I\end{bmatrix}
\begin{bmatrix}M_1&0\\
\noalign{\medskip}
0&M_4-M_3\, M_1^{-1} \,M_2\end{bmatrix}
\begin{bmatrix}I&M_1^{-1}\,M_2\\
\noalign{\medskip}
0&I\end{bmatrix},
\end{equation}
which corresponds to (1.11) on p.~17 of \cite{D2006}.
In the alternative proof of Proposition~\ref{proposition2.2}, it is sufficient to use
 \eqref{2.38} in the special case where the block matrices $M_1,$ $M_2,$ $M_3,$ $M_4$
have the same size $n\times n$ and $M_1$ is invertible.

Let us use $[f(x);g(x)]$ to denote the Wronskian of two $n\times n$ matrix-valued functions of $x,$ where we have defined
\begin{equation*}
 [f(x);g(x)]:=f(x)\,g'(x)-f'(x)\,g(x).
\end{equation*}
Given any two $n\times n$ matrix-valued solutions $\xi(k,x)$ and $\psi(k,x)$ to \eqref{1.1},
one can directly verify that the Wronskian $[\xi(\pm k^*,x)^\dagger;\psi(k,x)]$ is independent of $x,$ 
where we use an asterisk to denote complex conjugation. 
Evaluating the Wronskians involving the Jost solutions to \eqref{1.1} as $x\to+\infty$ and
also as $x\to-\infty,$ respectively, for $k\in\mathbb R$
we obtain
\begin{equation}
\label{2.40}
 [f_{\rm{l}}(k,x)^\dagger;f_{\rm{l}}(k,x)]=2ikI=2ik\,[T_{\rm{l}}(k)^\dagger]^{-1}\left[I-L(k)^\dagger L(k)\right]T_{\rm{l}}(k)^{-1},
\end{equation}
\begin{equation}
\label{2.41}
 [f_{\rm{l}}(-k,x)^\dagger;f_{\rm{l}}(k,x)]=0=2ik\, [T_{\rm{l}}(-k)^\dagger]^{-1}\left[L(-k)^\dagger-L(k)\right] T_{\rm{l}}(k)^{-1},
\end{equation}
\begin{equation}
\label{2.42}
 [f_{\rm{r}}(k,x)^\dagger;f_{\rm{r}}(k,x)]=
 -2ik\,[T_{\rm{r}}(k)^\dagger]^{-1}\left[I-R(k)^\dagger R(k)\right]T_{\rm{r}}(k)^{-1}=-2ikI,
\end{equation}
\begin{equation}
\label{2.43}
 [f_{\rm{r}}(-k,x)^\dagger;f_{\rm{r}}(k,x)]=
 2ik\, [T_{\rm{r}}(-k)^\dagger]^{-1}\left[R(k)-R(-k)^\dagger\right] T_{\rm{r}}(k)^{-1}=0,
\end{equation}
\begin{equation}
\label{2.44}
 [f_{\rm{l}}(k,x)^\dagger;f_{\rm{r}}(k,x)]=2ik\,R(k)\,T_{\rm{r}}(k)^{-1}=-2ik\,[T_{\rm{l}}(k)^\dagger]^{-1}
 L(k)^\dagger,
\end{equation}
and for $k\in\overline{\mathbb C^+}$ we get
 \begin{equation}
\label{2.45}
 [f_{\rm{l}}(-k^*,x)^\dagger;f_{\rm{r}}(k,x)]=-2ik\,T_{\rm{r}}(k)^{-1}=-2ik\,[T_{\rm{l}}(-k^\ast)^\dagger]^{-1},
\end{equation}
 \begin{equation}
\label{2.46}
 [f_{\rm{r}}(-k^*,x)^\dagger;f_{\rm{l}}(k,x)]=2ik\,[T_{\rm{r}}(-k^\ast)^\dagger]^{-1}=2ik\,T_{\rm{l}}(k)^{-1}.
\end{equation}

With the help of \eqref{2.40}, \eqref{2.42}, and \eqref{2.44}, we can prove that
\begin{equation}
\label{2.47}
S(k)^\dagger S(k)=\begin{bmatrix}I&0\\
0&I\end{bmatrix},\qquad k\in\mathbb R,
\end{equation}
where $S(k)$ is the scattering matrix defined in \eqref{2.5}.
From \eqref{2.47} we conclude that  $S(k)$ is unitary, and hence we also have
\begin{equation}
\label{2.48}
S(k)\,S(k)^\dagger=\begin{bmatrix}I&0\\
0&I\end{bmatrix},\qquad k\in\mathbb R.
\end{equation}
From \eqref{2.41}, \eqref{2.43}, and \eqref{2.45} we obtain
\begin{equation}
\label{2.49}
L(-k)=L(k)^\dagger,\quad R(-k)=R(k)^\dagger, \quad T_{\rm{l}}(-k)=T_{\rm{r}}(k)^\dagger,\qquad k\in\mathbb R,
\end{equation}
which can equivalently be expressed as
\begin{equation*}
S(k)^\dagger= Q S(-k)Q,  \qquad k \in \mathbb R,
\end{equation*}
where $Q$ is the constant $2n\times 2n$ matrix given by
\begin{equation}
\label{2.51}
Q:=\begin{bmatrix}0&I\\
I&0\end{bmatrix}.
\end{equation}
We remark that the matrix $Q$ is equal to its own inverse.

For easy referencing, for $k\in\mathbb R$
we write \eqref{2.47} and \eqref{2.48} explicitly as
\begin{equation}
\label{2.52}
\begin{bmatrix}T_{\rm{l}}(k)^\dagger \,T_{\rm{l}}(k)+L(k)^\dagger\, L(k) &T_{\rm{l}}(k)^\dagger \,R(k)+L(k)^\dagger\, T_{\rm{r}}(k)\\
\noalign{\medskip}
R(k)^\dagger \,T_{\rm{l}}(k)+T_{\rm{r}}(k)^\dagger\, L(k)&T_{\rm{r}}(k)^\dagger \,T_{\rm{r}}(k)+R(k)^\dagger\, R(k)\end{bmatrix}=\begin{bmatrix}I&0\\
0&I\end{bmatrix},
\end{equation}
\begin{equation}
\label{2.53}
\begin{bmatrix}T_{\rm{l}} (k)\,T_{\rm{l}}(k)^\dagger+R(k) \,R(k)^\dagger&T_{\rm{l}}(k)\, L(k)^\dagger+R(k) \,T_{\rm{r}}(k)^\dagger\\
\noalign{\medskip}
L(k)\, T_{\rm{l}}(k)^\dagger +T_{\rm{r}} (k)\,R(k)^\dagger&T_{\rm{r}} (k)\,T_{\rm{r}}(k)^\dagger+L(k) \,L(k)^\dagger \end{bmatrix}=\begin{bmatrix}I&0\\
0&I\end{bmatrix}.
\end{equation}

In Proposition~\ref{proposition2.2}(a) we have evaluated the determinants of
the matrices $F_{\rm{l}}(k,x),$ $F_{\rm{r}}(k,x),$ and $G(k,x)$ appearing in
\eqref{2.12}, \eqref{2.13}, and \eqref{2.14}, respectively.
In the next theorem we determine their matrix inverses explicitly in terms of
the Jost solutions $f_{\rm{l}}(k,x)$ and $f_{\rm{r}}(k,x).$

\begin{theorem}\label{theorem2.3}
Assume that the $n\times n$ matrix-valued potential $V$ in \eqref{1.1} satisfies \eqref{1.2} and \eqref{1.3}.
We then have the following:

\begin{enumerate}
\item[\text{\rm(a)}] The $2n\times 2n$ matrix $F_{\rm{l}}(k,x)$ defined in \eqref{2.12} is invertible
when $k\in\mathbb R\setminus\{0\},$ and we have
\begin{equation}
\label{2.54}
F_{\rm{l}}(k,x)^{-1}=\displaystyle\frac{1}{2ik}
\begin{bmatrix}-f'_{\rm{l}}(k,x)^\dagger&f_{\rm{l}}(k,x)^\dagger\\
\noalign{\medskip}
f'_{\rm{l}}(-k,x)^\dagger&-f_{\rm{l}}(-k,x)^\dagger\end{bmatrix},\qquad k\in\mathbb R\setminus\{0\}.
\end{equation}

\item[\text{\rm(b)}] The $2n\times 2n$ matrix $F_{\rm{r}}(k,x)$ defined in \eqref{2.13} is invertible
when $k\in\mathbb R\setminus\{0\},$ and we have
\begin{equation}
\label{2.55}
F_{\rm{r}}(k,x)^{-1}=\displaystyle\frac{1}{2ik}
\begin{bmatrix}-f'_{\rm{r}}(-k,x)^\dagger&f_{\rm{r}}(-k,x)^\dagger\\
\noalign{\medskip}
f'_{\rm{r}}(k,x)^\dagger&-f_{\rm{r}}(k,x)^\dagger\end{bmatrix},\qquad k\in\mathbb R\setminus\{0\}.
\end{equation}

\item[\text{\rm(c)}] The $2n\times 2n$ matrix $G(k,x)$ given in \eqref{2.14} is invertible
for $k\in\overline{\mathbb C^+}\setminus\{0\}$ except at the poles
of the determinant of the transmission coefficient $T_{\rm{l}}(k),$
where such poles can only occur on the positive
imaginary axis in the complex $k$-plane, those $k$-values correspond to
the bound states of \eqref{1.1}, and the number of such poles is finite.
Furthermore, for $k\in\overline{\mathbb C^+}\setminus\{0\}$ we have
\begin{equation}
\label{2.56}
G(k,x)^{-1}=-\displaystyle\frac{1}{2ik}
\begin{bmatrix}T_{\rm{l}} (k)&0\\
\noalign{\medskip}
0&T_{\rm{r}} (k)\end{bmatrix}
\begin{bmatrix}f'_{\rm{r}}(-k^\ast,x)^\dagger&-f_{\rm{r}}(-k^\ast,x)^\dagger\\
\noalign{\medskip}
-f'_{\rm{l}}(-k^\ast,x)^\dagger&f_{\rm{l}}(-k^\ast,x)^\dagger\end{bmatrix}.
\end{equation}
\end{enumerate}
\end{theorem}

\begin{proof}
From \eqref{2.15} we see that the determinant of $F_{\rm{l}}(k,x)$ vanishes
only at $k=0$ on the real axis.
We confirm \eqref{2.54} by direct verification. This is done
by first postmultiplying the right-hand side of \eqref{2.54} with the matrix
$F_{\rm{l}}(k,x)$ given in \eqref{2.12} and then by simplifying the block entries of the resulting matrix product
with the help of \eqref{2.40} and \eqref{2.41}.
Thus, the proof of (a) is complete.
We prove (b) in a similar manner. From \eqref{2.16} we observe that
$\det[F_{\rm{r}}(k,x)]$ is nonzero when $k\in\mathbb R$ except at $k=0.$
Consequently, the matrix $F_{\rm{r}}(k,x)$ is invertible when $k\in\mathbb R\setminus \{0\}.$
By postmultiplying the right-hand side of \eqref{2.55}
by the matrix $F_{\rm{r}}(k,x)$ given in \eqref{2.13}, we simplify
the resulting matrix product with the help of \eqref{2.42} and \eqref{2.43}
and
verify
that we obtain the $2n\times 2n$ identity matrix as the product.
Thus, the proof of (b) is also complete.
For the proof of (c) we proceed as follows.
From \eqref{2.17} we observe that $\det[G(k,x)]$ is nonzero
when $k\in\overline{\mathbb C^+}$ except when $k=0$
and when $\det[T_{\rm{r}}(k)]$ has poles. From \eqref{2.18} we know that
the determinants of $T_{\rm{l}}(k)$ and $T_{\rm{r}}(k)]$ coincide, and 
from \cite{AKV2001} we know that
the poles of $\det[T_{\rm{l}}(k)]$ correspond to
the $k$-values at which the bound states of \eqref{1.1} occur.
It is also known \cite{AKV2001} that the bound-state
$k$-values can only occur on the positive imaginary axis in the complex $k$-plane and that
the number of such $k$-values is finite.
We verify \eqref{2.56} directly. That is done by postmultiplying both sides of \eqref{2.56} with the matrix $G(k,x)$ defined in \eqref{2.14},
by simplifying the matrix product by using \eqref{2.41}, \eqref{2.43}, \eqref{2.45}, and \eqref{2.46}, and by showing that the corresponding product is equal to
the $2n\times 2n$ identity matrix.
\end{proof}

\section{The factorization formulas}
\label{section3}

In this section we provide a factorization formula for the full-line matrix Schr\"odinger equation \eqref{1.1}, by
relating the matrix-valued scattering coefficients corresponding to the potential $V$
appearing in \eqref{1.1} to the matrix-valued scattering coefficients corresponding to the fragments of $V.$
We also present an alternate version of
the factorization formula, which is equivalent to the original version.

We already know that $f_{\rm{l}}(k,x)$ and $f_{\rm{l}}(-k,x)$ are both $n\times n$ matrix-valued solutions to
\eqref{1.1}, and from \eqref{2.1} we conclude that their combined $2n$ columns  form a fundamental set of 
column-vector solutions to \eqref{1.1} when $k\in\mathbb R\setminus\{0\}.$ Hence, we can express $f_{\rm{r}}(k,x)$ 
as a linear combination of
those $2n$ columns, and we get
\begin{equation}\label{3.1}
f_{\rm{r}}(k,x)= f_{\rm{l}}(k,x)\, R(k) \,T_{\rm{r}}(k)^{-1}+ f_{\rm{l}}(-k,x) \,T_{\rm{r}}(k)^{-1}, \qquad k \in \mathbb R ,
\end{equation}
where the coefficient matrices are obtained by letting $x\to+\infty$ in \eqref{3.1} and using
\eqref{2.1} and \eqref{2.4}. Note that we have included
$k=0$ in \eqref{3.1} by using the continuity of $f_{\rm{l}}(k,x)$ at $k=0$ for
each fixed $x\in\mathbb R.$
Similarly, $f_{\rm{r}}(k,x)$ and $f_{\rm{r}}(-k,x)$ are both $n\times n$ matrix-valued solutions to
\eqref{1.1}, and from \eqref{2.2} we conclude that their combined $2n$ columns form a fundamental set of 
column-vector solutions to \eqref{1.1}
 when $k\in\mathbb R\setminus\{0\}.$ Thus, we have
\begin{equation}\label{3.2}
f_{\rm{l}}(k,x)= f_{\rm{r}}(k,x) \,L(k)\,T_{\rm{l}}(k)^{-1}+ f_{\rm{r}}(-k,x)\, T_{\rm{l}}(k)^{-1}, \qquad k \in \mathbb R,
\end{equation}
where we have obtained the coefficient matrices by letting $x\to-\infty$ in \eqref{3.2}
and by using \eqref{2.2} and \eqref{2.3}.
Note that we can write \eqref{3.1} and \eqref{3.2}, respectively, as
\begin{equation}\label{3.3}
f_{\rm{l}}(-k,x)= f_{\rm{r}}(k,x)\, T_{\rm{r}}(k)- f_{\rm{l}}(k,x) \,R(k), \qquad k \in \mathbb R ,
\end{equation}
\begin{equation}\label{3.4}
f_{\rm{r}}(-k,x)= f_{\rm{l}}(k,x)\, T_{\rm{l}}(k)- f_{\rm{r}}(k,x) \,L(k), \qquad k \in \mathbb R.
\end{equation}

Using \eqref{3.2} in \eqref{2.12} and comparing the result with \eqref{2.13}, we see that the matrices $F_{\rm{l}}(k,x)$ and $F_{\rm{r}}(k,x)$ 
defined in \eqref{2.12} and \eqref{2.13}, respectively, are related to each other as
\begin{equation}
\label{3.5}
F_{\rm{l}}(k,x)=F_{\rm{r}}(k,x)\,\Lambda(k),\qquad k\in\mathbb R,
\end{equation}
where $\Lambda(k)$ is the matrix defined in \eqref{2.20}.
We remark that, even though $\Lambda(k)$ has
the behavior $O(1/k)$ as $k\to 0,$ by the continuity the equality in
\eqref{3.5} holds also at $k=0.$
In a similar way, by using 
\eqref{3.1} in \eqref{2.13} and comparing the result with \eqref{2.12},
we obtain
\begin{equation}
\label{3.6}
F_{\rm{r}}(k,x)=F_{\rm{l}}(k,x)\,\Sigma(k),\qquad k\in\mathbb R,
\end{equation}
where $\Sigma(k)$ is the matrix defined in \eqref{2.21}.
By \eqref{2.15} and \eqref{2.16}, we know that
the matrices $F_{\rm{l}}(k,x)$ and $F_{\rm{r}}(k,x)$ are invertible
when $k\in\mathbb R\setminus\{ 0 \}.$
Thus, from \eqref{3.5} and \eqref{3.6} we conclude that 
$\Lambda(k)$ and $\Sigma(k)$ are inverses of each other for each $k\in\mathbb R\setminus\{ 0 \},$ i.e. we have
\begin{equation}
\label{3.7}
 \Lambda(k)\,\Sigma(k)=\Sigma(k)\,\Lambda(k)=\begin{bmatrix}I&0\\
0&I\end{bmatrix},\qquad k\in\mathbb R,
\end{equation}
where the result in \eqref{3.7} holds also at $k=0$ by the continuity.
We already know from \eqref{2.19} that the determinants
of the matrices $\Lambda(k)$ and $\Sigma(k)$ are both equal to $1.$
We remark that \eqref{3.7} yields a wealth of relations among the left and right matrix-valued scattering coefficients, which
are similar to those given in \eqref{2.49}, \eqref{2.52}, and \eqref{2.53}.

Using \eqref{3.3} and \eqref{3.4} in \eqref{2.14}, we get
\begin{equation}\label{3.8}
G(-k,x)= G(k,x) \begin{bmatrix}- R(k)& T_{\rm{l}}(k)\\
\noalign{\medskip}
T_{\rm{r}}(k)& -L(k) \end{bmatrix},  \qquad k \in \mathbb R.
\end{equation}
We can write \eqref{3.8} in terms of the scattering matrix $S(k)$ appearing in \eqref{2.5} as
\begin{equation}\label{3.9}
G(-k,x) = G(k,x) \,J \,S(k)\, J \,Q, \qquad k \in \mathbb R,
\end{equation}
where $Q$ is the $2n\times 2n$ constant matrix defined in \eqref{2.51} and 
$J$ is the $2n\times 2n$ involution matrix defined as
\begin{equation}
\label{3.10}
J:=\begin{bmatrix}I&0\\
0&-I\end{bmatrix}.
\end{equation}
Note that $J$ is equal to its own inverse.
By taking the determinants of both sides of \eqref{3.9} and using \eqref{2.17} and
the fact that $\det[J]=(-1)^n$ and $\det[Q]=(-1)^n,$ we obtain the determinant of the scattering matrix as
\begin{equation}
\label{3.11}
\det\left[S(k)\right]= \frac{\det\left[T_{\rm{r}}(k)\right]}{\det\left[T_{\rm{r}}(-k)\right]}, \qquad k \in \mathbb R.
\end{equation}
Because of \eqref{2.18}, we can write \eqref{3.11} also as
\begin{equation*}
\det\left[S(k)\right]= \frac{\det\left[T_{\rm{l}}(k)\right]}{\det\left[T_{\rm{l}}(-k)\right]}, \qquad k \in \mathbb R.
\end{equation*}
From \eqref{2.18} and \eqref{2.49} we see that
\begin{equation}
\label{3.13}
\det\left[T_{\rm{l}}(-k)\right]=\left(\det\left[T_{\rm{l}}(k)\right]\right)^*,\quad
\det\left[T_{\rm{r}}(-k)\right]=\left(\det\left[T_{\rm{r}}(k)\right]\right)^*,\qquad k \in \mathbb R,
\end{equation}
where we recall that an asterisk is used to denote complex conjugation.
Hence, \eqref{3.11} and \eqref{3.13} imply that
\begin{equation}
\label{3.14}
\det\left[S(k)\right]= \frac{\det\left[T_{\rm{l}}(k)\right]}{\left(\det\left[T_{\rm{l}}(k)\right]\right)^\ast}, \qquad k \in \mathbb R.
\end{equation}

The next proposition indicates how 
the relevant quantities related to the full-line Schr\"odinger equation 
\eqref{1.1} are affected when the potential is shifted by $b$ units to the right, i.e. when we replace $V(x)$ in \eqref{1.1} by $V^{(b)}(x)$ defined as
\begin{equation}\label{3.15}
V^{(b)}(x):= V(x+b), \qquad b \in \mathbb R.
\end{equation}
We use the superscript $(b)$ to denote the corresponding transformed quantities.
The result will be useful in showing that the factorization formulas \eqref{3.45} and \eqref{3.46} remain unchanged if the potential is decomposed into two pieces at any fragmentation point
instead the fragmentation point $x=0$ used in \eqref{1.4}.

\begin{proposition}
\label{proposition3.1}
Consider the full-line $n\times n$ matrix Schr\"odinger equation \eqref{1.1} with the $n\times n$ matrix
potential $V$ satisfying \eqref{1.2} and \eqref{1.3}. Under the transformation $ V(x)\mapsto V^{(b)}(x)$
described in \eqref{3.15}, the  quantities relevant to \eqref{1.1}
are transformed as follows:

\begin{enumerate}

\item[\text{\rm(a)}] The $n\times n$ Jost solutions $f_{\rm{l}}(k,x)$ and $f_{\rm{r}}(k,x)$ are transformed 
into $f_{\rm{l}}^{(b)}(k,x)$ and $f_{\rm{r}}^{(b)}(k,x),$ respectively, where we have defined
\begin{equation}
\label{3.16} f_{\rm{l}}^{(b)}(k,x):=  e^{-ikb }f_{\rm{l}}(k,x+b),\quad f_{\rm{r}}^{(b)}(k,x):=  e^{ikb}  f_{\rm{r}}(k,x+b).
\end{equation}

\item[\text{\rm(b)}] The $n\times n$ matrix-valued left and right transmission coefficients
appearing in \eqref{2.3} and \eqref{2.4} remain unchanged, i.e. we have
\begin{equation}
\label{3.17} T_{\rm{l}}^{(b)}(k)= T_{\rm{l}}(k),\quad T_{\rm{r}}^{(b)}(k)= T_{\rm{r}}(k).
\end{equation}
The  $n\times n$ matrix-valued left and right reflection coefficients $L(k)$ and
$R(k)$ appearing in \eqref{2.3} and \eqref{2.4},
respectively, are transformed into $L^{(b)}(k)$ and $R^{(b)}(k),$ which are defined as
\begin{equation}
\label{3.18}  L^{(b)}(k):=L(k)\,e^{-2ikb}, \quad R^{(b)}(k) :=R(k)\,e^{2ikb}.
\end{equation}

\item[\text{\rm(c)}] The $2n\times 2n$ transition matrix  $\Lambda(k)$ appearing in \eqref{2.20} is transformed into
$\Lambda^{(b)}(k)$ given by
\begin{equation}
\label{3.19}   \Lambda^{(b)}(k):=  \begin{bmatrix}
e^{ikb} I&0
\\
\noalign{\medskip}
0&e^{-ikb} I\end{bmatrix} \Lambda(k)\begin{bmatrix}
e^{-ikb} I&0
\\
\noalign{\medskip}
0&e^{ikb} I\end{bmatrix}.
\end{equation}
The $2n\times 2n$ transition matrix $\Sigma(k)$ appearing in \eqref{2.21} is transformed into
$\Sigma^{(b)}(k)$ given by
\begin{equation}
\label{3.20}   
\Sigma^{(b)}(k):=  \begin{bmatrix}
e^{ikb} I&0
\\
\noalign{\medskip}
0&e^{-ikb} I\end{bmatrix} \Sigma(k)\begin{bmatrix}
e^{-ikb} I&0
\\
\noalign{\medskip}
0&e^{ikb} I\end{bmatrix} .
\end{equation}

\end{enumerate}
\end{proposition}

\begin{proof}
The matrix $f_{\rm{l}}^{(b)}(k,x)$ defined in
the first equality of \eqref{3.16} is the transformed left Jost solution because it
satisfies the transformed matrix-valued Schr\"odinger equation
\begin{equation}
\label{3.21}
-\psi''(k,x)+V(x+b)\,\psi(k,x)=k^2\,\psi(k,x),\qquad x\in\mathbb R,
\end{equation}  
and is asymptotic to $e^{ikx}[I+o(1)]$ as $x\to +\infty.$
Similarly, the matrix $f_{\rm{r}}^{(b)}(k,x)$ defined in
the second equality of \eqref{3.16} is the transformed right Jost solution because it
satisfies \eqref{3.21} and
is asymptotic to $e^{-ikx}[I+o(1)]$ as $x\to -\infty.$
Thus, the proof of (a) is complete.
Using $f_{\rm{l}}^{(b)}(k,x)$ in the analog of \eqref{2.3}, as $x\to-\infty$ we have
\begin{equation}
\label{3.22}
f_{\rm{l}}^{(b)}(k,x)=e^{ikx}\,[T_{\rm{l}}^{(b)}(k)]^{-1}+e^{-ikx} L^{(b)}(k)\,[T_{\rm{l}}^{(b)}(k)]^{-1}+o(1).
\end{equation}
Using the right-hand side of the first equality of \eqref{3.16} in \eqref{3.22} and comparing the result
with \eqref{2.3} we get the first equalities
of \eqref{3.17} and \eqref{3.18}, respectively.
Similarly, using $f_{\rm{r}}^{(b)}(k,x)$ in the analog of \eqref{2.4}, as $x\to+\infty$ we get
\begin{equation}
\label{3.23}
f_{\rm{r}}^{(b)}(k,x)=e^{-ikx}\,[T_{\rm{r}}^{(b)}(k)]^{-1}+e^{ikx} R^{(b)}(k)\,[T_{\rm{r}}^{(b)}(k)]^{-1}+o(1).
\end{equation}
Using the right-hand side of the second equality of \eqref{3.16} in \eqref{3.23} and comparing the result
with \eqref{2.4}, we obtain the second equalities
of \eqref{3.17} and \eqref{3.18}, respectively.
Hence, the proof of (b) is also complete.
Using \eqref{3.17} and \eqref{3.18} in the analogs of
\eqref{2.20} and \eqref{2.21} corresponding to the shifted potential
$V^{(b)},$ we obtain \eqref{3.19} and \eqref{3.20}.
\end{proof}

When the potential $V$ in \eqref{1.1} satisfies \eqref{1.2} and \eqref{1.3}. 
let us decompose it as in \eqref{1.4}.
For the left fragment $V_1$ and the right fragment $V_2$
defined in \eqref{1.5}, let us use
the subscripts $1$ and $2,$ respectively, to denote the corresponding relevant quantities.
Thus, analogous to \eqref{2.5} we define the $2n\times 2n$ scattering matrices
$S_1(k)$ and $S_2(k)$ corresponding to $V_1$ and $V_2,$ respectively, as
\begin{equation}
\label{3.24}
S_1(k):=\begin{bmatrix}T_{\rm{l}1}(k)&R_1(k) \\ 
\noalign{\medskip}
 L_1(k)& T_{\rm{r}1}(k)\end{bmatrix},\quad S_2(k):=\begin{bmatrix}T_{\rm{l}2}(k)&R_2(k) \\ 
\noalign{\medskip}
 L_2(k)& T_{\rm{r}2}(k)\end{bmatrix},\qquad k\in\mathbb R,
\end{equation}
where $T_{\rm{l}1}(k)$ and $T_{\rm{l}2}(k)$ are the respective left
transmission coefficients,
$T_{\rm{r}1}(k)$ and $T_{\rm{r}2}(k)$ are the respective right
transmission coefficients,
$L_1(k)$ and $L_2(k)$ are the respective left
reflection
coefficients, and
$R_1(k)$ and $R_2(k)$ are the respective right
reflection
coefficients.
In terms of the scattering coefficients for the respective fragments, we use
$\Lambda_1(k)$ and $\Lambda_2(k)$ as in \eqref{2.20} and
use
$\Sigma_1(k)$ and $\Sigma_2(k)$
as in \eqref{2.21} to denote
the transition matrices corresponding to the left and
right potential fragments $V_1$ and $V_2.$ Thus, we have
\begin{equation}
\label{3.25}
\Lambda_1(k):=\begin{bmatrix}T_{\rm{l}1}(k)^{-1}&L_1(-k)\,T_{\rm{l}1}(-k)^{-1} \\
\noalign{\medskip}
L_1(k)\,T_{\rm{l}1}(k)^{-1}&T_{\rm{l}1}(-k)^{-1}\end{bmatrix},\qquad k\in\mathbb R \setminus\{0\},\end{equation}
\begin{equation}
\label{3.26}
 \Lambda_2(k):=\begin{bmatrix}T_{\rm{l}2}(k)^{-1}&L_2(-k)\,T_{\rm{l}2}(-k)^{-1} \\
\noalign{\medskip}
L_2(k)\,T_{\rm{l}2}(k)^{-1}&T_{\rm{l}2}(-k)^{-1}\end{bmatrix},\qquad k\in\mathbb R\setminus\{0\},
\end{equation}
\begin{equation}
\label{3.27}
\Sigma_1(k):=\begin{bmatrix}T_{\rm{r}1}(-k)^{-1}&R_1(k)\,T_{\rm{r}1}(k)^{-1} \\
\noalign{\medskip}
R_1(-k)\,T_{\rm{r}1}(-k)^{-1}&T_{\rm{r}1}(k)^{-1}\end{bmatrix},\qquad k\in\mathbb R\setminus\{0\},\end{equation}
\begin{equation}
\label{3.28}
 \Sigma_2(k):=\begin{bmatrix}T_{\rm{r}2}(-k)^{-1}&R_2(k)\,T_{\rm{r}2}(k)^{-1} \\
\noalign{\medskip}
R_2(-k)\,T_{\rm{r}2}(-k)^{-1}&T_{\rm{r}2}(k)^{-1}\end{bmatrix},\qquad k\in\mathbb R\setminus\{0\}.
\end{equation}

In preparation for the proof of our factorization formula, in the next proposition we express the value at $x=0$ of
the matrix $F_{\rm{l}}(k,x)$ defined in \eqref{2.12} in terms of the left scattering coefficients for the right fragment $V_2,$ and
similarly we 
express the value at $x=0$ of the matrix
$F_{\rm{r}}(k,x)$ defined in
\eqref{2.13} 
in terms of the right scattering coefficients  for the left fragment $V_1.$

\begin{proposition}
\label{proposition3.2}
Consider the full-line $n\times n$ matrix Schr\"odinger equation \eqref{1.1}, where the potential $V$ satisfies
\eqref{1.2} and \eqref{1.3} and is fragmented as in \eqref{1.4} into the left fragment
$V_1$ and the right fragment $V_2$ defined in \eqref{1.5}.
We then have the following:

\begin{enumerate}

\item[\text{\rm(a)}] For $k\in\mathbb R,$ the $2n\times 2n$ matrix $F_{\rm{l}}(k,x)$ defined in \eqref{2.12} satisfies
\begin{equation}
\label{3.29}
F_{\rm{l}}(k,0)=\begin{bmatrix}\left[I+L_2(k)\right]T_{\rm{l}2}(k)^{-1}&\left[I+L_2(-k)\right]T_{\rm{l}2}(-k)^{-1}\\
\noalign{\medskip}
ik\left[I-L_2(k)\right]T_{\rm{l}2}(k)^{-1}&-ik\left[I-L_2(-k)\right]T_{\rm{l}2}(-k)^{-1}\end{bmatrix},
\end{equation}
where $T_{\rm{l}2}(k)$ and $L_2(k)$ are the left transmission and left reflection coefficients, respectively,
for the right fragment $V_2.$

\item[\text{\rm(b)}] For $k\in\mathbb R,$ the $2n\times 2n$ matrix $F_{\rm{r}}(k,x)$ defined in \eqref{2.13} satisfies
\begin{equation}
\label{3.30}
F_{\rm{r}}(k,0)=\begin{bmatrix}\left[I+R_1(-k)\right] T_{\rm{r}1}(-k)^{-1}&\left[I+R_1(k)\right] T_{\rm{r}1}(k)^{-1}\\
\noalign{\medskip}
ik\left[I-R_1(-k)\right] T_{\rm{r}1}(-k)^{-1}&ik\left[-I+R_1(k)\right] T_{\rm{r}1}(k)^{-1}\end{bmatrix},
\end{equation}
where $T_{\rm{r}1}(k)$ and $R_1(k)$ are the right transmission and right reflection coefficients, respectively,
for the left fragment $V_1.$

\item[\text{\rm(c)}] The matrix $F_{\rm{l}}(k,0)$ appearing in \eqref{3.29} can be written as a matrix 
product as
\begin{equation}
\label{3.31}
F_{\rm{l}}(k,0)=q_9\, q_{10}\,q_{11},\qquad k\in\mathbb R,
\end{equation}
where we have defined
\begin{equation}
\label{3.32}
q_9:=\begin{bmatrix}I&0\\
0&ikI\end{bmatrix}\begin{bmatrix}I&-I\\
0&I\end{bmatrix}
\begin{bmatrix}2I&0\\
0&I\end{bmatrix}
\begin{bmatrix}I&0\\
0&-I\end{bmatrix}
\begin{bmatrix}I&0\\
-I&I\end{bmatrix},\qquad k\in\mathbb R,
\end{equation}
\begin{equation}
\label{3.33}
q_{10}:=\begin{bmatrix}
I&L_2(-k)
\\
\noalign{\medskip}
0&I\end{bmatrix}\begin{bmatrix}
I-L_2(-k)\,L_2(k)&0
\\
\noalign{\medskip}
0&I\end{bmatrix},\qquad k\in\mathbb R,
\end{equation}
\begin{equation}
\label{3.34}
q_{11}:=\begin{bmatrix}
I&0
\\
\noalign{\medskip}
L_2(k)&I\end{bmatrix}\begin{bmatrix}
T_{\rm{l}2}(k)^{-1}&0
\\
\noalign{\medskip}
0&T_{\rm{l}2}(-k)^{-1}\end{bmatrix},\qquad k\in\mathbb R\setminus\{0\}.
\end{equation}
In fact, we have
\begin{equation}
\label{3.35}
q_9=\begin{bmatrix}I&I\\
ikI&-ikI\end{bmatrix},\qquad k\in\mathbb R,
\end{equation}
\begin{equation}
\label{3.36}
q_{10}\,q_{11}=\Lambda_2(k), \qquad k\in\mathbb R\setminus\{0\},
\end{equation}
where $\Lambda_2(k)$ is the transition matrix given in \eqref{3.26} 
corresponding to the right potential fragment $V_2.$

\item[\text{\rm(d)}] The matrix $F_{\rm{r}}(k,0)$ appearing in \eqref{3.30} can be written as a matrix product as
\begin{equation}
\label{3.37}
F_{\rm{r}}(k,0)=q_9\, q_{12}\,q_{13},\qquad k\in\mathbb R,
\end{equation}
where $q_9$ is the matrix defined in \eqref{3.32} and we have let
\begin{equation}
\label{3.38}
q_{12}:=\begin{bmatrix}
I&R_1(k)
\\
\noalign{\medskip}
0&I\end{bmatrix}\begin{bmatrix}
I-R_1(k)\,R_1(-k)&0
\\
\noalign{\medskip}
0&I\end{bmatrix},\qquad k\in\mathbb R,
\end{equation}
\begin{equation}
\label{3.39}
q_{13}:=\begin{bmatrix}
I&0
\\
\noalign{\medskip}
R_1(-k)&I\end{bmatrix}\begin{bmatrix}
T_{\rm{r}1}(-k)^{-1}&0
\\
\noalign{\medskip}
0&T_{\rm{r}1}(k)^{-1}\end{bmatrix},\qquad k\in\mathbb R\setminus\{0\}.
\end{equation}
In fact, we have
\begin{equation}
\label{3.40}
\quad q_{12}\,q_{13}=\Sigma_1(k),\qquad k\in\mathbb R\setminus\{0\},
\end{equation}
where $\Sigma_1(k)$ is 
the transition matrix in \eqref{3.27} corresponding to the left potential fragment $V_1.$

\end{enumerate}

\end{proposition}

\begin{proof}
We remark that, although the matrices $q_{11}$ and $q_{13}$ behave as $O(1/k)$ as $k\to 0,$ 
the equalities in
\eqref{3.31} and \eqref{3.37} hold also at $k=0$ by the continuity.
Let us use $f_{\rm{l}1}(k,x)$ and $f_{\rm{r}1}(k,x)$ to denote the left and right Jost solutions, respectively, 
corresponding to the potential fragment $V_1.$ Similarly, let us use
$f_{\rm{l}2}(k,x)$ and $f_{\rm{r}2}(k,x)$ to denote the left and right Jost solutions, respectively, corresponding to the potential
fragment $V_2.$ Since $f_{\rm{l}}(k,x)$ and $f_{\rm{l}2}(k,x)$ both satisfy \eqref{1.1} and the asymptotics \eqref{2.1}, we have
\begin{equation}
\label{3.41}
f_{\rm{l}2}(k,x)=f_{\rm{l}}(k,x),\quad f'_{\rm{l}2}(k,x)=f'_{\rm{l}}(k,x),\qquad x\in[0,+\infty).
\end{equation}
Similarly, since $f_{\rm{r}}(k,x)$ and $f_{\rm{r}1}(k,x)$ both satisfy \eqref{1.1} and the asymptotics \eqref{2.2}, we have
\begin{equation}
\label{3.42}
f_{\rm{r}1}(k,x)=f_{\rm{r}}(k,x),\quad f'_{\rm{r}1}(k,x)=f'_{\rm{r}}(k,x),\qquad x\in(-\infty,0].
\end{equation}
From \eqref{1.5}, \eqref{2.3}, \eqref{2.4}, \eqref{3.41}, and \eqref{3.42}, we get
\begin{equation}
\label{3.43}
f_{\rm{l}2}(k,x)=e^{ikx}\, T_{\rm{l}2}(k)^{-1}+e^{-ikx} L_2(k)\,T_{\rm{l}2}^{-1}(k), \qquad x\in(-\infty,0],
\end{equation}
\begin{equation}
\label{3.44}
f_{\rm{r}1}(k,x)=e^{-ikx}\, T_{\rm{r}1}(k)^{-1}+e^{ikx} R_1(k)\,T_{\rm{r}1}^{-1}(k), \qquad x\in[0,+\infty).
\end{equation}
Comparing \eqref{3.41} and \eqref{3.43} at $x=0$ and using the result in \eqref{2.12},
we establish \eqref{3.29}, which completes the proof of (a).
Similarly, comparing  
\eqref{3.42} and \eqref{3.44} at $x=0$ and using the result in \eqref{2.13},
we establish \eqref{3.30}. Hence, the proof of (b) is also complete.
We can confirm \eqref{3.31} directly by evaluating the matrix product
on its right-hand side with the help of \eqref{3.32}--\eqref{3.34}. Similarly,
\eqref{3.35} can directly be confirmed by evaluating the matrix products
on the right-hand sides of \eqref{3.32}--\eqref{3.34}. Thus, the proof of (c) is complete.
Let us now prove (d).
We can verify \eqref{3.37} directly by evaluating the matrix product
on its right-hand side with the help of  \eqref{3.35}, \eqref{3.38}, and \eqref{3.39}.
In the same manner,
\eqref{3.40} can be directly confirmed by evaluating the matrix products
on the right-hand sides of \eqref{3.38} and \eqref{3.39} and by comparing the result with
 \eqref{3.27}. Hence, the proof of (d) is also complete.
\end{proof}

In the next theorem we present our factorization formula corresponding to the potential fragmentation given in \eqref{1.4}.

\begin{theorem}
\label{theorem3.3} Consider the full-line $n\times n$ matrix Schr\"odinger equation \eqref{1.1} with the
potential $V$ satisfying \eqref{1.2} and \eqref{1.3}. Let $V_1$ and $V_2$ denote
the left and right fragments of $V$ described in \eqref{1.4} and \eqref{1.5}. 
Let $\Lambda(k),$ $\Lambda_1(k),$ and $\Lambda_2(k)$ be the $2n\times 2n$
transition matrices defined in \eqref{2.20}, \eqref{3.25}, and \eqref{3.26}
corresponding to $V,$ $V_1,$ and $V_2,$ respectively.
Similarly, let $\Sigma(k),$ $\Sigma_1(k),$ and $\Sigma_2(k)$ denote the $2n\times 2n$
transition matrices defined in \eqref{2.21}, \eqref{3.27}, and \eqref{3.28}
corresponding to $V,$ $V_1,$ and $V_2,$ respectively.
Then, we have the following:

\begin{enumerate}

\item[\text{\rm(a)}]  The transition matrix $\Lambda(k)$ is equal to the ordered matrix product 
$\Lambda_1(k)\,\Lambda_2(k),$ i.e. we have
\begin{equation}
\label{3.45}\Lambda(k)=\Lambda_1(k)\,\Lambda_2(k),\qquad k\in\mathbb R\setminus\{0\}.
\end{equation}
\item[\text{\rm(b)}] The factorization formula \eqref{3.45} can also be expressed in terms of the transition matrices 
$\Sigma(k),$ $\Sigma_1(k),$ and $\Sigma_2(k)$ as
 \begin{equation}
\label{3.46}
\Sigma(k)=\Sigma_2(k)\,\Sigma_1(k),\qquad k\in\mathbb R\setminus\{0\}.
\end{equation}

\end{enumerate}

\end{theorem}

\begin{proof}
Evaluating \eqref{3.5} at $x=0,$ we get
\begin{equation}
\label{3.47}
F_{\rm{l}}(k,0)=F_{\rm{r}}(k,0)\,\Lambda(k),\qquad k\in\mathbb R.
\end{equation}
Using \eqref{3.31} and \eqref{3.37} on the left and right-hand sides of \eqref{3.47},
respectively, we obtain
\begin{equation}
\label{3.48}
q_9\,q_{10}\,q_{11}=q_9\,q_{12}\,q_{13}\,\Lambda(k),\qquad k\in\mathbb R\setminus\{0\}.
\end{equation}
 From \eqref{3.35} we see that $q_9$ is invertible when $k\in\mathbb R\setminus\{0\}.$ Thus, using 
 \eqref{3.36} and \eqref{3.40} in \eqref{3.48}, we get
\begin{equation*}
\Lambda_2(k)=\Sigma_1(k)\,\Lambda(k),\qquad k\in\mathbb R\setminus\{0\},
\end{equation*}
or equivalently 
 \begin{equation}
\label{3.50}
\Sigma_1(k)^{-1}\,\Lambda_2(k)=\Lambda(k),\qquad k\in\mathbb R\setminus\{0\}.
\end{equation}
From \eqref{3.7}
we already know that $\Sigma_1(k)^{-1}=\Lambda_1(k),$ and hence \eqref{3.50} yields \eqref{3.45}.
Thus, the proof of (a) is complete.
Taking the matrix inverses of both sides of \eqref{3.45} and then making use of \eqref{3.7}, we
obtain \eqref{3.46}. Thus, the proof of (b) is also complete.
\end{proof}

The following theorem shows that the factorization formulas \eqref{3.45} and \eqref{3.46}
also hold if the potential $V$ in \eqref{1.1} is decomposed into $V_1$ and $V_2$
by choosing the fragmentation point anywhere on the real axis, not necessarily at $x=0.$

\begin{theorem}
\label{theorem3.4} Consider the full-line $n\times n$ matrix Schr\"odinger equation \eqref{1.1} with the
potential $V$ satisfying \eqref{1.2} and \eqref{1.3}. Let $V_1$ and $V_2$ denote
the left and right fragments of $V$ described as in \eqref{1.4}, but \eqref{1.5} replaced with
\begin{equation}
\label{3.51}
V_1(x):=
\begin{cases}V(x),\qquad x < b,\\
\noalign{\medskip}
0,\qquad x>b,
\end{cases}
\quad
V_2(x):=
\begin{cases}0,\qquad x < b,\\
\noalign{\medskip}
V(x),\qquad x>b,  
\end{cases}
\end{equation}
where $b$ is a fixed real constant.
Let $\Lambda(k),$ $\Lambda_1(k),$ and $\Lambda_2(k)$ appearing 
 in \eqref{2.20}, \eqref{3.25}, \eqref{3.26}, respectively, be the
 transition matrices 
corresponding to $V,$ $V_1,$ and $V_2,$ respectively.
Similarly, let $\Sigma(k),$ $\Sigma_1(k),$ and $\Sigma_2(k)$ appearing 
 in \eqref{2.21}, \eqref{3.27}, \eqref{3.28}
 be the
 transition matrices 
corresponding to $V,$ $V_1,$ and $V_2,$ respectively.
Then, we have the following:

\begin{enumerate}

\item[\text{\rm(a)}]  The transition matrix $\Lambda(k)$ is equal to the ordered matrix product 
$\Lambda_1(k)\,\Lambda_2(k),$ i.e. we have
\begin{equation}
\label{3.52}\Lambda(k)=\Lambda_1(k)\,\Lambda_2(k),\qquad k\in\mathbb R\setminus\{0\}.
\end{equation}
\item[\text{\rm(b)}] The factorization formula \eqref{3.52} can also be expressed in terms of the transition matrices 
$\Sigma(k),$ $\Sigma_1(k),$ and $\Sigma_2(k)$ as
 \begin{equation}
\label{3.53}
\Sigma(k)=\Sigma_2(k)\,\Sigma_1(k),\qquad k\in\mathbb R\setminus\{0\}.
\end{equation}

\end{enumerate}

\end{theorem}

\begin{proof} Let us translate the potential $V$ and its fragments $V_1$ and $V_2$ as in \eqref{3.15}. The
shifted potentials satisfy
\begin{equation*}
V^{(b)}(x)= V_1^{(b)}(x)+ V_2^{(b)}(x), \qquad x \in \mathbb R,
\end{equation*}
where we have defined
\begin{equation*}
V^{(b)}(x):= V(x+b),\quad
V_1^{(b)}(x):= V_1(x+b),\quad
V_2^{(b)}(x):= V_2(x+b), \qquad x \in \mathbb R.
\end{equation*}
Note that the shifted potential fragments $V_1^{(b)}$ and
$V_2^{(b)}$ correspond to the fragments of
$V^{(b)}$ with the fragmentation point $x=0,$ i.e. we have
\begin{equation*}
V_1^{(b)}(x)=
\begin{cases}V^{(b)}(x),\qquad x < 0,\\
\noalign{\medskip}
0,\qquad x>0,
\end{cases}
\quad
V_2^{(b)}(x)= 
\begin{cases}0,\qquad x < 0,\\
\noalign{\medskip}
V^{(b)}(x),\qquad x>0.
\end{cases}
\end{equation*}
Since the shifted potential $V^{(b)}$ is fragmented at $x=0,$ we can apply Theorem~\ref{theorem3.3} to $V^{(b)}.$
Let us use $T_{\rm{l}j}^{(b)}(k),$ $T_{\rm{r}j}^{(b)}(k),$ $L_j^{(b)}(k),$ $R_j^{(b)}(k),$
$\Lambda_j^{(b)}(k),$ and $\Sigma_j^{(b)}(k)$ to denote
the left and right transmission coefficients,  the left and right
reflections coefficients, and  the transition matrices,  respectively,
for the shifted potentials $V^{(b)}_j(x)$ with $j=1$ and $j=2.$
In analogy with \eqref{3.25}--\eqref{3.28}, for $k\in\mathbb R\setminus\{0\}$ we have
\begin{equation*}
\Lambda_j^{(b)}(k):=\begin{bmatrix}[T_{\rm{l}j}^{(b)}(k)]^{-1}&L_j^{(b)}(-k) \,[T_{\rm{l}j}^{(b)}(-k)]^{-1} \\
\noalign{\medskip}
L_j^{(b)}(k)\,[T^{(b)}_{\rm{l}j}(k)]^{-1}&[T_{\rm{l}j}^{(b)}(-k)]^{-1}\end{bmatrix},\qquad j=1,2,
\end{equation*}
\begin{equation*}
\Sigma_j^{(b)}(k):=\begin{bmatrix}[T_{\rm{r}j}^{(b)}(-k)]^{-1}&R_j^{(b)}(k)\,[T_{\rm{r}j}^{(b)}(k)]^{-1} \\
\noalign{\medskip}
R_j^{(b)}(-k)\,[T_{\rm{r}j}^{(b)}(-k)]^{-1}&[T_{\rm{r}j}^{(b)}(k)]^{-1}\end{bmatrix},\qquad j=1,2.
\end{equation*}
From Theorem~\ref{theorem3.3} we have
\begin{equation}
\label{3.59}\Lambda^{(b)}(k)=\Lambda_1^{(b)}(k)\,\Lambda_2^{(b)}(k),\quad
\Sigma^{(b)}(k)=\Sigma_2^{(b)}(k)\,\Sigma_1^{(b)}(k),
\qquad k\in\mathbb R\setminus\{0\}.
\end{equation}
Using \eqref{3.19} and \eqref{3.20} in \eqref{3.59}, after
some minor simplification we get 
\eqref{3.52} and \eqref{3.53}.
\end{proof}

The result of Theorem~\ref{theorem3.4} can easily be extended from two fragments to any finite number of fragments. This is because any existing fragment
can be decomposed into further subfragments by applying the factorization formulas \eqref{3.52} and \eqref{3.53}
to each fragment and to its subfragments. Since a proof can be obtained by using an induction on the number of fragments,
we state the result as a corollary without a proof.

\begin{corollary}
\label{corollary3.5} 
Consider the full-line $n\times n$ matrix Schr\"odinger equation \eqref{1.1} with the matrix-valued
potential $V$ satisfying \eqref{1.2} and \eqref{1.3}. Let $S(k),$ $\Lambda(k),$ and 
$\Sigma(k)$ defined in \eqref{2.5}, \eqref{2.20}, and \eqref{2.21}, respectively
be the corresponding scattering matrix and the transition matrices,
with $T_{\rm{l}}(k),$ $T_{\rm{r}}(k),$ $L(k),$ and $R(k)$
denoting the corresponding left and right transmission coefficients and the left and right reflection coefficients, respectively.
Assume that $V$ is partitioned into $P+1$ fragments 
$V_j$ at the fragmentation points $b_j$ with $1\le j\le P$ as
\begin{equation*}
V(x)= \sum_{j=1}^{P+1} V_j(x),
\end{equation*}
where $P$ is any fixed positive integer and 
\begin{equation*}
V_j(x):=
\begin{cases}V(x),\qquad x \in (b_{j-1}, b_j),\\
\noalign{\medskip}
0,\qquad x\notin  (b_{j-1}, b_ j), 
\end{cases}
\end{equation*}
with $b_0:=-\infty,$ $b_{P+1}:=+\infty,$ and $b_j<b_{j+1}.$ Let
$T_{\rm{l}j}(k),$
$T_{\rm{r}j}(k),$
$L_j(k),$
$R_j(k)$ be the corresponding 
left and right transmission coefficients and the left and right reflection coefficients, respectively,
for the potential fragment $V_j.$ Let 
$S_j(k),$
$\Lambda_j(k),$ and
$\Sigma_j(k)$ denote the  
scattering matrix and the transition matrices for the corresponding
fragment $V_j,$ which are defined as
\begin{equation*}
S_{j}(k):=\begin{bmatrix}T_{\rm{l}j}(k)&R_j(k) \\ 
\noalign{\medskip}
 L_{ j }(k)& T_{\rm{r} j}(k)\end{bmatrix}, \qquad k\in\mathbb R,
\end{equation*}
\begin{equation*}
\Lambda_j(k):=\begin{bmatrix}T_{\rm{l}j}(k)^{-1}&L_j(-k)\,T_{\rm{l}j}(-k)^{-1} \\
\noalign{\medskip}
L_j(k)\,T_{\rm{l}j}(k)^{-1}&T_{\rm{l}j}(-k)^{-1}\end{bmatrix},\qquad k\in\mathbb R \setminus\{0\},
\end{equation*}
\begin{equation*}
\Sigma_j(k):=\begin{bmatrix}T_{\rm{r}j}(-k)^{-1}&R_j(k)\,T_{\rm{r}j}(k)^{-1} \\
\noalign{\medskip}
R_j(-k)\,T_{\rm{r}j}(-k)^{-1}&T_{\rm{r}j}(k)^{-1}\end{bmatrix},\qquad k\in\mathbb R \setminus\{0\}.
\end{equation*}
Then, the transition matrices $\Lambda(k)$ and $\Sigma(k)$ for the whole potential $V$
are expressed as ordered matrix products of the corresponding
transition matrices for the fragments as
\begin{equation*}
\Lambda(k)=\Lambda_1(k)\, \Lambda_2(k)\cdots\Lambda_P(k)\,\Lambda_{P+1}(k),\qquad k\in\mathbb R\setminus\{0\},
\end{equation*}
\begin{equation*}
\Sigma(k)=\Sigma_{P+1}(k)\, \Sigma_P(k)\cdots\Sigma_2(k)\,\Sigma_1(k),\qquad k\in\mathbb R \setminus\{0\}.
\end{equation*}
\end{corollary}

In Theorem~\ref{theorem3.4} the transition matrix for a potential on the full line is expressed as a matrix product of
the transition matrices for the left and right potential fragments.
In the next theorem, we express the scattering coefficients 
of a potential on the full line in terms of the scattering coefficients of the two potential fragments.

\begin{theorem}
\label{theorem3.6} Consider the full-line $n\times n$ matrix Schr\"odinger equation \eqref{1.1} with the
potential $V$ satisfying \eqref{1.2} and \eqref{1.3}. Assume that $V$ is fragmented at an arbitrary point $x=b$ into the two pieces $V_1$ and $V_2$
as described in \eqref{1.4} and \eqref{3.51}. Let $S(k)$ given in \eqref{2.5} be the scattering matrix for the potential $V,$ and let $S_1(k)$ and $S_2(k)$
given in \eqref{3.24} be the scattering matrices corresponding to the potential fragments $V_1$ and
$V_2,$ respectively. Then, for $k\in\mathbb R$
the scattering coefficients in $S(k)$ are related to the right scattering coefficients
in $S_1(k)$ and the left scattering coefficients in $S_2(k)$ as
\begin{equation}
\label{3.67}
T_{\rm{l}}(k)=T_{\rm{l}2}(k)\left[I-R_1(k)\,L_2(k)\right]^{-1} T_{\rm{r}1}(-k)^\dagger,
\end{equation}
\begin{equation}
\label{3.68}
L(k)=[T_{\rm{r}1}(k)^\dagger]^{-1}\left[L_2(k)-R_1(-k)\right]\left[I-R_1(k)\,L_2(k)\right]^{-1} T_{\rm{r}1}(-k)^\dagger,
\end{equation}
\begin{equation}
\label{3.69}
T_{\rm{r}}(k)=T_{\rm{r}1}(k)\left[I-L_2(k)\,R_1(k)\right]^{-1} T_{\rm{l}2}(-k)^\dagger,
\end{equation}
\begin{equation}
\label{3.70}
R(k)=T_{\rm{l}2}(k)\left[I-R_1(k)\,L_2(k)\right]^{-1}  \left[R_1(k)-L_2(-k)\right] T_{\rm{l}2}(-k)^{-1},
\end{equation}
where we recall that $I$ is the $n\times n$ identity matrix and the dagger denotes the matrix adjoint.
\end{theorem}

\begin{proof}
From the $(1,1)$ entry in \eqref{3.52}, for 
$k\in\mathbb R\setminus\{0\}$
 we have
 \begin{equation}
\label{3.71}T_{\rm{l}}(k)^{-1}=T_{\rm{l}1}(k)^{-1} \,T_{\rm{l}2}(k)^{-1} +L_1(-k)\,T_{\rm{l}1}(-k)^{-1}\,L_2(k)\,T_{\rm{l}2}(k)^{-1}.
\end{equation}
 From the $(2,1)$ entry of \eqref{2.52} we know that 
\begin{equation}
\label{3.72}
R_1(k)^\dagger\,T_{\rm{l}1}(k)+ T_{\rm{r}1}(k)^\dagger\,L_1(k)=0,\qquad k\in\mathbb R.
\end{equation}
Multiplying both sides of \eqref{3.72} by $[T_{\rm{r}1}^\dagger(k)]^{-1}$ from the left and 
by $T_{\rm{l}1}(k)^{-1}$ from the right, we get
\begin{equation}\label{3.73}
[T_{\rm{r}1}(k)^\dagger]^{-1} R_1^\dagger(k) + L_1(k) \,T_{\rm{l}1}(k)^{-1}=0, \qquad k \in \mathbb R\setminus\{0\}.
\end{equation}
In \eqref{3.73}, after replacing $k$ by $-k$ we obtain
\begin{equation}\label{3.74}
L_1(-k) \, T_{\rm{l}1}(-k)^{-1}=-[T_{\rm{r}1}(-k)^\dagger]^{-1} R_1(-k)^\dagger, \qquad k \in \mathbb R\setminus\{0\}.
\end{equation}
Next, using \eqref{3.74} and the second equality of \eqref{2.49} in \eqref{3.71}, for $k\in\mathbb R\setminus\{0\}$ we get
\begin{equation}
\label{3.75}T_{\rm{l}}(k)^{-1}=T_{\rm{l}1}(k)^{-1} \,T_{\rm{l}2}(k)^{-1}-[T_{\rm{r}1}(-k)^\dagger]^{-1} 
 R_1(k)\,L_2(k)\,T_{\rm{l}2}(k)^{-1}.
\end{equation}
Using the third equality of \eqref{2.49} in the first term
on the right-hand side of \eqref{3.75}, for $k\in\mathbb R\setminus\{0\}$ we have
\begin{equation}
\label{3.76}T_{\rm{l}}(k)^{-1}=[T_{\rm{r}1}(-k)^\dagger]^{-1} 
T_{\rm{l}2}(k)^{-1}  -[T_{\rm{r}1}(-k)^\dagger]^{-1}  R_1(k)\,L_2(k)\,T_{\rm{l}2}(k)^{-1}.
\end{equation}
Factoring the right-hand side of \eqref{3.76} and then taking the inverses of both sides
of the resulting equation,
we obtain \eqref{3.67}.
Let us next prove \eqref{3.68}. From the $(2,1)$ entry of \eqref{3.52}, for $k\in\mathbb R\setminus\{0\}$ we have
\begin{equation}\label{3.77}
L(k) \,T_{\rm{l}}(k)^{-1}= L_1(k) \,T_{\rm{l}1}(k)^{-1}\, T_{\rm{l}2}(k) ^{-1} +T_{\rm{l}1}(-k)^{-1}\,L_2(k)\,
T_{\rm{l}2}(k)^{-1}.
\end{equation}
In \eqref{3.74} by replacing $k$ by $-k$ and using the resulting equality in \eqref{3.77}, when $k\in\mathbb R\setminus\{0\}$ we get
\begin{equation}\label{3.78}
L(k) \,T_{\rm{l}}(k)^{-1}= -[T_{\rm{r}1}(k)^\dagger]^{-1} R_1(k)^\dagger\, T_{\rm{l}2}(k) ^{-1} +T_{\rm{l}1}(-k)^{-1}\,L_2(k)\,
T_{\rm{l}2}(k)^{-1}.
\end{equation}
Next, using the third equality of \eqref{2.49} in the second term on the right-hand side of \eqref{3.78}, for $k\in\mathbb R\setminus\{0\}$ we obtain
\begin{equation*}
L(k) \,T_{\rm{l}}(k)^{-1}= -[T_{\rm{r}1}(k)^\dagger]^{-1} R_1(k)^\dagger\, T_{\rm{l}2}(k) ^{-1} +[T_{\rm{r}1}(k)^\dagger]^{-1}L_2(k)\,
T_{\rm{l}2}(k)^{-1},
\end{equation*}
which is equivalent to
\begin{equation}\label{3.80}
L(k) \,T_{\rm{l}}(k)^{-1}=[T_{\rm{r}1}(k)^\dagger]^{-1} \left[L_2(k)-R_1(k)^\dagger\right]
T_{\rm{l}2}(k)^{-1},  \qquad k\in\mathbb R\setminus\{0\}.
\end{equation}
Using the second equality of \eqref{2.49} in \eqref{3.80}, we have
\begin{equation}\label{3.81}
L(k) \,T_{\rm{l}}(k)^{-1}=[T_{\rm{r}1}(k)^\dagger]^{-1} \left[L_2(k)-R_1(-k)\right]
T_{\rm{l}2}(k)^{-1},  \qquad k\in\mathbb R\setminus\{0\}.
\end{equation}
Then, multiplying \eqref{3.81} from the right by the respective sides of \eqref{3.67}, we obtain \eqref{3.68}.
Let us now prove \eqref{3.69}. From the third equality of \eqref{2.49} we have
$T_{\rm{r}}(k)=T_{\rm{l}}(-k)^\dagger.$ Thus, by taking the matrix adjoint of both sides of \eqref{3.67}, 
then replacing $k$ by $-k$ in the resulting equality, and then using the first two
equalities of \eqref{2.49}, we get \eqref{3.69}.
Let us finally prove \eqref{3.70}. From the $(1,2)$ entry of \eqref{2.53}, we have
\begin{equation}
\label{3.82}
R(k)=-T_{\rm{l}}(k)\,L(k)^\dagger\,[T_{\rm{r}}(k)^\dagger]^{-1},\qquad k\in\mathbb R.
\end{equation}
Using the first and third equalities of \eqref{2.49} in \eqref{3.82}, we get
\begin{equation}
\label{3.83}R(k)=-T_{\rm{l}}(k)\,L(-k)\,T_{\rm{l}}(-k)^{-1},\qquad k\in\mathbb R.
\end{equation}
Then, on the right-hand side of \eqref{3.83}, we replace $T_{\rm{l}}(k)$ by the right-hand side of \eqref{3.67} and 
we also replace $L(-k)\,[T_{\rm{l}}(-k)]^{-1}$ by the right-hand side of \eqref{3.81} after the substitution $k\mapsto -k.$
After simplifying the resulting modified version of \eqref{3.83}, we 
obtain \eqref{3.70}.
\end{proof}

\section{The matrix-valued scattering coefficients}
\label{section4}

The choice $n=1$ in \eqref{1.1} corresponds to the scalar case. In the scalar case, 
the potential $V$ satisfying \eqref{1.2} and \eqref{1.3} is real valued and 
the corresponding left and right transmission coefficients
are equal to each other. However, when $n\ge 2$ the matrix-valued
transmission coefficients are in general not equal to each other. In the matrix case, as seen from \eqref{2.45} 
we have 
\begin{equation}
\label{4.1}
T_{\rm{l}}(-k^\ast)^\dagger=T_{\rm{r}}(k),\qquad k\in\overline{\mathbb C^+}.
\end{equation}
On the other hand, from \eqref{2.18} we know that
the determinants of the left and right transmission coefficients are always equal to each other for
$n\ge 1.$
In this section, we first provide some relevant properties of the scattering coefficients
for \eqref{1.1}, and then we present some explicit examples demonstrating the unequivalence of
the matrix-valued left and right transmission coefficients.

The following theorem summarizes the large $k$-asymptotics of the scattering
coefficients for the full-line matrix Schr\"odinger equation. 

\begin{theorem}
\label{theorem4.1} 
Consider the full-line matrix Schr\"odinger equation \eqref{1.1} with the $n\times n$ matrix
potential $V$ satisfying \eqref{1.2} and \eqref{1.3}.  The large $k$-asymptotics of the corresponding 
$n\times n$ matrix-valued scattering coefficients are given by
\begin{equation}
\label{4.2}
T_{\rm{l}}(k)=I+\displaystyle\frac{1}{2ik}\displaystyle\int_{-\infty}^\infty dx\, V(x)+O\left(\displaystyle\frac{1}{k^2}\right),\qquad k\to\infty
\text{\rm{ in }}\overline{\mathbb C^+},
\end{equation}
\begin{equation}
\label{4.3}
T_{\rm{r}}(k)=I+\displaystyle\frac{1}{2ik}\displaystyle\int_{-\infty}^\infty dx\, V(x)+O\left(\displaystyle\frac{1}{k^2}\right),\qquad k\to\infty
\text{\rm{ in }}\overline{\mathbb C^+},
\end{equation}
\begin{equation}
\label{4.4}
L(k)=\displaystyle\frac{1}{2ik}\displaystyle\int_{-\infty}^\infty dx\, V(x)\,e^{2ikx} +O\left(\displaystyle\frac{1}{k^2}\right),\qquad k\to\pm\infty
\text{\rm{ in }} \mathbb R,
\end{equation}
\begin{equation}
\label{4.5}
R(k)=\displaystyle\frac{1}{2ik}\displaystyle\int_{-\infty}^\infty dx\, V(x)\,e^{-2ikx} +O\left(\displaystyle\frac{1}{k^2}\right),\qquad k\to\pm\infty
\text{\rm{ in }} \mathbb R.
\end{equation}

\end{theorem}

\begin{proof}
By solving \eqref{4.6} and \eqref{4.7}
via the method of successive approximation, for each fixed $x\in\mathbb R$ we obtain the large $k$-asymptotics of the 
Jost solutions as
\begin{equation}
\label{4.8}
e^{-ikx}\,f_{\rm{l}}(k,x)=I+O\left(\displaystyle\frac{1}{k}\right),\qquad k\to\infty
\text{ \rm{in} }\overline{\mathbb C^+},
\end{equation}
\begin{equation}
\label{4.9}
e^{ikx}f_{\rm{r}}(k,x)=I+O\left(\displaystyle\frac{1}{k}\right),\qquad k\to\infty
\text{ \rm{in} }\overline{\mathbb C^+}.
\end{equation}
The integral representations involving the scattering coefficients are obtained
from \eqref{4.6} and \eqref{4.7} with the help of 
\eqref{2.3} and \eqref{2.4}. As listed in (2.10)--(2.13) of
\cite{AKV2001}, we have
\begin{equation}
\label{4.10}
T_{\rm{l}}(k)^{-1}=I-\displaystyle\frac{1}{2ik}\displaystyle\int_{-\infty}^\infty dx\, V(x)\,e^{-ikx}f_{\rm{l}}(k,x),
\end{equation}
\begin{equation}
\label{4.11}
T_{\rm{r}}(k)^{-1}=I-\displaystyle\frac{1}{2ik}\displaystyle\int_{-\infty}^\infty dx\, V(x)\,e^{ikx}f_{\rm{r}}(k,x),
\end{equation}
\begin{equation}
\label{4.12}
L(k)\,T_{\rm{l}}(k)^{-1}=\displaystyle\frac{1}{2ik}\displaystyle\int_{-\infty}^\infty dx\, V(x)\,e^{ikx}f_{\rm{l}}(k,x),
\end{equation}
\begin{equation}
\label{4.13}
R(k)\,T_{\rm{r}}(k)^{-1}=\displaystyle\frac{1}{2ik}\displaystyle\int_{-\infty}^\infty dx\, V(x)\,e^{-ikx}f_{\rm{r}}(k,x).
\end{equation}
With the help of \eqref{4.8} and \eqref{4.9}, from \eqref{4.10}--\eqref{4.13}
we obtain \eqref{4.2}--\eqref{4.5} in their appropriate domains. 
\end{proof}

We observe that it is impossible to
tell the unequivalence of $T_{\rm{l}}(k)$ and $T_{\rm{r}}(k)$ from the large $k$-limits
given in \eqref{4.2} and \eqref{4.3}. However, as the next example shows,
the small $k$-asymptotics
of $T_{\rm{l}}(k)$ and $T_{\rm{r}}(k)$ may be used to see their unequivalence in the matrix
case with $n\ge 2.$

\begin{example}
\label{example4.2}
\normalfont
Consider the full-line Schr\"odinger equation \eqref{1.1}
when the potential $V$ is a $2\times 2$ matrix and fragmented as in \eqref{1.4} and \eqref{1.5},
where the fragments $V_1$ and $V_2$ are compactly supported and given by
\begin{equation}
\label{4.14}
\begin{cases}
V_1(x)=\begin{bmatrix}3&-2+i\\
-2-i&-5\end{bmatrix},\qquad -2<x<0,\\
\noalign{\medskip}
V_2(x)=\begin{bmatrix}2&1+i\\
1-i&-2\end{bmatrix},\qquad 0<x<1,\end{cases}
\end{equation}
with the understanding that $V_1$ vanishes when
$x\not\in(-2,0)$ and that $V_2$ vanishes
when $x\not\in(0,1).$
Thus, the support of the potential $V$ is confined to the interval $(-2,1).$
Since $V_2$ vanishes when $x>1,$ as seen from \eqref{2.1}
the corresponding Jost solution $f_{\rm{l}2}(k,x)$ is equal to $e^{ikx}I$ there.
The evaluation of the corresponding scattering coefficients for the potential specified in \eqref{4.14} is not
trivial. A relatively efficient way for that evaluation is accomplished as follows.
For $0<x<1$ we construct the Jost solution $f_{\rm{l}2}(k,x)$
corresponding to the constant $2\times 2$ matrix $V_2$ by diagonalizing $V_2,$ obtaining the 
corresponding eigenvalues and eigenvectors of $V_2,$ 
and then constructing the general solution to \eqref{1.1} with the help of those eigenvalues and eigenvectors.
Next, by using the continuity of $f_{\rm{l}2}(k,x)$ and $f'_{\rm{l}2}(k,x)$ at the point $x=1,$ we construct
$f_{\rm{l}2}(k,x)$ explicitly when $0<x<1.$ Then, by using \eqref{3.43} and its derivative at $x=0$ we obtain
the corresponding $2\times 2$ scattering coefficients $T_{\rm{l}2}(k)$ and $L_2(k).$ By using a similar procedure, we construct
the right Jost solution
$f_{\rm{r}1}(k,x)$ corresponding to the potential $V_1$ as well as the corresponding 
$2\times 2$ scattering coefficients $T_{\rm{r}1}(k)$ and $R_1(k).$
Next, we construct the
transmission coefficients $T_{\rm{l}}(k)$ and $T_{\rm{r}}(k)$
corresponding to the whole potential $V$ by using \eqref{3.67}
and \eqref{3.69}, respectively.
In fact, with the help of the symbolic software Mathematica, we evaluate
$T_{\rm{l}}(k)$ and $T_{\rm{r}}(k)$ explicitly in a closed form.
However, the corresponding explicit expressions are extremely lengthy, and hence it is not
feasible to display them here.
Instead, we present the small $k$-limits of
$T_{\rm{l}}(k)$ and $T_{\rm{r}}(k),$ which shows that
$T_{\rm{l}}(k)\not\equiv T_{\rm{r}}(k).$ We have
\begin{equation}
\label{4.16}
T_{\rm{l}}(k)^{-1}=\begin{bmatrix}a_{\rm{l}}(k)&b_{\rm{l}}(k)\\
\noalign{\medskip}
c_{\rm{l}}(k)&d_{\rm{l}}(k)\end{bmatrix}+O(k^3),\qquad k\to 0 \text{\rm{ in }} \overline{\mathbb C^+},
\end{equation}
\begin{equation}
\label{4.17}
T_{\rm{r}}(k)^{-1}=\begin{bmatrix}a_{\rm{r}}(k)&b_{\rm{r}}(k)\\
\noalign{\medskip}
c_{\rm{r}}(k)&d_{\rm{r}}(k)\end{bmatrix}+O(k^3),\qquad k\to 0 \text{\rm{ in }} \overline{\mathbb C^+},
\end{equation}
where we have defined
\begin{equation}
\label{4.18}
\begin{split}
a_{\rm{l}}(k):=&-\displaystyle\frac{7.0565\overline{2} -74.135\overline{2}\,i}{k} - \left(133.84\overline{4} + 
   14.352\overline{2}\,i\right) \\
   &+ \left(19.075\overline{6} - 170.82\overline{7} \,i\right) k + \left(160.95\overline{5} + 18.172\overline{9}\,i\right) k^2,
   \end{split}
\end{equation}
\begin{equation}
\label{4.19}
\begin{split}
b_{\rm{l}}(k):=&-\displaystyle\frac{19.983\overline{9} -22.417\overline{6}\,i}{k} - \left(50.518\overline{8} + 39.089\overline{3}\,i\right) \\
   &+ \left(51.28\overline{9} - 68.661\overline{5}\,i\right) k + \left(65.051\overline{6} + 48.408\overline{9}\,i\right) k^2,
   \end{split}
\end{equation}
\begin{equation}
\label{4.20}
\begin{split}
c_{\rm{l}}(k):=&\displaystyle\frac{10.575\overline{2} - 15.217\overline{2} \,i}{k} + \left(26.026\overline{5} + 19.952\overline{9}\,i\right) \\
   &- \left(25.854\overline{8} -33.075\overline{2} \,i\right) k - \left(31.870\overline{5} + 24.178\overline{3} \,i\right) k^2,
   \end{split}
\end{equation}
\begin{equation}
\label{4.21}
\begin{split}
d_{\rm{l}}(k):=&\displaystyle\frac{7.0565\overline{2} - 2.5552\overline{8}\,i}{k} + \left(7.2733\overline{5} + 14.352\overline{2}\,i\right) \\
   &- \left(19.075\overline{6} - 11.165\overline{9} \,i\right) k - \left(10.490\overline{3} + 18.172\overline{9}\,i\right) k^2,
   \end{split}
\end{equation}
\begin{equation}
\label{4.22}
\begin{split}
a_{\rm{r}}(k):=&\displaystyle\frac{7.0565\overline{2} +74.135\overline{2}\,i}{k} - \left(133.84\overline{4} -
   14.352\overline{2}\,i\right) \\
   &- \left(19.075\overline{6} + 170.82\overline{7}\,i\right) k + \left(160.95\overline{5} - 18.172\overline{9}\,i\right) k^2,
   \end{split}
\end{equation}
\begin{equation}
\label{4.23}
\begin{split}
b_{\rm{r}}(k):=&-\displaystyle\frac{10.575\overline{2} +15.217\overline{2}\,i}{k} + \left(26.026\overline{5} -19.952\overline{9}\,i\right) \\
   &+\left(25.854\overline{8} +33.075\overline{2}\,i\right) k - \left(31.870\overline{5} - 24.178\overline{3}\,i\right) k^2,
   \end{split}
\end{equation}
\begin{equation}
\label{4.24}
\begin{split}
c_{\rm{r}}(k):=&\displaystyle\frac{19.983\overline{9} + 22.417\overline{6}\,i}{k} - \left(50.518\overline{8} - 39.089\overline{3}\,i\right) \\
   &-\left(51.28\overline{9} +68.661\overline{5}\,i\right) k + \left(65.051\overline{6} - 48.408\overline{9}\,i\right) k^2,
   \end{split}
\end{equation}
\begin{equation}
\label{4.25}
\begin{split}
d_{\rm{r}}(k):=&-\displaystyle\frac{7.0565\overline{2} + 2.5552\overline{8}\,i}{k} + \left(7.2733\overline{5} - 14.352\overline{2} \,i\right) \\
   &+\left(19.075\overline{6} +11.165\overline{9}\,i\right) k - \left(10.490\overline{3} - 18.172\overline{9}\,i\right) k^2.
   \end{split}
\end{equation}
Note that we use an overbar on a digit to denote a roundoff on that digit.
From \eqref{4.16}--\eqref{4.25} we see that $T_{\rm{l}}(k)\not\equiv T_{\rm{r}}(k),$ and in fact we confirm 
\eqref{4.1} up to $O(k^3)$ as $k\to 0.$

\end{example}

In Example~\ref{example4.2}, we have illustrated that the matrix-valued left and right transmission coefficients are not
equal to each other in general, and that has been done when the potential is selfadjoint but not real. In order to demonstrate that
the unequivalence of $T_{\rm{l}}(k)$ and $T_{\rm{r}}(k)$ is not caused because the potential
is complex valued, we would like to 
demonstrate that, in general, when $n\ge 2$ we have
$T_{\rm{l}}(k)\not\equiv T_{\rm{r}}(k)$ even when the matrix potential $V$ in \eqref{1.1} is real valued. Suppose
that, in addition to
\eqref{1.2} and \eqref{1.3},
the potential $V$ is real valued, i.e. we have
\begin{equation*}
V(x)^\ast=V(x),\qquad k\in\mathbb R,
\end{equation*}
where we recall that we use an asterisk to denote complex conjugation.
Then, from \eqref{1.1} we see that if $\psi(k,x)$ is a solution to
\eqref{1.1} then $\psi(\pm k^\ast,x)^\ast$ is also a solution. In particular, using \eqref{2.1} and \eqref{2.2}
we observe that, when the potential is real valued, we have
\begin{equation}
\label{4.27}
f_{\rm{l}}(-k^\ast,x)^\ast =f_{\rm{l}}(k,x),\quad 
f_{\rm{r}}(-k^\ast ,x)^\ast =f_{\rm{r}}(k,x),\qquad k\in\overline{\mathbb C^+}.
\end{equation}
Then, using \eqref{4.27}, from  \eqref{2.3} and \eqref{2.5} we obtain
\begin{equation}
\label{4.28}
T_{\rm{l}}(-k^\ast)^\ast=T_{\rm{l}}(k),\quad T_{\rm{r}}(-k^\ast)^\ast=T_{\rm{r}}(k),\qquad k\in\overline{\mathbb C^+},
\end{equation}
\begin{equation}
\label{4.29}
R(-k)^\ast=R(k),\quad L(-k)^\ast=L(k),\qquad k\in\mathbb R.
\end{equation}
Comparing \eqref{4.1} and \eqref{4.28}, we have
\begin{equation}
\label{4.30}
T_{\rm{l}}(k)^t=T_{\rm{r}}(k),\qquad k\in\overline{\mathbb C^+},
\end{equation}
and similarly, by comparing the first two equalities in \eqref{2.49} with \eqref{4.29} we get
\begin{equation*}
R(k)^t=R(k),\quad L(k)^t=L(k),\qquad k\in\mathbb R,
\end{equation*}
where we use the superscript $t$ to denote the matrix transpose.
We remark that, in the scalar case, since the potential $V$ in \eqref{1.1} is real valued, 
the equality of the left and right transmission coefficients directly follow from \eqref{4.30}.
Let us mention that, if the matrix potential is real valued, then \eqref{3.67}--\eqref{3.70}
of Theorem~\ref{theorem3.6}
yield
\begin{equation}
\label{4.32}
T_{\rm{l}}(k)=T_{\rm{l}2}(k)\left[I-R_1(k)\,L_2(k)\right]^{-1} T_{\rm{r}1}(k)^t,\qquad k\in\mathbb R,
\end{equation}
\begin{equation*}
L(k)=[T_{\rm{r}1}(k)^\dagger]^{-1}\left[L_2(k)-R_1(k)^\ast\right]\left[I-R_1(k)\,L_2(k)\right]^{-1} T_{\rm{r}1}(k)^t,
\qquad  k\in\mathbb R,
\end{equation*}
\begin{equation}
\label{4.34}
T_{\rm{r}}(k)=T_{\rm{r}1}(k)\left[I-L_2(k)\,R_1(k)\right]^{-1} T_{\rm{l}2}(k)^t,\qquad k\in\mathbb R,
\end{equation}
\begin{equation*}
R(k)=T_{\rm{l}2}(k)\left[I-R_1(k)\,L_2(k)\right]^{-1}  \left[R_1(k)-L_2(k)^\ast\right] T_{\rm{l}2}(-k)^{-1},
\qquad k\in\mathbb R.
\end{equation*}

In the next example, we illustrate the unequivalence of the left and right
transmission coefficients when the matrix potential $V$ in \eqref{1.1} is real valued.
As in the previous example, we construct $T_{\rm{l}}(k)$ and $T_{\rm{r}}(k)$ in terms of the
scattering coefficients corresponding to the fragments $V_1$ and $V_2$
specified in \eqref{1.4} and \eqref{1.5}, and then we check the unequivalence of the
resulting expressions.

\begin{example}
\label{example4.3}
\normalfont
Consider the full-line Schr\"odinger equation \eqref{1.1} with the matrix potential $V$ 
consisting of the fragments $V_1$ and $V_2$
as in \eqref{1.4} and \eqref{1.5}, where those fragments are real valued and given by
\begin{equation}
\label{4.36}
V_1(x)=\begin{bmatrix}3&2\\
2&c\end{bmatrix},\qquad -2<x<0,
\end{equation}
\begin{equation}
\label{4.37}
\quad V_2(x)=\begin{bmatrix}2&1\\
1&1\end{bmatrix},\qquad 0<x<1.
\end{equation}
The parameter $c$ appearing in \eqref{4.36} is a real constant, and it is understood that
the support of $V_1$ is the interval $x\in(-2,0)$ and the support
of $V_2$ is $x\in(0,1).$
The procedure to evaluate the corresponding scattering coefficients is not trivial.
As described in Example~\ref{example4.2}, we evaluate
the scattering coefficients corresponding to $V_1$ and $V_2$
by diagonalizing each of the constant $2\times 2$ matrices
$V_1$ and $V_2$ and by constructing the corresponding
Jost solutions $f_{\rm{r}1}(k,x)$ and $f_{\rm{l}2}(k,x).$
We then use \eqref{4.32} and \eqref{4.34} to obtain
$T_{\rm{l}}(k)$ and $T_{\rm{r}}(k).$
As in Example~\ref{example4.2} we have prepared a Mathematica notebook
to evaluate $T_{\rm{l}}(k)$ and $T_{\rm{r}}(k)$ explicitly 
in a closed form. The resulting expressions are extremely lengthy, and hence
it is not practical to display them in our paper.
Because the potential $V$ is real valued in this example, in order to check
if $T_{\rm{l}}(k)\not\equiv T_{\rm{r}}(k),$ as seen from \eqref{4.30} it is sufficient
to check whether the $2\times 2$ matrix $T_{\rm{l}}(k)$ is symmetric or not.
We remark that it is also possible to
evaluate $T_{\rm{l}}(k)$ directly by constructing the Jost solution $f_{\rm{l}}(k,x)$ when $x<0$ and then by using \eqref{2.3}.
However, the evaluations in that case are much more involved, and Mathematica is not capable of carrying out the computations 
properly unless a more powerful computer is used.
This indicates the usefulness of the factorization formula in evaluating the matrix-valued scattering coefficients for the full-line matrix  Schr\"odinger equation.
For various values of the parameter $c,$ by evaluating the difference $T_{\rm{l}}(k)-T_{\rm{l}}(k)^t$ at any $k$-value 
we confirm that $T_{\rm{l}}(k)$
is not a symmetric matrix. At any particular $k$-value, we have $T_{\rm{l}}(k)=T_{\rm{l}}(k)^t$ if and only if 
the matrix norm of
$T_{\rm{l}}(k)-T_{\rm{l}}(k)^t$ is zero. 
Using Mathematica, we are able to plot that matrix norm as a function of $k$ in the interval $k\in[0,b]$
for any positive $b.$ We then observe that that norm is strictly positive, and hence
we are able to confirm that in general we have $T_{\rm{l}}(k)\not\equiv T_{\rm{r}}(k)$ when $n\ge 2.$ One other reason
for us to use the parameter $c$ in \eqref{4.36} is the following.
The value of $c$ affects the number of eigenvalues for the full-line Schr\"odinger
operator with the specified potential $V,$ and hence by using various different values of $c$ as input
we can check the unequivalence of $T_{\rm{l}}(k)$ and $T_{\rm{r}}(k)$ as the number of eigenvalues changes.
The eigenvalues occur at
the $k$-values on the positive imaginary axis in the complex plane where the
determinant of $T_{\rm{l}}(k)$ has poles. Thus, we are able to identify the eigenvalues
by locating the zeros of $\det[T_{\rm{l}}(i\kappa)^{-1}]$ when $\kappa>0.$ Let us recall that we use an overbar on a 
digit to indicate a roundoff. We find that the numerically approximate value
$c=1.1372\overline{5}$ corresponds to an exceptional case, where the number of eigenvalues
changes by $1.$ For example, when
$c>1.1372\overline{5}$ we observe that there are no eigenvalues and that
there is exactly one eigenvalue when $0<c<1.1372\overline{5}.$
 When $c=0$ we have an eigenvalue at
$k=0.55\overline{1}i,$ the eigenvalue shifts to $k=0.069\overline{5}i$ when $c=1.$ When
$c=1.3$ we do not have any eigenvalues and the zero of $\det[T_{\rm{l}}(k)^{-1}]$ occurs on the negative
imaginary axis at
$k=-0.79\overline{4}i.$
We repeat our examination of the unequivalence of
$T_{\rm{l}}(k)$ and $T_{\rm{r}}(k)$ by also changing other 
entries of the matrices $V_1(x)$ and $V_2(x)$ appearing in \eqref{4.36} and \eqref{4.37}, respectively.
For example, by letting
\begin{equation}
\label{4.38}
V_1(x)=\begin{bmatrix}3&-2\\
-2&-5\end{bmatrix},\qquad -2<x<0,
\end{equation}
\begin{equation}
\label{4.39}
\quad V_2(x)=\begin{bmatrix}2&1\\
1&-2\end{bmatrix},\qquad 0<x<1,
\end{equation}
we evaluate $T_{\rm{l}}(k)$ using \eqref{4.32}.
In this case we observe that
there are three eigenvalues occurring when
\begin{equation*}
k=0.00563\overline{5}\,i,\quad k=1.273\overline{7}\,i,\quad k=2.0880\overline{2}\,i,
\end{equation*}
and we still observe that $T_{\rm{l}}(k)\not\equiv T_{\rm{r}}(k)$ in this case. 
In fact, using $V_1$ and $V_2$ appearing in
\eqref{4.38} and \eqref{4.39}, respectively, as input,
we observe that the corresponding transmission coefficients satisfy
\begin{equation*}
2ik\,T_{\rm{l}}(k)^{-1}=\begin{bmatrix}-130.26\overline{7}&28.398\overline{4}\\
\noalign{\medskip}
-43.55\overline{5}&9.4650\overline{8}\end{bmatrix}+O(k),\qquad k\to 0,
\end{equation*}
\begin{equation*}
2ik\,T_{\rm{r}}(k)^{-1}=\begin{bmatrix}-130.26\overline{7}&-43.55\overline{5}\\
\noalign{\medskip}
28.398\overline{4}&9.4650\overline{8}\end{bmatrix}+O(k),\qquad k\to 0,
\end{equation*}
confirming that we cannot have $T_{\rm{l}}(k)\equiv T_{\rm{r}}(k).$ In Figure~\ref{figure4.1} 
we present the plot of the matrix norm of $T_{\rm{l}}(k)-T_{\rm{r}}(k)$
as a function of $k,$ which also shows that $T_{\rm{l}}(k)\not\equiv T_{\rm{r}}(k).$

\end{example}

\begin{figure}[!ht]
\centering
\includegraphics[width=12cm,height=10cm,keepaspectratio]{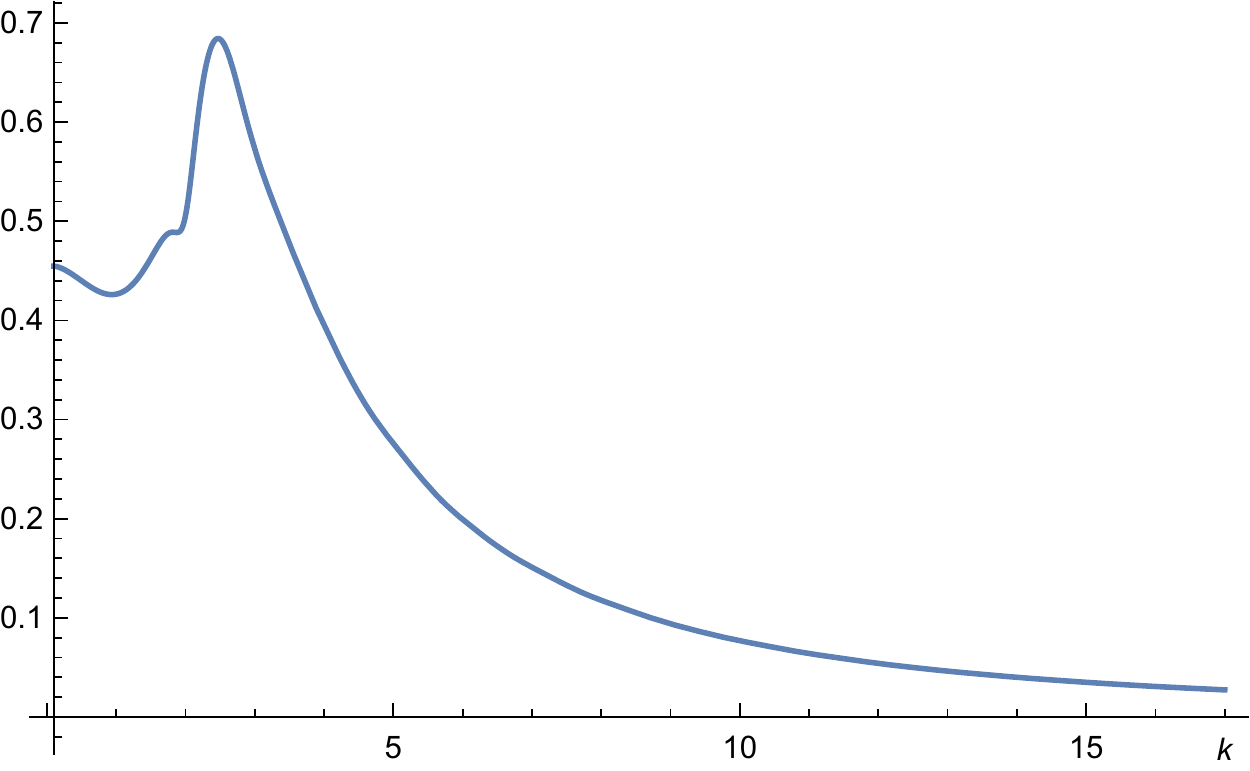}
\caption{The matrix norm $|T_{\rm{l}}(k)-T_{\rm{r}}(k)|$ vs $k$  in Example~\ref{example4.3} for the potential 
fragmented as in \eqref{1.4} and \eqref{1.5} with 
$V_1$ and $V_2$ given in \eqref{4.38} and \eqref{4.39}, respectively.}
\label{figure4.1}
\end{figure}

\section{The connection to the half-line Schr\"odinger operator}
\label{section5}

In this section we explore an important connection between 
the full-line $n\times n$ matrix Schr\"odinger equation \eqref{1.1} and
the half-line $2n\times 2n$ matrix Schr\"odinger equation
\begin{equation}
\label{5.1}
-\phi''+\mathbf V(x)\,\phi=k^2\phi,\qquad x\in\mathbb R^+,
\end{equation}
where $\mathbb R^+:=(0,+\infty)$ and 
the potential $\mathbf V$ is a $2n\times 2n$ matrix-valued
function of $x.$ The connection will be made by choosing the potential
$\mathbf V$ in terms of the fragments $V_1$ and $V_2$ appearing in \eqref{1.4} and \eqref{1.5} for the full-line potential $V$ in an appropriate way
and also by supplementing \eqref{5.1} with an appropriate boundary condition. 
To make a distinction between the quantities relevant to the full-line Schr\"odinger equation
\eqref{1.1} and the quantities relevant to the half-line
Schr\"odinger equation \eqref{5.1}, we use boldface to denote some of quantities
related to \eqref{5.1}.

Before making the connection between the half-line and full-line Schr\"odinger operators, we 
first provide a summary of some basic relevant facts related to \eqref{5.1}.
Since our interest in the half-line Schr\"odinger operator is restricted to its connection to the
full-line Schr\"odinger operator, we consider \eqref{5.1} when the matrix
potential has size $2n\times 2n,$ where $n$ is the positive integer related to the matrix size $n\times n$ of
the potential $V$ in \eqref{1.1}.
We refer the reader to
\cite{AW2021} for the analysis of \eqref{5.1} when the size of the matrix potential $\mathbf V$ is 
$n\times n,$ where $n$ can be chosen as any positive integer.

We now present some basic relevant facts related to \eqref{5.1}
by assuming that the half-line $2n\times 2n$ matrix potential $\mathbf V$ in \eqref{5.1} satisfies
\begin{equation}
\label{5.2}
\mathbf V(x)^\dagger=\mathbf V(x),\qquad x\in\mathbb R^+,
\end{equation}
\begin{equation}
\label{5.3}
\int_0^\infty dx\,(1+x)\,|\mathbf V(x)|<+\infty.
\end{equation}
To construct the half-line matrix Schr\"odinger operator related to \eqref{5.1},
we supplement \eqref{5.1}
with the general selfadjoint boundary condition
\begin{equation}
\label{5.4}
-B^\dagger \phi(0)+A^\dagger \phi'(0)=0,
\end{equation}
where $A$ and $B$ are two constant $2n\times 2n$ matrices satisfying
\begin{equation}
\label{5.5}
A^\dagger A+B^\dagger B>0,\quad B^\dagger A=A^\dagger B.
\end{equation}
We recall that a matrix is positive (also called positive definite) when all its eigenvalues are positive.

It is possible to express the $2n$ boundary conditions listed in \eqref{5.4} in an uncoupled form.
We refer the reader to Section~3.4 of \cite{AW2021} for the explicit steps to transform from any pair $(A,B)$ of
boundary matrices appearing in the general selfadjoint boundary condition described
in \eqref{5.4} and \eqref{5.5} to the diagonal boundary matrix pair
$(\tilde A,\tilde B)$ given by
\begin{equation}\label{5.6}
\begin{cases}
\tilde A=-\text{\rm{diag}}\left\{\sin\theta_1,\sin\theta_2,\cdots,\sin\theta_{2n}\right\},\\
\tilde B=\text{\rm{diag}}\left\{\cos\theta_1,\cos\theta_2,\cdots,\cos\theta_{2n}\right\},
\end{cases}
\end{equation}
for some appropriate real parameters $\theta_j\in(0,\pi].$ It can directly be verified that the matrix pair $(\tilde A,\tilde B)$ satisfies the two
equalities in \eqref{5.5} and that \eqref{5.4} is equivalent to the uncoupled system of $2n$ equations given by
\begin{equation}\label{5.7}
\left(\cos\theta_j\right)\phi_j(0)+\left(\sin\theta_j\right)\phi'_j(0)
=0,\qquad 1\le j\le 2n.
\end{equation}
We remark that the case $\theta_{j}=\pi$
corresponds to the Dirichlet boundary condition and the case $\theta_{j}=\pi/2$ corresponds to the Neumann boundary condition. 
Let us use $n_{\text{\rm D}}$ and $n_{\text{\rm N}}$ to denote the number of
Dirichlet and Neumann boundary conditions in \eqref{5.7}, and let us use
$n_{\text{\rm M}}$ to denote the number of mixed boundary conditions in \eqref{5.7},
where a mixed boundary condition occurs when $\theta\in(0,\pi/2)$ or $\theta\in(\pi/2,\pi).$
It is clear that we have
\begin{equation*}
n_{\text{\rm D}}+n_{\text{\rm N}}+n_{\text{\rm M}}=2n.
\end{equation*}
In Sections~3.3 and 3.5 of \cite{AW2021} we have constructed
a selfadjoint realization of the matrix Schr\"odinger operator $ -d^2/dx^2+ \mathbf V(x)$ with the boundary condition 
described in \eqref{5.4} and \eqref{5.5}, and we have
used $H_{A,B,\mathbf V}$ to denote it.
In Section~3.6 of \cite{AW2021} we have shown that
that selfadjoint realization with the boundary matrices
$(A,B)$ and the selfadjoint realization with the boundary matrices
$(\tilde A,\tilde B)$ are related to each other as
\begin{equation}\label{5.9}
H_{A,B, \mathbf V}= M H_{\tilde{A}, \tilde{B}, M^\dagger\mathbf V M}  M^\dagger.
\end{equation}
for some $2n\times 2n$ unitary matrix $M.$

In Section~2.4 of \cite{AW2021} we have established a unitary transformation between the
half-line $2n\times 2n$ matrix Schr\"odinger operator and the full-line
$n\times n$ matrix Schr\"odinger operator by choosing the boundary matrices $A$ and $B$ in
\eqref{5.4} appropriately so that 
the full-line potential $V$ at $x=0$ includes a point interaction. Using that unitary transformation, in \cite{W2022} the
half-line physical solution to \eqref{5.1} and the full-line physical solutions to \eqref{1.1} are related to each other,
and also the corresponding half-line scattering matrix and full-line scattering matrix are related to each other.
In this section of our current paper, in the absence of a point interaction, via a unitary transformation
 we are interested in analyzing the connection between various half-line quantities
and the corresponding full-line quantities.

By proceeding as in Section~2.4  of \cite{AW2021}, we establish our unitary operator
$\mathbf U$ from $L^2(\mathbb R^+;\mathbb C^{2n})$ onto $L^2(\mathbb R;\mathbb C^n)$ 
as follows.
We decompose any complex-valued, square-integrable column vector $\phi(x)$ with $2n$ components
into two pieces of column vectors $\phi_+(x)$ and $\phi_-(x),$ each with $n$ components, as
\begin{equation*}
\phi(x)=:\begin{bmatrix}\phi_+(x)\\
\noalign{\medskip}
\phi_-(x)\end{bmatrix},\qquad x\in\mathbb R^+.
\end{equation*}
Then, our unitary transformation $\mathbf U$ maps $\phi(x)$ onto
the complex-valued, square-integrable column vector $\psi(x)$ with $n$ components in such a way that
\begin{equation*}
\psi(x)=\begin{cases}\phi_+(x),\qquad x>0,\\
\noalign{\medskip}
\phi_-(-x),\qquad x<0.
\end{cases}
\end{equation*}
We use the decomposition of the full-line $n\times n$ matrix potential $V$
into the potential fragments $V_1$ and $V_2$ as described in \eqref{1.4} and \eqref{1.5}. We choose the 
half-line $2n\times 2n$ matrix potential $\mathbf V$ in terms of $V_1$ and $V_2$ 
by letting
\begin{equation}\label{5.10}
\mathbf V(x)=\begin{bmatrix}V_2(x)&0\\
\noalign{\medskip}
0&V_1(-x)\end{bmatrix},\qquad x\in\mathbb R^+.
\end{equation}
The inverse transformation $\mathbf U^{-1},$ which is equivalent to $\mathbf U^\dagger,$ maps $\psi(x)$ onto
$\phi(x)$ via
\begin{equation*}
\phi(x)=\begin{bmatrix}\psi(x)\\
\noalign{\medskip}
\psi(-x)\end{bmatrix},\qquad x\in\mathbb R^+.
\end{equation*}
Under the action of the unitary transformation $\mathbf U,$ the half-line Hamiltonian
$H_{A,B, \mathbf V}$ appearing on the left-hand side of \eqref{5.9} is unitarily transformed
onto the full-line Hamiltonian $H_V,$ where $H_V$ is related to $H_{A,B, \mathbf V}$ as
\begin{equation}\label{5.12}
H_V=\mathbf U H_{A,B, \mathbf V}\mathbf U^\dagger,
\end{equation}
and its domain $D[H_V]$ is given by
\begin{equation*}
D[H_V]=\left\{\psi\in L^2(\mathbb R;\mathbb C^n): \mathbf U^\dagger\psi\in D[H_{A,B, \mathbf V}]\right\},
\end{equation*}
where $D[H_{A,B, \mathbf V}]$ denotes the domain of $H_{A,B, \mathbf V}.$
The operator $H_V$ specified in \eqref{5.12} is a 
selfadjoint realization in $L^2(\mathbb R;\mathbb C^n)$ of
the formal differential operator
$-d^2/dx^2+ V(x),$ where $V$ is the full-line potential 
appearing in \eqref{1.4} and satisfying \eqref{1.2} and \eqref{1.3}.

The boundary condition \eqref{5.4} at $x=0$ 
of $\mathbb R^+$ satisfied by the functions in $D[H_{A,B,\mathbf V}]$ implies that the functions in 
$D[H_V]$ themselves satisfy a transmission condition at $x=0$ of the full line $\mathbb R.$ In order to determine
 that transmission condition, we express the boundary
 matrices $A$ and $B$ appearing in \eqref{5.4} in terms of $n\times 2n$ block matrices
 $A_1,$ $A_2,$ $B_1,$ and $B_2$ as
 \begin{equation}\label{5.14}
A=:\begin{bmatrix} A_1\\
\noalign{\medskip}
A_2\end{bmatrix},\quad
B=:\begin{bmatrix} B_1\\
\noalign{\medskip}
B_2\end{bmatrix}.
\end{equation}
Using \eqref{5.14} in \eqref{5.4} we see that any function
$\psi(x)$ in $D[H_V]$ satisfies 
the transmission condition at $x=0$ given by
\begin{equation}\label{5.15}
-B_1^\dagger\,\psi(0^+)-B_2^\dagger\,\psi(0^-)+A_1^\dagger\,\psi'(0^+)-A_2^\dagger\,\psi(0^-)=0. 
\end{equation}
For example, let us choose the boundary matrices
$A$ and $B$ 
as
\begin{equation}\label{5.16}
A=\begin{bmatrix} 0&I\\
\noalign{\medskip}
0&I\end{bmatrix},\quad
B=\begin{bmatrix} -I&0\\
\noalign{\medskip}
I&0\end{bmatrix},
\end{equation}
where we recall that $I$ is the $n\times n$ identity matrix
and $0$ denotes the $n\times n$ zero matrix.
It can directly be verified that the matrices $A$ and $B$ appearing
in \eqref{5.16} satisfy the two matrix equalities in \eqref{5.5}.
Then, the transmission condition at $x=0$ of $\mathbb R$ given in \eqref{5.15} is equivalent to
the two conditions 
\begin{equation*}
\psi(0^+)=\psi(0^-),\quad \psi'(0^+)=\psi'(0^-), 
\end{equation*}
which indicate that the functions $\psi$ and
$\psi'$ are continuous at $x=0.$
In this case, $H_V$ is the standard matrix Schr\"odinger operator 
on the full line without a point interaction.
In the special case when the boundary
matrices $A$ and $B$ are chosen as
in \eqref{5.16}, as shown in Proposition~5.9 of \cite{W2022}, 
the corresponding transformed boundary matrices $\tilde A$ and $\tilde B$ in 
\eqref{5.6} yield precisely $n$ Dirichlet and $n$ Neumann boundary conditions.

We now introduce some relevant quantities
for the half-line Schr\"odinger equation \eqref{5.1} with the boundary condition
\eqref{5.4}. We assume that the half-line potential $\mathbf V$ is chosen as in \eqref{5.10}, where
the full-line potential $V$ satisfies \eqref{1.2} and \eqref{1.3}. Hence,
$\mathbf V$ 
satisfies
\eqref{5.2} and \eqref{5.3}.
As already mentioned, we use
boldface to denote some of the half-line quantities in order to make a contrast with
the corresponding full-line quantities. For example, $\mathbf V$ is the half-line $2n\times 2n$ matrix potential
whereas $V$ is the full-line $n\times n$ matrix potential,
$\mathbf S(k)$ is the half-line $2n\times 2n$ scattering matrix
while $S(k)$ is the full-line $2n\times 2n$ scattering matrix, $\mathbf f(k,x)$
is the half-line $2n\times 2n$ matrix-valued Jost solution whereas $f_{\rm{l}}(k,x)$ and $f_{\rm{r}}(k,x)$ are
the full-line $n\times n$ matrix-valued Jost solutions, $\mathbf \Phi(k,x)$ is the half-line $2n\times 2n$ matrix-valued
regular solution, the quantity $\mathbf \Psi(k,x)$ denotes the half-line $2n\times 2n$ matrix-valued
physical solution whereas
$\Psi_{\rm{l}}(k,x)$ and $\Psi_{\rm{r}}(k,x)$ are the full-line $n\times n$ matrix-valued
physical solutions. We use $I$ for the 
$n\times n$ identity matrix and use $\mathbf I$ for the $2n\times 2n$ identity matrix.
We recall that $J$ denotes the $2n\times 2n$ constant matrix defined in \eqref{3.10}
and it should not be confused with the half-line 
$2n\times 2n$ Jost matrix $\mathbf J(k).$

The Jost solution $\mathbf f(k,x)$ is the solution to \eqref{5.1} satisfying the spacial asymptotics
\begin{equation}
\label{5.18}
\mathbf f(k,x)=e^{ikx}\left[\mathbf I
+o(1)\right],\quad \mathbf f'(k,x)=ik\,e^{ikx}\left[\mathbf I+
o(1)\right],\qquad x\to+\infty.
\end{equation}
In terms of the $2n\times 2n$ boundary matrices $A$ and $B$ appearing in \eqref{5.4} and the Jost solution $\mathbf f(k,x),$
the $2n\times 2n$ Jost matrix $\mathbf J(k)$ is defined as
\begin{equation}
\label{5.19}
\mathbf J(k):=\mathbf f(-k^*,0)^\dagger B-\mathbf f'(-k^*,0)^\dagger A,\qquad k\in\overline{\mathbb C^+},
\end{equation}
where $-k^*$ is used because $\mathbf J(k)$ has \cite{AW2021} an analytic extension in $k$ from 
$\mathbb R$ to $\mathbb C^+$
and  $\mathbf J(k)$ is continuous in $\overline{\mathbb C^+}.$ 
The half-line $2n\times 2n$ scattering matrix $\mathbf S(k)$ is defined in terms of the 
Jost matrix $\mathbf J(k)$ as
\begin{equation}
\label{5.20}
\mathbf S(k):=-\mathbf J(-k)\, \mathbf J(k)^{-1},\qquad k\in\mathbb R.
\end{equation}
When 
the potential $\mathbf V$ satisfies \eqref{5.2} and \eqref{5.3}, in the so-called exceptional case the matrix
$\mathbf J(k)^{-1}$ has a singularity at $k=0$ even though the limit of the right-hand side
of \eqref{5.20} exists when $k\to 0.$ Thus, $\mathbf S(k)$ is continuous 
in $k\in\mathbb R$ including $k=0.$ 
 It is known that $\mathbf S(k)$ satisfies 
\begin{equation*}
\mathbf S(k)^{-1}=\mathbf S(k)^\dagger=\mathbf S(-k),\qquad k\in\mathbb R.
\end{equation*}
We refer the reader to 
Theorem~3.8.15 of \cite{AW2021} regarding the small $k$-behavior of
$\mathbf J(k)^{-1}$ and $\mathbf S(k).$

The $2n\times 2n$ matrix-valued physical solution $\mathbf\Psi(k,x)$ to \eqref{5.1} is defined in terms of the
Jost solution $\mathbf f(k,x)$ and the scattering matrix $\mathbf S(k)$ as
\begin{equation}
\label{5.22}
\mathbf\Psi(k,x):=\mathbf f(-k,x)+\mathbf f(k,x)\,\mathbf S(k).
\end{equation}
It is known \cite{AW2021} that $\mathbf\Psi(k,x)$ has a meromorphic extension from $k\in\mathbb R$ to
$k\in\mathbb C^+$ and it also satisfies the boundary condition \eqref{5.4}, i.e. we have
\begin{equation}
\label{5.23}
-B^\dagger \mathbf\Psi(k,0)+A^\dagger \mathbf\Psi'(k,0)=0.
\end{equation}
There is also a particular $2n\times 2n$ matrix-valued solution $\mathbf\Phi(k,x)$ to \eqref{5.1}
satisfying the initial conditions
\begin{equation*}
\mathbf\Phi(k,0)=A,\quad \mathbf\Phi'(k,0)=B.
\end{equation*}
Because $\mathbf \Phi(k,x)$ is entire in $k$ for each fixed $x\in\mathbb R^+,$
it is usually called the regular solution.
The physical solution $\mathbf\Psi(k,x)$ and the regular solution $\mathbf\Phi(k,x)$ are related to each other 
via the Jost matrix $\mathbf J(k)$ as
\begin{equation*}
\mathbf\Psi(k,x)=-2ik\,\mathbf\Phi(k,x)\,\mathbf J(k)^{-1}.
\end{equation*}

In the next theorem, we consider the special case where 
the half-line Schr\"odinger operator $H_{A,B,\mathbf V}$ and the full-line Schr\"odinger operator $H_V$ are related to each other as in \eqref{5.12}, with
the potentials $\mathbf V$ and $V$ being related as in \eqref{1.4} and \eqref{5.10} and with the boundary matrices
$A$ and $B$ chosen as in \eqref{5.16}. We show how the corresponding half-line Jost solution, half-line
physical solution,
half-line scattering matrix, and half-line Jost matrix are related to
the appropriate full-line quantities. We remark that the results in
\eqref{5.28} and \eqref{5.29} have already been proved in \cite{W2022} by using a different method.

\begin{theorem}
\label{theorem5.1} Consider the full-line matrix Schr\"odinger equation \eqref{1.1} with the $n\times n$
matrix potential $V$ satisfying \eqref{1.2} and \eqref{1.3}. Assume that the corresponding full-line
Hamiltonian $H_V$ and the half-line Hamiltonian
$H_{A,B,\mathbf V}$ are unitarily connected as in \eqref{5.12} by relating the half-line $2n\times 2n$ matrix potential $\mathbf V$ to
$V$ as in \eqref{1.4} and \eqref{5.10} and
by choosing the boundary matrices $A$ and $B$ as in \eqref{5.16}.
Then, we have the following:

\begin{enumerate}

\item[\text{\rm(a)}] The half-line $2n\times 2n$ matrix-valued Jost solution $\mathbf f(k,x)$ to \eqref{5.1}
appearing in \eqref{5.18} is related to the full-line $n\times n$ matrix-valued Jost solutions
$f_{\rm{l}}(k,x)$ and $f_{\rm{r}}(k,x)$ appearing in \eqref{2.1} and \eqref{2.2}, respectively, as
\begin{equation}
\label{5.26}
\mathbf f(k,x)=\begin{bmatrix}f_{\rm{l}}(k,x)&0\\
\noalign{\medskip}
0&f_{\rm{r}}(k,-x)\end{bmatrix},\qquad  x\in\mathbb R^+, \quad k \in \overline{\mathbb C^+}.
\end{equation}

\item[\text{\rm(b)}] The half-line $2n\times 2n$ scattering matrix $\mathbf S(k)$ defined in \eqref{5.20} is related to
the full-line $2n\times 2n$ scattering matrix $S(k)$ defined in \eqref{2.5} as 
\begin{equation}\label{5.27}
\mathbf S(k)=S(k)\,Q,\qquad k\in\mathbb R,
\end{equation}
 where $Q$ is the $2n\times 2n$ constant matrix
defined in \eqref{2.51}. Hence, the half-line scattering matrix $\mathbf S(k)$ is related to the full-line
$n\times n$ matrix-valued scattering coefficients as
\begin{equation}
\label{5.28}
\mathbf S(k)=\begin{bmatrix}R(k) &T_{\rm{l}}(k)\\
\noalign{\medskip}
T_{\rm{r}}(k)&L(k)\end{bmatrix},\qquad k\in\mathbb R.
\end{equation}

\item[\text{\rm(c)}] The half-line $2n\times 2n$ physical solution $\mathbf\Psi(k,x)$ defined in \eqref{5.22} is related to the full-line 
$n\times n$ matrix-valued
physical solutions $\Psi_{\rm{l}}(k,x)$ and $\Psi_{\rm{r}}(k,x)$ appearing in \eqref{2.6} as
\begin{equation}
\label{5.29}
\mathbf\Psi(k,x)=\begin{bmatrix}\Psi_{\rm{r}}(k,x) &\Psi_{\rm{l}}(k,x) \\
\noalign{\medskip}
\Psi_{\rm{r}}(k,-x) &\Psi_{\rm{l}}(k,-x) \end{bmatrix},\qquad  x\in\mathbb R^+, \quad k \in \overline{\mathbb C^+}.
\end{equation}

\item[\text{\rm(d)}] The half-line $2n\times 2n$ Jost matrix $\mathbf J(k)$ defined in \eqref{5.19}
and its inverse $\mathbf J(k)^{-1}$ are related to 
the full-line $n\times n$ matrix-valued Jost solutions $f_{\rm{l}}(k,x)$ and $f_{\rm{r}}(k,x)$
appearing in \eqref{2.1} and \eqref{2.2}
and the $n\times n$ matrix-valued transmission coefficients $T_{\rm{l}}(k)$ and $T_{\rm{r}}(k)$
appearing in \eqref{2.3} and \eqref{2.4} as
\begin{equation}
\label{5.30}
\mathbf J(k)=\begin{bmatrix}-f_{\rm{l}}(-k^*,0)^\dagger &-f'_{\rm{l}}(-k^*,0)^\dagger\\
\noalign{\medskip}
f_{\rm{r}}(-k^*,0)^\dagger&f'_{\rm{r}}(-k^*,0)^\dagger\end{bmatrix},\qquad k\in\overline{\mathbb C^+},
\end{equation}
\begin{equation}
\label{5.31}
\mathbf J(k)^{-1}=\displaystyle\frac{1}{2ik} \begin{bmatrix}f'_{\rm{r}}(k,0)&f'_{\rm{l}}(k,0)\\
\noalign{\medskip}
-f_{\rm{r}}(k,0)&-f_{\rm{l}}(k,0))\end{bmatrix}\begin{bmatrix}T_{\rm{r}}(k) &0\\
\noalign{\medskip}
0&T_{\rm{l}}(k)\end{bmatrix},\qquad k\in\overline{\mathbb C^+}\setminus\{0\}.
\end{equation}

\end{enumerate}

\end{theorem}

\begin{proof}
The Jost solutions $f_{\rm{l}}(k,x)$ and $f_{\rm{r}}(k,x)$ satisfy
\eqref{1.1}. Since $k$ appears as $k^2$ in \eqref{1.1}, 
the quantities 
$f_{\rm{l}}(-k,x),$ and $f_{\rm{r}}(-k,x)$ also satisfy \eqref{1.1}.
Furthermore, $f_{\rm{l}}(k,x)$ and $f_{\rm{r}}(k,x)$ satisfy the respective spacial asymptotics in \eqref{2.1} and \eqref{2.2}.
Then, with the help of \eqref{1.5} and \eqref{5.10} we see that
the half-line Jost solution $\mathbf f(k,x)$ given in \eqref{5.26} satisfies the
half-line Schr\"odinger equation \eqref{5.1} and the spacial asymptotics 
\eqref{5.18}. Thus, the proof of (a) is complete.
Let us now prove (b). We evaluate \eqref{5.22} and
its $x$-derivative at the point $x=0,$ and then we use the result in \eqref{5.23}, where
$A$ and $B$ are the matrices in \eqref{5.16}. This yields
\begin{equation}\label{5.32}
\begin{bmatrix}f_{\rm{l}}(-k,0)&-f_{\rm{r}}(-k,0)\\
\noalign{\medskip}
f'_{\rm{l}}(-k,0)&-f'_{\rm{r}}(-k,0)\end{bmatrix}+
\begin{bmatrix}f_{\rm{l}}(k,0)&-f_{\rm{r}}(k,0)\\
\noalign{\medskip}
f'_{\rm{l}}(k,0)&-f'_{\rm{r}}(k,0)\end{bmatrix}\mathbf S(k)=0.
\end{equation}
Using the matrices $G(k,x)$ and $J$ defined in \eqref{2.14} and
\eqref{3.10}, respectively, we write \eqref{5.32} as
\begin{equation*}
G(-k,0)\,J+G(k,0)\,J\,\mathbf S(k)=0,\qquad k\in\mathbb R,
\end{equation*}
which yields
\begin{equation}\label{5.34}
\mathbf S(k)=-J\,G(k,0)^{-1}\,G(-k,0)\,J,\qquad k\in\mathbb R.
\end{equation}
We remark that the invertibility of
$G(k,0)$ for $k\in\mathbb R\setminus\{0\}$ is assured by
Theorem~\ref{theorem2.3}(c) and that
\eqref{5.34} holds also at $k=0$ as a consequence of the continuity of $\mathbf S(k)$ in $k\in\mathbb R.$
From \eqref{3.9} we have
\begin{equation}\label{5.35}
G(k,0)^{-1}\,G(-k,0)=J\,S(k)\,J\,Q,\qquad k\in\mathbb R,
\end{equation}
where we recall that $S(k)$ is the full-line scattering matrix
in \eqref{2.5} and
$Q$ is the constant matrix in \eqref{2.51}.
Using \eqref{5.35} in \eqref{5.34} we get
\begin{equation}\label{5.36}
\mathbf S(k)=-J\left[J\,S(k)\,J\,Q\right] J,\qquad k\in\mathbb R.
\end{equation}
Using $J^2=I$ and $JQ=-QJ$ in \eqref{5.36}, we obtain \eqref{5.27}. Then, we get \eqref{5.28} by
using \eqref{2.5} in \eqref{5.27}. Thus, the proof of (b) is complete.
 Let us now prove (c). Using \eqref{5.26} and \eqref{5.28} in \eqref{5.22},
 we obtain
 \begin{equation}
\label{5.37}
\mathbf\Psi(k,x)=\begin{bmatrix}f_{\rm{l}}(-k,x)+f_{\rm{l}}(k,x)\,R(k)&f_{\rm{l}}(k,x)\,T_{\rm{l}}(k)\\
\noalign{\medskip}
f_{\rm{r}}(k,-x)\,T_{\rm{r}}(k)&f_{\rm{r}}(-k,-x)+f_{\rm{r}}(k,-x)\,L(k)\end{bmatrix}.
\end{equation}
The use of \eqref{3.3} and \eqref{3.4} in \eqref{5.37} yields
\begin{equation}
\label{5.38}
\mathbf\Psi(k,x)=\begin{bmatrix}f_{\rm{r}}(k,x)\,T_{\rm{r}}(k)&f_{\rm{l}}(k,x)\,T_{\rm{l}}(k)\\
\noalign{\medskip}
f_{\rm{r}}(k,-x)\,T_{\rm{r}}(k)&f_{\rm{l}}(k,-x)\,T_{\rm{l}}(k)\end{bmatrix}.
\end{equation}
Then, using \eqref{2.6} on the right-hand side of \eqref{5.38}, we obtain \eqref{5.29},
which completes the proof of (c).
We now turn to the proof of (d). 
In \eqref{5.19} we use \eqref{5.16}, \eqref{5.26}, and the $x$-derivative of \eqref{5.26}.
This gives us \eqref{5.30}. Finally, by postmultiplying
both sides of \eqref{5.31} with the respective sides of \eqref{5.30}, one can verify that
$\mathbf J(k)\,\mathbf J(k)^{-1}=\mathbf I.$ In the simplification of the left-hand side of
the last equality, one uses \eqref{2.41}, \eqref{2.43}, \eqref{2.45}, and \eqref{2.46}.
Thus, the proof of (d) is complete.
\end{proof}

We recall that, as \eqref{5.12} indicates, the full-line and half-line Hamiltonians
$H_V$ and $H_{A,B,\mathbf V},$ respectively, are unitarily equivalent.
However, as seen from \eqref{5.27}, the corresponding full-line and half-line scattering matrices
$S(k)$ and $\mathbf S(k),$ respectively, are not unitarily equivalent.
For an elaboration on this issue, we refer the reader to \cite{W2022}.

Let us also mention that it is possible to establish \eqref{5.31} without the direct verification used in the proof of 
Theorem~\ref{theorem5.1}. This can be accomplished as follows.
Comparing \eqref{2.14} and \eqref{5.30}, we see that
\begin{equation}\label{5.39}
\mathbf J(k)=-J\,G(-k^\ast,0)^\dagger
,\qquad k\in\overline{\mathbb C^+}.
\end{equation}
By taking the matrix inverse of both sides of \eqref{5.39}, we obtain
\begin{equation}\label{5.40}
\mathbf J(k)^{-1}=-\left[G(-k^\ast,0)^{-1}\right]^\dagger
J
,\qquad k\in\overline{\mathbb C^+}\setminus\{0\}.
\end{equation}
With the help of \eqref{2.14}, \eqref{2.51}, and \eqref{3.10}, we can write \eqref{2.56} as
\begin{equation}\label{5.41}
G(k,x)^{-1}=-\displaystyle\frac{1}{2ik}\,
\begin{bmatrix}T_{\rm{l}}(k)&0\\
\noalign{\medskip}
0&T_{\rm{r}}(k)\end{bmatrix}\,J\,Q\,G(-k^\ast,x)^\dagger \,Q\,J,
\qquad k\in\overline{\mathbb C^+}\setminus\{0\}.
\end{equation}
From \eqref{5.41},  for $k\in\overline{\mathbb C^+}\setminus\{0\}$ we get
\begin{equation}\label{5.42}
\left[G(-k^\ast,x)^{-1}\right]^\dagger=-\displaystyle\frac{1}{2ik}\,J\,Q\,
G(k,x)\,Q\,J\,
\begin{bmatrix}T_{\rm{l}}(-k^\ast)^\dagger&0\\
\noalign{\medskip}
0&T_{\rm{r}}(-k^\ast)^\dagger\end{bmatrix}.
\end{equation}
Using \eqref{4.1} in the last matrix factor on the right-hand side of
\eqref{5.42}, we can write
\eqref{5.42} as
\begin{equation}\label{5.43}
\left[G(-k^\ast,x)^{-1}\right]^\dagger=-\displaystyle\frac{1}{2ik}\,J\,Q\,
G(k,x)\,Q\,J\,
\begin{bmatrix}T_{\rm{r}}(k)&0\\
\noalign{\medskip}
0&T_{\rm{l}}(k)\end{bmatrix},
\qquad k\in\overline{\mathbb C^+}\setminus\{0\}.
\end{equation}
Finally, using \eqref{5.43} on the right-hand side of \eqref{5.40}, we 
obtain 
\begin{equation}\label{5.44}
\mathbf J(k)^{-1}=\displaystyle\frac{1}{2ik}\,J\,Q\,
G(k,0)\,Q\,
\begin{bmatrix}T_{\rm{r}}(k)&0\\
\noalign{\medskip}
0&T_{\rm{l}}(k)\end{bmatrix},
\qquad k\in\overline{\mathbb C^+}\setminus\{0\},
\end{equation}
which is equivalent to \eqref{5.31}.

In the next theorem, we relate the determinant of the half-line $2n\times 2n$ Jost matrix $\mathbf J(k)$ to the determinant of the full-line 
$n\times n$ matrix-valued
transmission coefficient $T_{\rm{l}}(k).$

\begin{theorem}
\label{theorem5.2} 
Consider the full-line matrix Schr\"odinger equation \eqref{1.1} with the $n\times n$ matrix
potential $V$ satisfying \eqref{1.2} and \eqref{1.3}. Assume that the corresponding full-line
Hamiltonian $H_V$ is unitarily connected to the half-line Hamiltonian
$H_{A,B,\mathbf V}$ as in \eqref{5.12} by relating the half-line $2n\times 2n$ matrix potential $\mathbf V$ to
$V$ as in \eqref{1.4} and \eqref{5.10} and
by choosing the boundary matrices $A$ and $B$ as in \eqref{5.16}.
Then, the determinant of the half-line $2n\times 2n$ Jost matrix $\mathbf J(k)$ defined in \eqref{5.19} is related to
the determinant of the 
corresponding full-line $n\times n$ matrix-valued transmission coefficient $T_{\rm{l}}(k)$ appearing in \eqref{2.3} as
\begin{equation}
\label{5.45}
\det[\mathbf J(k)]=\displaystyle\frac{(2ik)^n}{\det[T_{\rm{l}}(k)]},
\qquad k\in\overline{\mathbb C^+}.
\end{equation}
\end{theorem}

\begin{proof}
By taking the determinants of both sides of \eqref{5.44}, we obtain
\begin{equation}\label{5.46}
\displaystyle\frac{1}{\det[\mathbf J(k)]}=\displaystyle\frac{(-1)^n}{(2ik)^{2n}}\,\det[G(k,0)]\,
\det[T_{\rm{l}}(k)]\,\det[T_{\rm{r}}(k)],
\qquad k\in\overline{\mathbb C^+}\setminus\{0\},
\end{equation}
where we have used $Q^2=\mathbf I$ and $\det[J]=(-1)^n,$ 
which follow from \eqref{2.51} and \eqref{3.10}, respectively. 
Using \eqref{2.17} in \eqref{5.46}, we see that \eqref{5.45} holds.
\end{proof}

When the potential $\mathbf V$ satisfies \eqref{5.2} and \eqref{5.3},
from Theorems~3.11.1 and 3.11.6 of \cite{AW2021} we know that
the zeros of $\det[\mathbf J(k)]$ in $\overline{\mathbb C^+}\setminus\{0\}$ 
can only occur on the positive imaginary axis and the number of such zeros is finite.
Assume that the full-line Hamiltonian $H_V$ and the
half-line Hamiltonian
$H_{A,B,\mathbf V}$ 
are connected to each other through the unitary transformation
$\mathbf U$ as described in \eqref{5.12},
where $V$ and $\mathbf V$ are related as in 
\eqref{1.4} and \eqref{5.10} and the boundary matrices are chosen as in \eqref{5.16}.
We then have the following important consequence.
The number of eigenvalues of
$H_V$ coincides with
the number of eigenvalues of $H_{A,B,\mathbf V},$ and
the multiplicities of the corresponding eigenvalues also coincide.
The eigenvalues of
$H_{A,B,\mathbf V}$ occur at the $k$-values on the positive imaginary axis
in the complex $k$-plane where $\det[\mathbf J(k)]$ vanishes, 
and the multiplicity of each of those eigenvalues is equal to the 
order of the corresponding zero of $\det[\mathbf J(k)].$
Thus, as seen from \eqref{5.45}
the eigenvalues of $H_V$ occur on the positive imaginary axis
in the complex $k$-plane where $\det[T_{\rm{l}}(k)]$ has poles, and
the multiplicity of each of those eigenvalues is equal to
the order of the corresponding pole of $\det[T_{\rm{l}}(k)].$
Let us use $N$ to denote the number of such poles without counting their multiplicities, 
and let us assume that
those poles occur at $k=i\kappa_j$ for $1\le j\le N.$
We use $m_j$ to denote the multiplicity of the pole at
$k=i\kappa_j.$ The nonnegative integer $N$ corresponds to the
number of distinct eigenvalues $-\kappa_j^2$ of the corresponding full-line Schr\"odinger
operator $H_V$ associated with \eqref{1.1}.
If $N=0,$ then there are no eigenvalues. Let us use $\mathcal N$ to denote the
number of eigenvalues including the multiplicities. Hence, $\mathcal N$ is related to 
$N$ as
\begin{equation}
\label{5.47}
\mathcal N:=\displaystyle\sum_{j=1}^N  m_j.
\end{equation}
From \eqref{5.45} it is seen that 
the zeros of the determinant of the corresponding
Jost matrix $\mathbf J(k)$ occur at $k=i\kappa_j$
with multiplicity $m_j$ for $1\le j\le N.$
It is known \cite{AW2021} that $\mathbf J(k)$ is analytic in $\mathbb C^+$ and
continuous in $\overline{\mathbb C^+}.$ Thus, from 
\eqref{5.45} we also see that $\det[T_{\rm{l}}(k)]$ cannot vanish in $\overline{\mathbb C^+}\setminus\{0\}.$

The unitary equivalence given in \eqref{5.12} between the half-line matrix Schr\"odinger operator
$H_{A,B,\mathbf V}$ and the full-line matrix Schr\"odinger operator
$H_V$ has other important consequences. Let us comment on the connection 
between the half-line and full-line cases when $k=0.$
From 
Corollary~3.8.16 of \cite{AW2021} it follows that
\begin{equation} \label{5.48}
\det[\mathbf J(k)]= c_1\, k^\mu \left[1+o(1)\right], \qquad k \to 0 \text{\rm{ in }}\overline{\mathbb C^+},
\end{equation} 
where $\mu$ is the geometric multiplicity of the zero eigenvalue of $\mathbf J(0)$ and $c_1$ is a nonzero constant. Furthermore, from Proposition~3.8.18
of \cite{AW2021} we know that $\mu$ coincides with the geometric and algebraic multiplicity of the eigenvalue $+1$ of $\mathbf S(0).$ 
In fact, the nonnegative integer $\mu$ is related to \eqref{5.1} when $k=0,$ i.e. related to the zero-energy Schr\"odinger equation given by
\begin{equation}
\label{5.49}
-\phi''+\mathbf V(x)\,\phi=0,\qquad x\in\mathbb R^+.
\end{equation}
When the $2n\times 2n$ matrix potential $\mathbf V$ satisfies \eqref{5.2} and \eqref{5.3},
from (3.2.157) and Proposition~3.2.6 of \cite{AW2021} it follows that, among any fundamental set
of $4n$ linearly independent column-vector solutions to \eqref{5.49}, exactly $2n$ of
them are bounded and $2n$ are unbounded. In fact, the $2n$ columns of the zero-energy Jost solution $\mathbf f(0,x)$ 
make up the $2n$ linearly independent bounded solutions to \eqref{5.49}.
Certainly, not all of those $2n$ column-vector solutions necessarily satisfy the selfadjoint
boundary condition \eqref{5.4} with the $2n\times 2n$ boundary matrices $A$ and $B$ 
chosen as in \eqref{5.5}. From 
Remark~3.8.10 of \cite{AW2021} it is known that $\mu$ coincides with the number of linearly independent bounded solutions to \eqref{5.49} 
satisfying the boundary condition \eqref{5.4} and that
we have
\begin{equation}\label{5.50}
0\le \mu\le 2n.
\end{equation}
From \eqref{5.45} and \eqref{5.48} we obtain
\begin{equation}\label{5.51}
\det[T_{\rm{l}}(k)]= \frac{1}{c_1}\, (2i)^n\,   k^{n-\mu}\left[1+o(1)\right], \qquad k \to 0\text{\rm{ in }}\overline{\mathbb C^+}.
\end{equation}
We have the following additional remarks.
Let us consider \eqref{1.1} when $k=0,$ i.e. let us consider the full-line zero-energy
Schr\"odinger equation
\begin{equation}
\label{5.52}
-\psi''+V(x)\,\psi=0,\qquad x\in\mathbb R,
\end{equation}
where $V$ is the $n\times n$ matrix potential appearing in \eqref{1.1} and satisfying
\eqref{1.2} and \eqref{1.3}.
We recall that \cite{AKV1996}, 
for the full-line 
Schr\"odinger equation in the scalar case, i.e. when $n=1,$ the generic case occurs when there are no 
bounded solutions to \eqref{5.52} and that the exceptional case occurs when there exists a bounded solution to \eqref{5.52}.
In the $n\times n$ matrix case, let us assume that \eqref{5.52} has $\nu$
linearly independent bounded solutions. It is known \cite{AKV2001} that 
\begin{equation}
\label{5.53}
0\le \nu\le n.
\end{equation}
We can call $\nu$ the {\it degree of exceptionality}. 
From 
the unitary connection \eqref{5.12} it follows that 
the number of linearly independent bounded solutions to \eqref{5.52}
is equal to the number of linearly independent bounded solutions to \eqref{5.49}. Thus, we have
\begin{equation}
\label{5.54}
\mu=\nu.
\end{equation}
As a result, we can write \eqref{5.51} also as
\begin{equation}\label{5.55}
\det[T_{\rm{l}}(k)]= \frac{1}{c_1}\, (2i)^n\,   k^{n-\nu}\left[1+o(1)\right], \qquad k \to 0\text{\rm{ in }}\overline{\mathbb C^+}.
\end{equation}
From \eqref{5.50}, \eqref{5.53}, and \eqref{5.54} we conclude that
\begin{equation}\label{5.56}
0\le \mu\le n.
\end{equation}

The restriction from \eqref{5.50} to \eqref{5.56}, when we have 
the unitary equivalence \eqref{5.12} between the half-line matrix Schr\"odinger operator
$H_{A,B,\mathbf V}$ and the full-line matrix Schr\"odinger operator
$H_V,$ can be made plausible as follows.
As seen from \eqref{5.10}, the $2n\times 2n$ matrix potential $\mathbf V(x)$ is a block diagonal matrix
consisting of the
two $n\times n$ square matrices $V_2(x)$ and $V_1(-x)$ for $x\in\mathbb R^+.$
Thus, we can decompose any column-vector solution $\phi(x)$ to \eqref{5.49} with $2n$ components
into two column vectors each having $n$ components as
\begin{equation}
\label{5.57}
\phi(x)=\begin{bmatrix} \phi_2(x)\\
\noalign{\medskip}
\phi_1(x)\end{bmatrix},\qquad x\in\mathbb R^+.
\end{equation}
Consequently, we can decompose the $2n\times 2n$ matrix system \eqref{5.49} into the two $n\times n$
matrix subsystems
\begin{equation}
\label{5.58}
\begin{cases}
-\phi_2''+V_2(x)\,\phi_2=0,\qquad x\in\mathbb R^+,\\
\noalign{\medskip}
-\phi_1''+V_1(-x)\,\phi_1=0,\qquad x\in\mathbb R^+,
\end{cases}
\end{equation}
in such a way that the two $n\times n$ subsystems given in the first and second lines, respectively, 
in \eqref{5.58} are uncoupled from each other.
We now look for $n\times n$ matrix solutions to each of the two $n\times n$ subsystems in \eqref{5.58}.
Let us first solve the first subsystem in \eqref{5.58}. From (3.2.157) and Proposition~3.2.6
of \cite{AW2021}, we know that there are $n$ linearly independent bounded solutions to the first
subsystem in \eqref{5.58}. With the help of \eqref{5.57} and
the boundary matrices $A$ and $B$ in \eqref{5.16}, we see that the boundary condition
\eqref{5.4} is expressed as two $n\times n$ systems as
\begin{equation}\label{5.59}
\phi_1(0)=\phi_2(0),\quad \phi'_1(0)=-\phi'_2(0).
\end{equation}
Thus, once we have $\phi_2(x)$ at hand, the initial values for $\phi_1(0)$ and $\phi'_1(0)$ are uniquely determined 
from \eqref{5.59}. Then, we solve the second subsystem
in \eqref{5.58} with the already determined initial conditions given in \eqref{5.59}.
Since the first subsystem in \eqref{5.58} can have at most $n$ linearly independent bounded solutions $\phi_2(x)$ and
we have $\phi_1(x)$ uniquely determined using $\phi_2(x),$ from the decomposition \eqref{5.57} we conclude that we
can have at most $n$ linearly independent bounded column-vector solutions $\phi(x)$ to the $2n\times 2n$ system \eqref{5.49}.

In \eqref{3.14} we have expressed the determinant of the scattering matrix $S(k)$ for the full-line Schr\"odinger equation
\eqref{1.1} in terms of the determinant of the left transmission coefficient $T_{\rm{l}}(k).$
When the full-line Hamiltonian $H_V$ and the half-line Hamiltonian $H_{A,B,\mathbf V}$ are related to each other unitarily as in 
\eqref{5.12}, in the next theorem we relate the determinant of
the corresponding half-line scattering matrix $\mathbf S(k)$ to $\det[S(k)].$

\begin{theorem}
\label{theorem5.3} 
Consider the full-line matrix Schr\"odinger equation \eqref{1.1} with the $n\times n$ matrix
potential $V$ satisfying \eqref{1.2} and \eqref{1.3}. Assume that the corresponding full-line
Hamiltonian $H_V$ and the half-line Hamiltonian
$H_{A,B,\mathbf V}$ are unitarily connected as in \eqref{5.12} by relating the half-line $2n\times 2n$ matrix potential $\mathbf V$ to
$V$ as in \eqref{1.4} and \eqref{5.10} and
by choosing the boundary matrices $A$ and $B$ as in \eqref{5.16}.
Then, the determinant of the half-line scattering matrix $\mathbf S(k)$ defined in \eqref{5.20} is related to the determinant
of the full-line scattering matrix $S(k)$ defined in \eqref{2.5} as
\begin{equation}
\label{5.60}
\det\left[\mathbf S(k)\right]=  (-1)^n \det\left[S(k)\right], \qquad k\in\mathbb R,
\end{equation}
and hence we have
\begin{equation}
\label{5.61}
\det\left[\mathbf S(k)\right]=(-1)^{n}\displaystyle\frac{\det[T_{\rm{l}}(k)]}{\left(\det[T_{\rm{l}}(k)]\right)^*},\qquad k\in\mathbb R,
\end{equation}
where we recall that $T_{\rm{l}}(k)$ is the $n\times n$ matrix-valued left transmission coefficient appearing in \eqref{2.3}.

\end{theorem}

\begin{proof}
The scattering matrices $\mathbf S(k)$ and $S(k)$ are related to each other as  in \eqref{5.27}, where $Q$ is the 
$2n\times 2n$ constant matrix defined in \eqref{2.51}. By interchanging the first and second row blocks in
$Q$ we get the $2n\times 2n$ identity matrix, and hence we have $\det[Q]=(-1)^n.$ Thus, by taking the determinant
of both sides of \eqref{5.27}, we obtain \eqref{5.60}. Then, using \eqref{3.14} on the right-hand side of \eqref{5.60}, we get \eqref{5.61}. 
\end{proof}

\section{Levinson's theorem}
\label{section6}

The main goal in this section is to prove Levinson's theorem for the full-line matrix Schr\"odinger equation \eqref{1.1}. In the scalar case, we recall that Levinson's theorem connects the continuous spectrum and the discrete spectrum to each other,
and it relates the number of discrete eigenvalues including the multiplicities to the
change in the phase of the scattering matrix when the independent variable
$k$ changes from $k=0$ to $k=+\infty.$ In the matrix case, in Levinson's theorem one needs to use the phase of the
determinant of the scattering matrix. In order to prove Levinson's theorem for \eqref{1.1}, we 
exploit the unitary transformation given in \eqref{5.12} relating the half-line and full-line
matrix Schr\"odinger operators and we use Levinson's theorem presented in Theorem~3.12.3 of
\cite{AW2021} for the half-line matrix Schr\"odinger operator.
We also provide an independent proof with the help of the argument principle
applied to the determinant of the matrix-valued left transmission coefficient appearing in \eqref{2.3}.

In preparation for the proof of Levinson's theorem for \eqref{1.1}, in the next theorem
we relate the large $k$-asymptotics of the half-line scattering matrix $\mathbf S(k)$
to the full-line matrix potential $V$ in \eqref{1.1}. Recall that the half-line potential $\mathbf V$ is related
to the full-line potential $V$ as in \eqref{5.10}, where $V_1$ and $V_2$ are the fragments of $V$
appearing in \eqref{1.4} and \eqref{1.5}. We also recall that the boundary matrices
$A$ and $B$ appearing in \eqref{5.16} are used in the boundary condition \eqref{5.4} in the formation
of the half-line scattering matrix $\mathbf S(k).$

\begin{theorem}
\label{theorem6.1} 
Consider the full-line matrix Schr\"odinger equation \eqref{1.1} with the $n\times n$ matrix
potential $V$ satisfying \eqref{1.2} and \eqref{1.3}. Assume that the corresponding full-line
Hamiltonian $H_V$ and the half-line Hamiltonian
$H_{A,B,\mathbf V}$ are unitarily connected by relating the half-line $2n\times 2n$ matrix potential $\mathbf V$ to
$V$ as in \eqref{1.4} and \eqref{5.10} and
by choosing the $2n\times 2n$ boundary matrices $A$ and $B$ as in \eqref{5.16}.
Then, the half-line scattering matrix $\mathbf S(k)$ in \eqref{5.20} has the large $k$-asymptotics
given by
\begin{equation}
\label{6.1}
\mathbf S(k)=\mathbf S_\infty+\displaystyle\frac{\mathbf G_1}{ik}+o\left(\displaystyle\frac{1}{k}\right),\qquad k\to\pm\infty
\text{ \rm{in} }\mathbb R,
\end{equation}
where we have
\begin{equation}
\label{6.2}
\mathbf S_\infty=Q,\quad \mathbf G_1=\displaystyle\frac{1}{2}\begin{bmatrix}0&\int_{-\infty}^\infty dx\,V(x)\\
\noalign{\medskip}
\int_{-\infty}^\infty dx\,V(x)
&0
\end{bmatrix},
\end{equation}
with $Q$ being the constant $2n\times 2n$ matrix defined in \eqref{2.51}.

\end{theorem}

\begin{proof}
We already know that the half-line scattering matrix $\mathbf S(k)$ is related to the full-line
scattering coefficients as in \eqref{5.28}. With the help of \eqref{4.2}--\eqref{4.5} expressing
the large $k$-asymptotics of
the scattering coefficients, from \eqref{5.28} we obtain \eqref{6.1}. 
\end{proof}

\begin{remark}\label{remark6.2}
\normalfont
We have the following comments related to the first equality in \eqref{6.2}.
One can directly evaluate the eigenvalues of the
matrix $Q$ defined in \eqref{2.51} and confirm that it has
the eigenvalue $-1$ with multiplicity $n$ and 
has
the eigenvalue $+1$ with multiplicity $n.$ By Theorem~3.10.6 of \cite{AW2021} we know that the number of 
Dirichlet boundary conditions associated with the boundary matrices $\tilde{A}$ and $\tilde{B}$ in \eqref{5.16} 
is equal to the (algebraic and geometric) multiplicity of the eigenvalue $-1$ of the matrix $\mathbf S_\infty,$  and that the number of 
non-Dirichlet boundary conditions is equal to the (algebraic and geometric) multiplicity of the eigenvalue $+1$ of 
the matrix $\mathbf S_\infty.$ Hence, the first equality in \eqref{6.2} implies that the boundary matrices $\tilde{A}$ and 
$\tilde{B}$ in \eqref{5.16}  correspond to $n$ Dirichlet and $n$ non-Dirichlet boundary conditions. In fact, 
it is already proved in Proposition~5.9 of \cite{W2022} with a different method that the number of
Dirichlet boundary conditions equal to $n,$ and that
the aforementioned $n$ non-Dirichlet boundary conditions
are actually all Neumann boundary conditions.
\end{remark}

Next, we provide a review of some relevant
facts on the small $k$-asymptotics related
to the full-line matrix Schr\"odinger equation \eqref{1.1} with the $n\times n$ matrix
potential $V$ satisfying \eqref{1.2} and \eqref{1.3}. 
We refer the reader to \cite{AKV2001}
for an elaborate analysis including the proofs and further details.

Let us consider the full-line $n\times n$ matrix-valued zero-energy Schr\"odinger equation given in \eqref{5.52}
when the $n\times n$ matrix potential $V$ satisfies \eqref{1.2} and \eqref{1.3}.
There are two $n\times n$ matrix-valued solutions to \eqref{5.52} whose $2n$ columns
form a fundamental set. We use $\phi_{\rm{l}}(x)$ to denote one of those two $n\times n$ matrix-valued solutions, where
$\phi_{\rm{l}}(x)$ satisfies the asymptotic conditions
  \begin{equation*}
\phi_{\rm{l}}(x)=x\left[I+o(1)\right],\quad \phi'_{\rm{l}}(x)=I+o(1),\qquad x\to+\infty.
\end{equation*}
Thus, all the $n$ columns of $\phi_{\rm{l}}(x)$ correspond to unbounded
solutions to \eqref{5.52}. The other $n\times n$ matrix-valued solution to \eqref{5.52} is given by
$f_{\rm{l}}(0,x),$ where $f_{\rm{l}}(k,x)$ is the left Jost solution to \eqref{1.1} appearing in \eqref{2.1}.
As seen from \eqref{2.1}, the function $f_{\rm{l}}(0,x)$ satisfies the asymptotics 
 \begin{equation*}
f_{\rm{l}}(0,x)=I+o(1),\quad f'_{\rm{l}}(0,x)=o(1),\qquad x\to+\infty,
\end{equation*}
and hence $f_{\rm{l}}(0,x)$ remains bounded as $x\to+\infty.$ On the other hand, some or all the $n$ columns
of $f_{\rm{l}}(0,x)$ may be unbounded as $x\to-\infty.$ In fact, we have \cite{AKV2001}
\begin{equation*}
f_{\rm{l}}(0,x)=x\left[\Delta_{\rm{l}}+o(1)\right],\quad f'_{\rm{l}}(0,x)=\Delta_{\rm{l}}+o(1),\qquad x\to-\infty,
\end{equation*}
where $\Delta_{\rm{l}}$ is the $n\times n$ constant matrix defined as
\begin{equation}
\label{6.6}
\Delta_{\rm{l}}:=\displaystyle\lim_{k\to 0} 2ik\,T_{\rm{l}}(k)^{-1},
\end{equation}
with the limit taken from within $\overline{\mathbb C^+}.$

From Section~\ref{section5}, we recall that the degree of exceptionality, denoted by $\nu,$ is defined as
the number of linearly independent bounded column-vector solutions to \eqref{5.52}, and we know that $\nu$ satisfies
\eqref{5.53}. From \eqref{5.53} we see that we have the purely generic case for \eqref{1.1} when $\nu=0$
and we have the purely exceptional case when $\nu=n.$ 
Hence, $n-\nu$ corresponds to 
the {\it degree of genericity} for \eqref{1.1}. Since the value of 
$\nu$
is uniquely determined by only the potential $V$ in \eqref{1.1}, we can also say that the $n\times n$ matrix
potential $V$ is exceptional with degree $\nu.$

We can characterize the value of $\nu$ in various different ways.
For example, $\nu$ corresponds to the geometric multiplicity of the zero eigenvalue of the $n\times n$ matrix $\Delta_{\rm{l}}$
defined in \eqref{6.6}. The value of $\nu$ is equal to
the nullity of the matrix $\Delta_{\rm{l}}.$ 
It is also equal to the number of linearly independent bounded columns of the
zero-energy left Jost solution $f_{\rm{l}}(0,x).$ Thus, the remaining $n-\nu$ columns of
$f_{\rm{l}}(0,x)$ are all unbounded solutions to \eqref{5.52}. Hence, \eqref{5.52}
has $2n-\nu$ linearly independent unbounded column-vector solutions and it has $\nu$ linearly independent bounded
column-vector solutions.

The degree of exceptionality $\nu$ can also be related to the zero-energy right Jost solution $f_{\rm{r}}(0,x)$ to \eqref{5.52}.
As we see from \eqref{2.2}, the function $f_{\rm{r}}(0,x)$ satisfies the asymptotics 
 \begin{equation}
\label{6.7}
f_{\rm{r}}(0,x)=I+o(1),\quad f'_{\rm{r}}(0,x)=o(1),\qquad x\to-\infty,
\end{equation}
and hence $f_{\rm{r}}(0,x)$ remains bounded as $x\to-\infty.$ 
On the other hand, we have 
\begin{equation}
\label{6.8}
f_{\rm{r}}(0,x)=-x\left[\Delta_{\rm{r}}+o(1)\right],\quad f'_{\rm{r}}(0,x)=-\Delta_{\rm{r}}+o(1),\qquad x\to+\infty,
\end{equation}
where $\Delta_{\rm{r}}$ is the constant $n\times n$ matrix defined as
\begin{equation}
\label{6.9}
\Delta_{\rm{r}}:=\displaystyle\lim_{k\to 0} 2ik\,T_{\rm{r}}(k)^{-1},
\end{equation}
with the limit taken from within $\overline{\mathbb C^+}.$
 From \eqref{4.1} it follows that the $n\times n$ matrices
 $\Delta_{\rm{l}}$ and $\Delta_{\rm{r}}$ satisfy
 \begin{equation}
\label{6.10}
\Delta_{\rm{r}}=\Delta_{\rm{l}}^\dagger.
\end{equation}
With the help of \eqref{6.7}--\eqref{6.10} we observe that
the degree of exceptionality $\nu$ is equal to the 
geometric multiplicity of the zero eigenvalue of the $n\times n$ matrix $\Delta_{\rm{r}},$ to
the nullity of $\Delta_{\rm{r}},$ and also to
the number of linearly independent bounded columns of $f_{\rm{r}}(0,x).$ 

One may ask whether the degree of exceptionality for the full-line 
potential $V$ can be determined from the
degrees of exceptionality for its fragments. The answer is no unless
all the fragments are purely exceptional, in which case the full-line potential
must also be exceptional. This answer can be obtained by first considering the scalar case
and then by generalizing it to the matrix case by arguing with a diagonal matrix potential. 
We refer the reader to \cite{A1994} for an explicit example in the scalar case,
where it is demonstrated that the potential $V$ may be generic or exceptional if 
at least one of the fragments is generic. In the next example, we also illustrate
this fact by a different example.

\begin{example}
\normalfont
\label{example6.2a}
We recall that,  in the full-line scalar case, if
the potential $V$ in \eqref{1.1} is real valued and satisfies \eqref{1.3} then generically
 the transmission coefficient $T(k)$ vanishes linearly as $k\to 0$ and
 we have $T(0)\ne 0$ in the exceptional case.  In the generic case we
have $L(0)=R(0)=-1$ and in the exceptional case we have $|L(0)|=|R(0)|<1,$
where we use $L(k)$ and $R(k)$ for the left and right reflection coefficients.
Hence, from \eqref{3.53} we observe that, in the scalar case, if the two fragments of a full-line potential are both exceptional
then the potential itself must be exceptional. 
By using induction, the result can be proved to hold also in the case where the number of
fragments is arbitrary.
In this example, we demonstrate that if at least one of the fragments 
is generic, then the potential itself can be exceptional or generic.
In the full-line scalar case, let us consider the square-well potential of width $a$ and depth $b,$ where
$a$ is a positive parameter and $b$ is a nonnegative parameter. Without any loss of generality,
because the transmission coefficient is not affected by shifting the potential, we can assume that the potential $V$ is 
given by
\begin{equation}
\label{6.a1}
V(x)=
\begin{cases}
-b,\qquad -\displaystyle\frac{a}{2}<x<\displaystyle\frac{a}{2},\\
\noalign{\medskip}
0,\qquad \text{\rm{elsewhere}}.
\end{cases}
\end{equation}
The transmission coefficient corresponding to the scalar potential $V$ in \eqref{6.a1} is given by
\begin{equation}
\label{6.a2}
T(k)=\displaystyle\frac{4k\sqrt{b+k^2}\,e^{-ika}}
{q_{14}+q_{15}+q_{16}},
\end{equation}
where we have defined
\begin{equation*}
q_{14}:=-b-k^2+k \sqrt{b+k^2},
\end{equation*}
\begin{equation*}
q_{15}:=\left[-k^2+k \sqrt{b+k^2}\right] \exp\left(ia \sqrt{b+k^2}\right),
\end{equation*}
\begin{equation*}
q_{16}:=\left[b+2 k^2+2k \sqrt{b+k^2}\right] \exp\left(-ia \sqrt{b+k^2}\right).
\end{equation*}
From \eqref{6.a2} we obtain
\begin{equation}
\label{6.a6}
\displaystyle\frac{1}{T(k)}=\displaystyle\frac{\sqrt{b}}{4k}\left(e^{-ia \sqrt{b}}-1\right)+O(1),\qquad k\to 0.
\end{equation}
Hence, \eqref{6.a6} implies that the exceptional case occurs if and only if $a\sqrt{b}$ is an integer multiple of $2\pi,$
in which case $T(k)$ does not vanish at $k=0.$ Thus, the potential $V$ in \eqref{6.a1} is exceptional if and only if
we have the depth $b$ of the square-well potential is related to the width $a$ as
\begin{equation}
\label{6.a7}
b=\displaystyle\frac{4 p^2 \pi^2}{a^2},
\end{equation}
for some nonnegative integer $p.$
For simplicity, let us use $a=2$ and choose our potential fragments $V_1$ and $V_2$ as
\begin{equation*}
V_1(x)=
\begin{cases}
-b,\qquad -1<x<c,\\
\noalign{\medskip}
0,\qquad \text{\rm{elsewhere}},
\end{cases}
\quad
V_2(x)=
\begin{cases}
-b,\qquad c<x<1,\\
\noalign{\medskip}
0,\qquad \text{\rm{elsewhere}},
\end{cases}
\end{equation*}
where $c$ is a parameter in $[-1,1]$ so that the potential $V$ in \eqref{6.a1} is given by
\begin{equation}
\label{6.a9}
V(x)=
\begin{cases}
-b,\qquad -1<x<1,\\
\noalign{\medskip}
0,\qquad \text{\rm{elsewhere}}.
\end{cases}
\end{equation}
We remark that the zero potential is exceptional.
With the help of \eqref{6.a7}, we conclude that $V_1$ is exceptional if and only if
 there exists a nonnegative integer
$p_1$ 
satisfying
\begin{equation}
\label{6.a10}
\displaystyle\frac{\sqrt{b}}{2\pi}=\displaystyle\frac{p_1}{1+c},
\end{equation}
 that $V_2$ is exceptional if and only if there exists a nonnegative integer
$p_2$ 
satisfying
\begin{equation}
\label{6.a11}
\displaystyle\frac{\sqrt{b}}{2\pi}=\displaystyle\frac{p_2}{1-c},
\end{equation}
and that
$V$ is exceptional if and only if there exists a nonnegative integer
$p$ 
satisfying
\begin{equation}
\label{6.a12}
\displaystyle\frac{\sqrt{b}}{2\pi}=\displaystyle\frac{p}{2}.
\end{equation}
In fact, \eqref{6.a12} happens if and only if $p=p_1+p_2.$
From \eqref{6.a10}--\eqref{6.a12} we have the following conclusions:

\begin{enumerate}
\item[\text{\rm(a)}] If $V_1$ and $V_2$ are both exceptional, then $V$ is also exceptional. This follows
from the restriction $p=p_1+p_2$ that if $p_1$ and $p_2$ are both nonnegative
integers then $p$ must also be a nonnegative integer. As argued earlier, if $V_1$ and $V_2$ 
are both exceptional then $V$ cannot be
generic.

\item[\text{\rm(b}] If the nonnegative integer $p_1$ satisfying \eqref{6.a10} exists but the 
nonnegative integer $p_2$ satisfying \eqref{6.a11} does not exist, then the
nonnegative integer $p$ satisfying \eqref{6.a12} cannot exist because of the restriction $p=p_1+p_2.$
Thus, we can conclude that if $V_1$ is exceptional then $V_2$ and $V$ are either
simultaneously exceptional or simultaneously generic.
Similarly, we can conclude that if $V_2$ is exceptional then $V_1$ and $V$ are either
simultaneously exceptional or simultaneously generic.

\item[\text{\rm(c)}] If both $V_1$ are $V_2$ are generic, then $V$ can be generic or exceptional.
This can be seen easily by arguing that in the equation $p=p_1+p_2,$ it may happen
that $p_1+p_2$ is a nonnegative integer or not a nonnegative integer even though neither $p_1$ nor $p_2$ are nonnegative integers.

\end{enumerate}
\end{example}

Finally in this section, we consider Levinson's theorem for \eqref{1.1}. 
We refer the reader to Theorem~3.12.3 of \cite{AW2021} for Levinson's theorem
for the half-line matrix Schr\"odinger operator with the selfadjoint boundary condition.
With the help of that theorem, we have the following result for the half-line matrix Schr\"odinger operator
associated with \eqref{5.1} and \eqref{5.4}. As mentioned in Remark~\ref{remark6.2}, the result given
in \eqref{6.12} in the next theorem has been proved in \cite{W2022} by using a different method.

\begin{theorem}
\label{theorem6.3} 
Consider the half-line matrix Schr\"odinger operator corresponding to \eqref{5.1}
with the $2n\times 2n$ matrix potential $\mathbf V$ satisfying
\eqref{5.2} and \eqref{5.3}, and with the selfadjoint boundary condition
\eqref{5.4} with the boundary matrices $A$ and $B$ 
satisfying \eqref{5.5}. Let $\mathcal N_{\rm{h}}$ denote the number of
eigenvalues including the multiplicities, $\mathbf S(k)$ be the corresponding scattering matrix
defined as in \eqref{5.20}, and
$\mathbf S_\infty$ be the
constant matrix appearing in \eqref{6.1} and \eqref{6.2}. Then, we have the following:

\begin{enumerate}

\item[\text{\rm(a)}] 
The nonnegative integer $\mathcal N_{\rm{h}}$ is related to the argument of the determinant of $\mathbf S(k)$ as
\begin{equation}
\label{6.11}
\arg[\det[\mathbf S(0^+)]]-\arg[\det[\mathbf S_{\infty}]]=\pi\left[
2 \,\mathcal N_{\rm{h}}+\mu-2\,n+n_{\text{\rm{D}}}\right],
\end{equation}
where $\arg$ is any single-valued branch of the argument
function, the nonnegative integer $\mu$ is the algebraic and geometric multiplicity of the eigenvalue $+1$ of
$\mathbf S(0),$ and the nonnegative integer $n_{\text{\rm{D}}}$ is the algebraic and geometric multiplicity of the eigenvalue $-1$ of
$\mathbf S_{\infty}.$

\item[\text{\rm(b)}] Assume further that 
the boundary matrices $A$ and $B$ appearing in \eqref{5.4} are chosen
as in \eqref{5.16}.
Then, we have 
\begin{equation}
\label{6.12}
n_{\text{\rm{D}}}=n,
\end{equation}
and hence in that special case, \eqref{6.11} yields
\begin{equation}
\label{6.13}
\arg[\det[\mathbf S(0^+)]]-\arg[\det[\mathbf S_{\infty}]]=\pi\left[
2\, \mathcal N_{\rm{h}}+\mu-n\right].
\end{equation}

\end{enumerate}
\end{theorem}

\begin{proof}
We remark that \eqref{6.11} directly follows from (3.12.15) of \cite{AW2021} by using
the fact that the matrix potential $\mathbf V$ has size $2n\times 2n.$ Hence, the proof of (a) is complete.
In order to prove (b), it is sufficient to prove \eqref{6.12}. We can explicitly evaluate the large $k$-asymptotics
of $\mathbf S(k)$ when $A$ and $B$ are chosen as in \eqref{5.16}. From (3.10.37) of \cite{AW2021}
we know that $\mathbf S_{\infty}$ is unchanged if the potential $\mathbf V$ is zero, which is not surprising because 
the potential cannot affect the leading term in the large $k$-asymptotics of the scattering matrix. Thus, 
$\mathbf S_{\infty}$ and its eigenvalues are solely determined by the boundary matrices $A$ and $B.$
When the potential is zero, as seen from (3.7.2) of \cite{AW2021} we have
\begin{equation}
\label{6.14}
\mathbf S(k)=-(B+ik A)(B-ikA)^{-1}.
\end{equation}
When we use $A$ and $B$ given in \eqref{5.16}, we get
\begin{equation}
\label{6.15}
B-ikA=\begin{bmatrix} -I&-ikI\\
\noalign{\medskip}
I&-ik I\end{bmatrix}.
\end{equation}
Using \eqref{6.15} in
\eqref{6.14}
we evaluate the large $k$-asymptotics of $\mathbf S(k)$ given in \eqref{6.14} as
\begin{equation*}
\mathbf S(k)=-\begin{bmatrix}-I&ik I\\
\noalign{\medskip}
I&ik I\end{bmatrix}\left(
-\displaystyle\frac{1}{2}\,
\begin{bmatrix}I&-I\\
\noalign{\medskip}
\displaystyle\frac{I}{ik}&\displaystyle\frac{I}{ik}\end{bmatrix}\right),
\end{equation*}
which yields the exact result
\begin{equation}
\label{6.17}
\mathbf S(k)=Q,
\end{equation}
where we recall that $Q$ is the $2n\times 2n$ constant matrix defined in \eqref{2.51}.
Thus, we have $\mathbf S_{\infty}=Q.$ 
As indicated in Remark~\ref{remark6.2}, the matrix $Q$ has eigenvalue $-1$
with multiplicity $n.$
Hence, \eqref{6.12} holds, and the proof of (b) is complete.
\end{proof}

We remark that \eqref{6.17} also follows from the first equality in \eqref{6.2} of
Theorem~\ref{theorem6.1} by using the unitary transformation \eqref{5.12} between the half-line and full-line Schr\"odinger operators.
However, we have established Theorem~\ref{theorem6.3}(b) without using that unitary transformation
and without making any connection to the full-line Schr\"odinger equation \eqref{1.1}.

Using \eqref{4.2}--\eqref{4.5} in \eqref{2.5}, we see that the full-line
scattering matrix $S(k)$ satisfies
\begin{equation*}
S_{\infty}=\mathbf I,
\end{equation*}
where we have defined
\begin{equation*}
S_{\infty}:=\displaystyle\lim_{k\to\pm\infty} S(k).
\end{equation*}

In the next theorem we state and prove Levinson's theorem for
the full-line matrix-valued Schr\"odinger equation \eqref{1.1}. As in Theorem~\ref{theorem6.3}, 
we again use $\arg$ to denote any single-valued branch of the argument function.

\begin{theorem}
\label{theorem6.4} 
Consider the full-line matrix Schr\"odinger operator corresponding to \eqref{1.1} with the $n\times n$ matrix potential
$V$ satisfying \eqref{1.2} and \eqref{1.3}.
Let $S(k)$ be the corresponding $2n\times 2n$ scattering matrix defined in \eqref{2.5}.
Then, the corresponding number $\mathcal N$ of eigenvalues including the multiplicities,
which appears in \eqref{5.47}, is related
to the argument of the determinant of $S(k)$ as
\begin{equation}
\label{6.18}
\arg[\det[S(0^+)]]-\arg[\det[S_{\infty}]]=\pi\left[
2\mathcal N-n+\nu\right],
\end{equation}
where we recall that the nonnegative integer $\nu$ is the degree of exceptionality appearing in \eqref{5.53} and
denoting the number of linearly independent bounded solutions to \eqref{5.52}.
\end{theorem}

\begin{proof}
Even though we can prove \eqref{6.18} independently without using any connection to the half-line $2n\times 2n$ matrix
Schr\"odinger equation \eqref{5.1}, it is illuminating to prove it by exploiting the unitary
connection \eqref{5.12} between the full-line Hamiltonian $H_V$ and
the half-line Hamiltonian $H_{A,B,\mathbf V},$ where 
the half-line $2n\times 2n$ matrix potential $\mathbf V$ is related to
$V$ as in \eqref{1.4} and \eqref{5.10} and the boundary matrices
$A$ and $B$ chosen as in \eqref{5.16}. Then, from \eqref{5.60} it follows that
the left-hand side of \eqref{6.18} is equal to the left-hand side of
\eqref{6.13}, i.e. we have
\begin{equation}
\label{6.19}
\arg[\det[S(0^+)]]-\arg[\det[S_{\infty}]]=\arg[\det[\mathbf S(0^+)]]-\arg[\det[\mathbf S_{\infty}]],
\end{equation}
where $\mathbf S(k)$ is the $2n\times 2n$ scattering matrix corresponding to
$H_{A,B,\mathbf V}.$ Furthermore, because of \eqref{5.12} we know that both
the number of eigenvalues and their multiplicities for $H_V$ and $H_{A,B,\mathbf V}$
coincide, and hence we have
\begin{equation}
\label{6.20}
\mathcal N_{\rm{h}}=\mathcal N,
\end{equation}
where we recall that $\mathcal N_{\rm{h}}$ is the number of eigenvalues
of $H_{A,B,\mathbf V}$ including the multiplicities. We also know that \eqref{5.54} holds,
where $\mu$ is the geometric and algebraic multiplicity of
the eigenvalue $+1$ of $\mathbf S(0).$
Thus, using \eqref{5.54} and \eqref{6.20} on the right-hand side of \eqref{6.13},
with the help of \eqref{6.19} we obtain \eqref{6.18}. Hence, the proof is complete.
\end{proof}

An independent proof of \eqref{6.18} without using the unitary connection \eqref{5.12} can be given
by applying the argument principle to the determinant of
$T_{\rm{l}}(k).$ For this, we can proceed as follows.
As indicated in Section~\ref{section5}, the Hamiltonian $H_V$ has
$N$ distinct eigenvalues $-\kappa_j^2,$ each with multiplicity $m_j$ 
for $1\le j\le N.$ In case $N=0,$ there are no eigenvalues.
The quantity $\det[T_{\rm{l}}(k)]$ has a meromorphic extension from
$k\in\mathbb R$ to $\mathbb C^+$ such that the only poles there
occur at $k=i\kappa_j$ with multiplicity $m_j$ for $1\le j\le N.$ Furthermore,
except for those poles,
 $\det[T_{\rm{l}}(k)]$ is continuous in $\overline{\mathbb C^+}$
 and nonzero in $\overline{\mathbb C^+}\setminus\{0\},$ and it has a zero at $k=0$
with order $n-\nu,$ as indicated in \eqref{5.55}. We note that $\det[T_{\rm{l}}(k)]$ is nonzero
at $k=0$ when $\nu=n.$ 
To apply the argument principle, we choose our contour $\mathcal C_{\varepsilon, r}$ as the positively oriented
closed curve given by
\begin{equation}
\label{6.21}
\mathcal C_{\varepsilon, r}:= (-r,-\varepsilon)\cup  \mathcal C_\varepsilon\cup(\varepsilon, r)\cup \mathcal C_r.
\end{equation}
Note that the first piece $(-r, -\varepsilon)$ on the right-hand side of 
\eqref{6.21} is the directed line segment on the real axis oriented from $-r$ to $-\varepsilon$ for some small positive $\varepsilon$ and some large positive $r$. The second piece $\mathcal C_\varepsilon$  is the upper semicircle centered at the origin with radius $\varepsilon$ and oriented from $-\varepsilon$ to $\varepsilon.$ The third piece is the real line segment $(\varepsilon,r)$ oriented from $\varepsilon$ to $r.$ Finally the fourth piece $\mathcal C_r$ is the upper semicircle centered at zero and with radius $r$ and oriented from $r$ to $-r.$
From \eqref{4.2} we see that the argument of $\det[T_{\rm{l}}(k)]$ does not change along the piece $C_r$ when
$r\to+\infty.$
In the limit $\varepsilon\to 0^+$ and $r\to+\infty,$ the application of the argument principle to $\det[T_{\rm{l}}(k)]$
along the contour $C_{\varepsilon,r}$ yields
 \begin{equation}
\label{6.22}
\begin{split}
\arg[\det[T_{\rm{l}}(+\infty)]]-&\arg[\det[T_{\rm{l}}(0^+)]]+\arg[\det[T_{\rm{l}}(0^-)]]\\
&
-\arg[\det[T_{\rm{l}}(-\infty)]]
=2\pi\left[\displaystyle\frac{n-\nu}{2}-\mathcal N\right],\end{split}
\end{equation}
where the factor $1/2$ in the brackets on the right-hand side is due to the fact that we integrate
along the semicircle $C_\varepsilon.$
From the first equality in \eqref{3.13}, we conclude that
\begin{equation}
\label{6.23}
\arg[\det[T_{\rm{l}}(0^-)]]-\arg[\det[T_{\rm{l}}(-\infty)]]=
\arg[\det[T_{\rm{l}}(+\infty)]]-\arg[\det[T_{\rm{l}}(0^+)]].
\end{equation}
Using \eqref{6.23} in \eqref{6.22}, we obtain
 \begin{equation}
\label{6.24}
\arg[\det[T_{\rm{l}}(+\infty)]]-\arg[\det[T_{\rm{l}}(0^+)]]=\pi\left[\displaystyle\frac{n-\nu}{2}-\mathcal N\right].
\end{equation}
Finally, using \eqref{3.14} in \eqref{6.24} we get
\begin{equation*}
\arg[\det[S_{\infty}]]-\arg[\det[S(0^+)]]=2\pi\left[\displaystyle\frac{n-\nu}{2}-\mathcal N\right],
\end{equation*}
which is equivalent to \eqref{6.18}.


\newpage



\begin{thebibliography}{16}

\bibitem{AM1963} Z. S. Agranovich and V. A. Marchenko, \textit{The inverse problem of
scattering theory,} Gordon and Breach, New York, 1963.

\bibitem{A1992} T. Aktosun, \textit{A factorization of the scattering matrix for the Schr\"odinger equation and for the wave equation in one dimension,} J. Math. Phys. {\bf 33}, 3865--3869 (1992).

\bibitem{A1994} T. Aktosun, \textit{Bound states and inverse scattering for the Schr\"odinger equation 
in one dimension,} J. Math. Phys. {\bf 35}, 6231--6236 (1994).


\bibitem{AKV1996} T. Aktosun, M. Klaus, and  C. van der Mee,
\textit{Factorization of the scattering matrix due to partitioning of potentials in one-dimensional Schr\"odinger type equations,} J. Math. Phys. {\bf 37}, 5897--5915 (1996). 

\bibitem{AKV2001}
T. Aktosun, M. Klaus, and  C. van der Mee, \textit{Small-energy asymptotics of the scattering matrix 
for the matrix Schr\"odinger equation on the line,} J. Math. Phys. {\bf 42}, 4627--4652 (2001).

\bibitem{AW2021}
T. Aktosun and R. Weder, \textit{Direct and inverse scattering for the matrix Schr\"odinger equation,} Springer Nature, Switzerland, 2021.

\bibitem{CL1955} E. A. Coddington and N. Levinson, \textit{Theory of ordinary differential equations,} McGraw Hill, New York, 1955.

\bibitem{D2006}
H. Dym, \textit{Linear algebra in action,} Am. Math. Soc. Publ., Providence, RI, 2006.

\bibitem{KS1999} V. Kostrykin and R. Schrader, \textit{Kirchhoff's rule for quantum wires,} J. Phys. A {\bf 32},
596--630 (1999).

\bibitem{LW2000} A. Laptev and T. Weidl, \textit{Sharp Lieb-Thirring inequalities in high dimensions,} Acta Math. {\bf 184}, 87--111 (2000). 

\bibitem{MO1982} I. Mart\'inez Alonso and E. Olmedilla, \textit{Trace identities in the inverse scattering transform method associated with matrix Schr\"odinger  operators,} J. Math. Phys. {\bf 23}, 2116--2121 (1982).

\bibitem{O1985} E. Olmedilla, \textit{Inverse scattering transform for general matrix Schr\"odinger operators and the related sympletic structure,} Inverse Probl. {\bf 1}, 219--236 (1985).

\bibitem{W2022} R. Weder, \textit{The $L^p$ boundedness of the wave operators for matrix Schr\"odinger equations,} J. Spectr. Theory
{\bf 12}, 707--744 (2022).

\end{thebibliography}
\end{document}